\documentclass[opre,nonblindrev]{informs3_arxiv}
\usepackage[margin=1in]{geometry}

\SingleSpacedXI

\usepackage{natbib}
\bibpunct[, ]{(}{)}{,}{a}{}{,}%
\def\bibfont{\small}%
\def\bibsep{\smallskipamount}%
\usepackage{booktabs}
\usepackage{tablefootnote}
\usepackage{graphicx}
\usepackage{fancyhdr}
\usepackage{multirow}
\usepackage{pdflscape}
\usepackage{subcaption}

\usepackage{amsmath}
\usepackage{amssymb}
\usepackage{mathtools}
\usepackage{bm}

\DeclareMathOperator{\Var}{Var}
\DeclareMathOperator{\Cov}{Cov}

\usepackage{tikz}
\usetikzlibrary{shapes,arrows,calc,plotmarks}

\usepackage{xcolor}
\usepackage{hyperref}
\usepackage[noabbrev,capitalise]{cleveref}

\crefname{appendix}{Appendix}{Appendices}
\Crefname{appendix}{Appendix}{Appendices}

\usepackage{thmtools}

\theoremstyle{plain}
\newtheorem{theorem}{Theorem}[section]
\newtheorem{corollary}[theorem]{Corollary}
\newtheorem{lemma}[theorem]{Lemma}
\newtheorem{proposition}[theorem]{Proposition}

\usepackage[linesnumbered,ruled,vlined]{algorithm2e}

\SetAlFnt{\small}
\SetAlCapFnt{\small}
\SetAlCapNameFnt{\small}
\SetAlCapHSkip{0pt}
\IncMargin{-\parindent}
\SetKwRepeat{Do}{do}{while}

\SetCommentSty{mycommfont}

\usepackage[normalem]{ulem}
\usepackage{footmisc}

\usepackage{color-edits}

\definecolor{green}{rgb}{1,0.5,0}

\addauthor[Amin]{mar}{blue}
\addauthor[Yuxin]{lyx}{red}
\addauthor[Space Save]{ss}{yellow}

\begin{document}


\RUNAUTHOR{Liu and Rahimian}

\RUNTITLE{Mechanism Design for Privacy-Preserving Information Sharing in Oligopoly Competition}

\TITLE{Mechanism Design for Privacy-Preserving Information Sharing in Oligopoly Competition}

\ARTICLEAUTHORS{%
\AUTHOR{Yuxin Liu$^{1}$, M. Amin Rahimian$^{1}$}
\AFF{$^{1}$Industrial Engineering, University of Pittsburgh\\
\EMAIL{yul435@pitt.edu, rahimian@pitt.edu}}
}


\ABSTRACT{

Information sharing among competing suppliers can improve decisions under demand uncertainty, but it may also intensify strategic interaction by aligning firms' beliefs. We study a Cournot oligopoly in which a platform designs an information-sharing mechanism using participation-contingent access, external platform information, and privacy-preserving noise. The central privacy-design challenge is that noise has two opposing effects: it limits how much a firm's report improves rivals' information, but it also reduces the value of the posterior signal released by the platform. In symmetric duopoly, privacy protection alone cannot implement sharing without an external platform signal. More generally, privacy can induce firms that would otherwise not share to participate only when combined with external platform information, which preserves an informational benefit independent of competitors' reports. The $n$-firm case adds a distinct force: under reciprocal access, non-participants lose access to the pooled signal generated by others, so a baseline sharing region may exist even without privacy protection or platform signals. We characterize this sharing-feasible region and show how external information and privacy noise expand it beyond the reciprocal-access baseline. We further provide conditions under which full sharing, rather than partial participation, is the unique participation equilibrium. Finally, we show that privacy noise is valuable as an implementation tool but costly for welfare, so the platform chooses the least distortionary privacy level that implements full sharing, and privacy-induced sharing improves total surplus only when participation gains outweigh informational losses.

}


\maketitle


\pagestyle{fancy}
\fancyhead{}


\section{Introduction}

Digital platforms increasingly serve as information intermediaries in modern supply chains. Suppliers operating on these platforms face demand uncertainty, cost shocks, and operational disruptions, and information sharing can improve forecasting, production planning, and system-wide efficiency \citep{ha2011sharing,guan2020demand,zhang2006competition}. Yet competing suppliers may be reluctant to disclose private information because such disclosures can alter rivals' beliefs and intensify strategic competition \citep{gal1985information,raith1996general}. As a result, markets may exhibit inefficiently low levels of information exchange, especially when firms can benefit from others' disclosures without making reciprocal contributions.

This paper studies how a platform can design information-sharing mechanisms that induce voluntary participation among competing suppliers. We consider a Cournot oligopoly with an uncertain demand state. Each supplier observes a Gaussian private signal of the demand state and chooses whether to participate in a platform-mediated sharing mechanism. A participating firm reports its signal to the platform, but the platform does not disclose raw reports. Instead, it adds independent Gaussian noise to each submitted signal and releases a posterior estimate computed from the noisy reports together with the platform's external signal. This mechanism creates a central trade-off: a larger noise variance makes each firm's report less informative to its rivals, but it also reduces the precision of the posterior estimate received by participants. The platform's design problem is therefore to choose privacy noise that reduces strategic exposure while preserving enough informational value to induce participation.

Our analysis delivers four main results. First, in a two-firm market without external platform information, privacy protection alone cannot sustain voluntary information sharing. Adding noise reduces the extent to which a firm's report informs its rival, but it also weakens the value of the information made available through the platform. Thus, when the platform's information is generated solely from firms' own reports, each firm still prefers to avoid helping its rival forecast demand, even if the shared signal is partially obscured by privacy noise.

Second, the role of participation-contingent access changes the incentives in larger markets. Following the quid-pro-quo information exchange structure in \citet{kirby1988trade}, we study a reciprocal-access rule under which firms receive the aggregated signal only if they contribute their own signals. Under this rule, a firm that opts out not only avoids disclosure but also loses access to the pooled information generated by other participants. This access-loss effect can sustain full participation in an $n$-firm market even without privacy protection or external platform information. We characterize the corresponding sharing-feasible region and show how it depends on market size and signal precision.

Third, when reciprocal access alone is insufficient, platform-owned external information and privacy design act as complementary instruments. An external signal creates informational value that does not originate from competitors' disclosures, while privacy noise limits how strongly a firm's report affects rivals' beliefs. We show that a sufficiently informative platform signal can restore participation, and that calibrated privacy noise can further expand the sharing-feasible region when the platform signal alone is not enough. This complementarity also extends to environments with heterogeneous signal precision, where individually calibrated privacy noise can be used to induce participation by firms with different information qualities.

Finally, we connect these participation results to welfare and platform design. To implement full sharing as the unique participation equilibrium, the platform must make joining profitable not only at the full-sharing profile, but at every level of existing participation, thereby ruling out partial-sharing equilibria. Our welfare analysis shows that privacy noise is useful only as an implementation device: it can relax participation constraints by reducing strategic exposure, but it also lowers posterior precision and therefore reduces surplus conditional on participation. Hence, the platform chooses the least distortionary privacy level that implements full sharing, and privacy-induced sharing is welfare improving only when the participation gain outweighs the informational loss. We also show that under heterogeneous signal precision, privacy protection must be calibrated across firms, since more informative signals create greater strategic exposure and may require stronger protection. Finally, in a sequential participation setting, early adoption can incentivize later firms to join by increasing the value of the pooled signal available through the platform.

The paper contributes to the literature on information sharing under oligopoly competition by highlighting the interaction between reciprocal access, platform information, and privacy-preserving mechanism design. The main message is that privacy should not be viewed merely as a compliance constraint. In competitive information-sharing environments, privacy is a strategic design instrument whose effectiveness depends critically on the surrounding information environment. 


The remainder of the paper is organized as follows. Section~2 reviews related literature; Section~3 presents the model and benchmark regimes; Section~4 studies privacy design and participation incentives; Section~5 analyzes welfare and platform choice; and Section~6 concludes. The appendix contains proofs, consumer- and producer-surplus derivations, and additional results on heterogeneous signals and sequential participation.

\section{Literature Review}

Our work connects three strands of literature: information sharing under competition, platform-mediated information design, and privacy-preserving mechanism design.

First, we build on the classical literature on information sharing among competing firms. Foundational work shows that disclosure can intensify competition by aligning firms' beliefs and increasing strategic responsiveness, thereby reducing firms' incentives to share private information \citep{clarke1983collusion,vives1984duopoly,gal1985information,raith1996general}. A key distinction in this literature concerns the access rule governing shared information. In some settings, firms may benefit from rivals' disclosures even without contributing their own information, which weakens voluntary sharing incentives. By contrast, \citet{kirby1988trade} studies quid-pro-quo information exchange, where access is tied to participation. Our baseline adopts this reciprocal-access structure: a firm receives the aggregated signal only if it contributes its own signal. This distinction changes the equilibrium incentives and the resulting design problem. Because opting out also means losing access to the pooled information generated by other participants, reciprocal access can sustain sharing in some multi-firm markets and determines when additional platform instruments, such as external information and privacy noise, are needed.

Second, our paper relates to work on platform-mediated information sharing. Recent studies examine how platforms use information to shape competition, pricing, and supply-chain outcomes \citep{li2021information,zha2023strategic,avinadav2025information,hu2025information,gong2024impact}. Closest to our setting, \citet{liu2021information} study a retail platform that owns superior demand information and chooses whether to disclose it to competing sellers. Our focus is different. Rather than studying whether a platform should disclose information it exclusively controls, we study whether competing firms voluntarily participate in a common information-sharing mechanism. This changes both the equilibrium question and the design problem: the platform must induce firms to contribute information, rule out partial participation, and combine external information with privacy noise when participation incentives fail.

Third, our paper contributes to the literature on privacy-preserving mechanism design. Foundational work studies how privacy guarantees can be incorporated into mechanism design without undermining incentives or efficiency \citep{mcsherry2007mechanism,nissim2012privacy,pai2013privacy,kearns2014mechanism}.
More recent work on privacy-aware data acquisition examines how platforms use privacy guarantees and compensation to induce data provision \citep{chen2018optimal,cummings2023optimal,fallah2024optimal,acemoglu2023good}. In our setting, privacy noise has a dual effect: it reduces strategic exposure by limiting how much a firm's report affects rivals' beliefs, but it also lowers the precision of the aggregated signal. We show that this trade-off implies that privacy protection alone generally cannot restore voluntary sharing. Instead, privacy becomes effective when combined with platform-owned external information, which provides informational value that does not originate from competitors' reports. Our main contribution is to characterize this complementarity between privacy design and the external information environment.

\section{The Efficiency Dilemma in Competitive and Uncertain Markets}
\label{sec:problem}

We consider a set of $n$ suppliers indexed by $i=1,\dots,n$ who compete in a Cournot market by choosing non-negative quantities $q_i \ge 0$. The market price depends on total output $Q = \sum_{i=1}^n q_i$ and an unknown common parameter $\theta \in \mathbb{R}$ that captures shared market uncertainty. Following the standard setup in the information-sharing literature (e.g., \citealt{gal1985information}), we adopt the linear inverse demand function
\[
P(Q,\theta) = a + \theta - Q,
\]
where $a > 0$ is a deterministic intercept and $\theta$ is a zero-mean random variable representing the uncertain market condition, with $\theta \sim \mathcal{N}(0,u)$ and $u>0$.

Each supplier privately observes a Gaussian signal of the demand state, $s_i \mid \theta \sim \mathcal{N}(\theta,\tau_i)$, where $\tau_i$ measures the noise variance of firm $i$'s private signal. Signals are conditionally independent across suppliers. For example, these signals may represent supplier-specific indicators of future demand, such as advance orders, quote requests, or regional sales leads, which the platform can aggregate into a common demand forecast. For the symmetric benchmark, we set $\tau_i=\tau$ for all $i$.

In addition to suppliers' private signals, the platform also observes an external, independent Gaussian signal of the demand state, $s_p \mid \theta \sim \mathcal{N}(\theta,\tau_p)$, where $\tau_p>0$ is fixed and exogenously given. This signal is collected independently by the platform and is available regardless of suppliers' participation decisions.

If a subset $J \subseteq \{1,\dots,n\}$ of suppliers participates, each firm $i\in J$ reports its signal $s_i$ to the platform. The platform applies a firm-specific privacy parameter $m_i$ and perturbs each report as
\[
\tilde{s}_i = s_i + v_i,
\qquad
v_i \sim \mathcal{N}(0,m_i),
\]
where the noise terms $\{v_i\}_{i\in J}$ are mutually independent and independent of $(\theta,\{s_i\}_{i\in J},s_p)$. Adding independent Gaussian noise provides a form of approximate differential privacy \citep{dwork2014algorithmic}. In our mechanism-design interpretation, the noise levels $m_i$ trade off strategic protection and informational precision.

The platform aggregates the noisy reports $\{\tilde{s}_j\}_{j \in J}$ together with its own signal $s_p$ and computes the posterior mean
\[
\mu_J := \mathbb{E}[\theta \mid \{ \tilde{s}_j \}_{j \in J}, s_p],
\]
which is shared with all participating firms. Under the Gaussian assumptions,
\[
\mu_J
=
V_J
\left(
\sum_{j \in J} \frac{\tilde{s}_j}{\tau + m_j}
+
\frac{s_p}{\tau_p}
\right),
\qquad
V_J
:=
\operatorname{Var}\!\left(\theta \mid \{ \tilde{s}_j \}_{j \in J}, s_p\right)
=
\left(
\frac{1}{u}
+
\sum_{j \in J} \frac{1}{\tau + m_j}
+
\frac{1}{\tau_p}
\right)^{-1}.
\]
The platform does not disclose raw supplier reports; it releases only the posterior mean. In symmetric environments with $m_j=m$, the posterior depends only on the number of sharing firms $k=|J|$. We write $\mu_k$ for the posterior mean and $V_k$ for the posterior variance of $\theta$ when $k$ symmetric firms participate. In particular, $\mu_1$ corresponds to the one-firm case, $\mu_2$ to the two-firm case, and $\mu_n$ to the full $n$-firm sharing profile.

\subsection{Information Sharing Platform and Participation Decisions}

Each supplier chooses whether to \emph{join} the platform. Let \(x_i \in \{0,1\}\) denote supplier $i$'s participation decision, where $x_i=1$ means joining the platform and $x_i=0$ means non-participation.

If supplier $i$ does not join the platform, its information set is \(\mathcal{I}_i^{NS}=\{s_i\}\). If firm $i$ participates, its information set is \(\mathcal{I}_i^{S}=\{s_i,\mu_J\}\). Firms then choose quantities according to Bayesian best responses:
\[
q_i^{NS}(\mathcal{I}_i^{NS})
=
\operatorname*{arg\,max}_{q_i\ge 0}
\mathbb{E}_{\theta\mid s_i}
\!\left[q_i(a+\theta-Q)\right],
\qquad
q_i^{S}(\mathcal{I}_i^{S})
=
\operatorname*{arg\,max}_{q_i\ge 0}
\mathbb{E}_{\theta\mid s_i,\mu_J}
\!\left[q_i(a+\theta-Q)\right].
\]

\subsection{Equilibrium Structure and Benchmark Regimes}

The model defines a three-stage game. First, the platform commits to a privacy-preserving aggregation mechanism by announcing the privacy parameters $\{m_i\}_{i=1}^n$. Second, suppliers choose participation decisions $x_i \in \{0,1\}$. Third, suppliers choose quantities given the induced information structure.

Let \(\pi_i(q_i,Q,\theta)=q_i(a+\theta-Q)\) denote firm $i$'s realized profit at total output $Q$. We write \(\pi_i^{S}\) for the equilibrium profit induced by the sharing information regime and \(\pi_i^{NS}\) for the equilibrium profit induced by the no-sharing information regime. That is,
\[
\pi_i^{S}(x_i=1,x_{-i}=\mathbf{1})
=
q_i^{S}(\mathcal{I}_i^{S})(a+\theta-Q^{S}),
\qquad
\pi_i^{NS}(x_i=0,x_{-i}=\mathbf{1})
=
q_i^{NS}(\mathcal{I}_i^{NS})(a+\theta-Q^{NS}).
\] Our focus is on whether there exists a privacy-preserving mechanism $\{m_i\}_{i=1}^n$ such that platform participation is a Nash equilibrium. We define the participation incentive for supplier $i$ as
\[
\Delta_i
:=
\mathbb{E}\!\left[
\pi_i^{S}(x_i=1,x_{-i}=\mathbf{1})
\right]
-
\mathbb{E}\!\left[
\pi_i^{NS}(x_i=0,x_{-i}=\mathbf{1})
\right],
\]
which captures the expected profit gain from joining the platform when all other suppliers participate. Full sharing is a Nash equilibrium if and only if $\Delta_i \ge 0$ for all $i \in \{1,\dots,n\}$. When the suppliers are symmetric, we write $\Delta := \Delta_i$, and the condition reduces to $\Delta \ge 0$.

\subsection{Platform Objective}

The platform chooses privacy noise to maximize expected total surplus under full participation, subject to participation being incentive compatible. Given total output $Q=\sum_{i=1}^n q_i$, consumer surplus, producer surplus, and total surplus are (the derivation is provided in \cref{app:surplus-derivations}):
\[
CS(Q,\theta)=\frac{Q^2}{2},
\qquad
PS(Q,\theta)=Q(a+\theta-Q),
\qquad
TS(Q,\theta)=(a+\theta)Q-\frac{Q^2}{2}.
\]

Let $q_i^S(m)$ denote firm $i$'s equilibrium quantity under full participation when the platform chooses privacy noise level $m$, and let $Q^S(m)=\sum_{i=1}^n q_i^S(m)$ denote the corresponding aggregate output.

To rule out partial-participation equilibria, the platform requires joining to be profitable at every possible level of existing participation. For each $k\in\{0,\ldots,n-1\}$, let $x_{-i}^{(k)}\in\{0,1\}^{n-1}$ denote any participation vector of firms other than $i$ with exactly $k$ participating firms. By symmetry, firm $i$'s expected equilibrium profit depends only on the number of participating firms in $x_{-i}^{(k)}$. Define
\[
\Delta^k(m)
=
\mathbb{E}\!\left[
\pi_i(x_i=1,x_{-i}^{(k)};m)
\right]
-
\mathbb{E}\!\left[
\pi_i(x_i=0,x_{-i}^{(k)};m)
\right],
\]
which is the expected gain from joining the platform when $k$ other firms participate. The full-participation feasible set is
\[
\mathcal{M}^{F}
=
\left\{
m\ge 0:
\Delta^k(m)>0,\quad k=0,\ldots,n-1
\right\}.
\]

The condition $\Delta^k(m)>0$ for all $k=0,\ldots,n-1$ implies that no partial-participation profile can be a Nash equilibrium. If fewer than $n$ firms participate, then any non-participating firm faces some number $k$ of existing participants and strictly prefers to join. At the full-participation profile, $\Delta^{n-1}(m)>0$ ensures that no firm wants to deviate by leaving the platform. Hence full participation is the unique pure-strategy Nash equilibrium of the participation game.

The platform's problem is
\[
\max_{m\in\mathcal{M}^{F}}
\mathbb{E}\left[
TS(Q^S(m),\theta)
\right],
\]
or equivalently,
\[
\max_{m\in\mathcal{M}^{F}}
\mathbb{E}\left[
(a+\theta)Q^S(m)-\frac{(Q^S(m))^2}{2}
\right].
\]
Thus, the platform selects a privacy level that maximizes full-participation surplus while ensuring that joining is profitable for an outsider at every possible level of existing participation.

\section{Privacy-Preserving Mechanism Design for Information Sharing}

\subsection{A Two-Firm Information-Sharing Model}

We begin with a stylized example involving two symmetric competing firms. To isolate the role of privacy noise in discouraging participation, we first consider a benchmark in which the platform does not possess any external signal. This corresponds to taking the platform signal variance \(\tau_p \to \infty\), so that any posterior released by the platform is based solely on supplier reports.

\begin{proposition}[No Sharing Equilibria for Two Firms Without External Signals]\label{prop1:negative-result}
Consider two competing symmetric firms in the absence of external information provided by the platform (i.e., $\tau_p \to \infty$).
There exists no privacy-preserving information aggregation mechanism that implements information sharing as a Nash equilibrium of the participation game.
\end{proposition}

\Cref{prop1:negative-result}, proved in \cref{app:proof:prop1:negative-result}, shows that privacy noise alone cannot sustain information sharing in symmetric duopoly. The reason is that when the platform's information is generated only from firms' own reports, releasing a common posterior strengthens strategic coupling between competitors. Adding noise weakens this exposure, but it also reduces the informational value of participation. Hence, the net participation gain remains negative.

This result is consistent with the no-sharing logic in \citet{gal1985information} and highlights the need for an exogenous source of informational value. If the platform also observes an external signal that is not contributed by firms, participation grants access to information that does not arise from disclosing one's own signal to a rival. The next result shows that such platform-owned information can expand the region in which $\Delta \ge 0$ and sustain information sharing in equilibrium.

\begin{proposition}[Two-Firm Information Sharing Equilibrium with Platform Signal]
\label{prop:platform-signal-2}
Consider the symmetric two-firm Cournot model with common state $\theta \sim \mathcal{N}(0, u)$, private signals $s_i \mid \theta \sim \mathcal{N}(\theta, \tau)$ for $i= 1,2$, and a platform-provided signal $s_p \mid \theta \sim \mathcal{N}(\theta, \tau_p)$.

For fixed accuracy parameters $u, \tau, \tau_p > 0$:
\begin{enumerate}
    \item There exists a threshold $\tau_p^* = \tau_p^*(u, \tau) \in (0, \infty)$ such that the Sharing profile is the unique symmetric Bayesian Nash equilibrium if and only if $\tau_p \le \tau_p^*$.
    \item The threshold $\tau_p^*$ is uniquely characterized by the indifference condition:
    \begin{equation}
    \label{eq:taup-star-equation-2}
    \frac{\text{Var}(\mu_{2})}{9} = \beta_0(\tau_p)^2 \text{Var}(\mu_1),
    \end{equation}
where $\mu_2=\mathbb{E}[\theta\mid s_1,s_2,s_p]=\left(\frac{1}{u}+\frac{2}{\tau}+\frac{1}{\tau_p}\right)^{-1}\left(\frac{s_1}{\tau}+\frac{s_2}{\tau}+\frac{s_p}{\tau_p}\right)$, $\mu_1=\mathbb{E}[\theta\mid s_1]=\left(\frac{1}{u}+\frac{1}{\tau}\right)^{-1}\frac{s_1}{\tau}$, and $\beta_0(\tau_p)$ is the equilibrium response coefficient defined in \eqref{eq:beta0-closed-2}.

\end{enumerate}
\end{proposition}

\Cref{prop:platform-signal-2} shows that an external platform signal can restore sharing incentives in duopoly. The indifference condition in \eqref{eq:taup-star-equation-2} equates the informational value of participation, captured by $\text{Var}(\mu_{2})$, with the strategic benefit of non-participation, captured by $\beta_0(\tau_p)^2 \text{Var}(\mu_1)$. Thus, $\tau_p^*$ is the threshold level of platform-signal precision required for full sharing to be sustained as the unique symmetric equilibrium.

This result motivates the role of privacy design when the platform signal exists but is not sufficiently precise on its own. Privacy noise reduces the extent to which a firm's report affects its rival's beliefs, but it also lowers the precision of the released signal. The next proposition, proved in \cref{app:proof:prop:privacy-restores-sharing-2}, shows that when the platform provides an external signal, an appropriately chosen noise level $m$ can further restore participation incentives.

\begin{proposition}[Privacy Noise Induces Voluntary Sharing  for Two Firms]
\label{prop:privacy-restores-sharing-2}
In the two-firm symmetric Cournot model, suppose that the platform signal has finite variance, i.e., $\tau_p<\infty$, but is not precise enough
to sustain full participation under the non-private mechanism, i.e., $\tau_p^*(u,\tau)<\tau_p<\infty$.
Then there exists a finite privacy noise level $m^*>0$ such that, under the privacy-preserving mechanism with $m \ge m^*$, each firm strictly prefers to participate when the other does.
Thus, full participation becomes a unique symmetric equilibrium.
\end{proposition}

\Cref{prop:privacy-restores-sharing-2} shows that privacy noise can restore participation when the platform signal is present but not sufficiently precise on its own. By reducing the extent to which a firm's report affects its rival's beliefs, privacy noise weakens strategic coupling; meanwhile, the external platform signal preserves an exogenous informational benefit from participation. Hence, for sufficiently large $m$, the participation gain becomes positive, and there exists a finite threshold $m^*>0$ that sustains full participation.

The two-firm model highlights the basic mechanism. We next extend the analysis to an $n$-firm Cournot market and characterize how the sharing-feasible region depends on market size $n$, platform signal noise $\tau_p$, and privacy noise $m$. Extensions to heterogeneous signal precision and sequential participation decisions are discussed in \cref{sec:heterogeneous-sequential}.

\subsection{An \texorpdfstring{$n$}{n}-Firm Information-Sharing Model}

We extend the analysis to an $n$-firm symmetric Cournot market. Participating firms receive a posterior mean computed from privatized firm reports and the platform's external signal, which improves inference while limiting strategic exposure. We first study a baseline without privacy protection or platform signals, where the platform simply aggregates raw firm reports. Unlike in the two-firm case, reciprocal access can make opting out costly in larger markets, so voluntary sharing may arise even without additional platform instruments. The next proposition characterizes this baseline equilibrium.

\begin{proposition}[$n$-Firm Sharing Equilibrium without Privacy Protection]
\label{prop:sharing_equilibrium_noprivacy}
Consider the $n$-firm Cournot information-sharing model without privacy protection ($m=0$).
Let $\Delta$ denote the unilateral participation gain of a representative firm when all other firms share.

Information sharing constitutes a Nash equilibrium if and only if $\Delta \ge 0$, 
which under symmetry is equivalent to the condition
\begin{equation}\label{eq:sharing_condition_nfirms}
\frac{n}{(n+1)^2\,(n u+\tau)}
\;\ge\;
\frac{(u+\tau)\big((n-1)u+n\tau\big)^2}
{\big((n^2-1)u^2+2n^2u\tau+2n\tau^2\big)^2}.
\end{equation}
\end{proposition}

\Cref{prop:sharing_equilibrium_noprivacy}, proved in 
\cref{app:proof:prop:sharing_equilibrium_noprivacy}, characterizes the baseline reciprocal-access environment without privacy protection or platform-owned information. Under this access rule, a firm receives the aggregated signal only if it contributes its own signal. Hence, a unilateral non-participant avoids disclosure but also loses access to the pooled information generated by other firms. When the number of participating firms is large, this access loss can outweigh the strategic benefit of withholding one's signal, so full information sharing may be sustained even without additional platform instruments.

To express this condition more transparently, define the precision ratio $r=\tau/u$. After simplification, the participation condition $\Delta\ge 0$ is equivalent to $\Phi(n,r)\ge 0$, where $\Phi(n,r)=a_0(n)+a_1(n)r+a_2(n)r^2+a_3(n)r^3$, with
\begin{equation} \label{eq:PHI}
\begin{aligned}
a_0(n)&=n^4-2n^3-2n^2+2n+1, 
&\quad a_1(n)&=n^4-2n^3-3n^2-n+1,\\
a_2(n)&=-n^4+2n^3-3n^2-2n, 
&\quad a_3(n)&=-(n-1)n^2<0\quad (n\ge 2).
\end{aligned}    
\end{equation}

\begin{corollary}[Characterizing the Threshold for $n$-Firm Sharing Equilibrium]
\label{cor:threshold_sharing_monotonicity}
For each $n\ge 3$, the cubic polynomial $\Phi(n,r)$ has a unique positive root, denoted by $r_0(n)$.
Information sharing constitutes a Nash equilibrium if and only if $\tau/u\le r_0(n)$.
Moreover, the threshold $r_0(n)$ is increasing in $n$.
Therefore, as the number of firms increases, the range of signal precision ratios that sustain voluntary sharing expands monotonically.
\end{corollary}

\Cref{cor:threshold_sharing_monotonicity}, proved in \cref{app:proof:cor:threshold_sharing_monotonicity}, shows that larger markets expand the sharing-feasible region. Intuitively, as more firms participate, the aggregate signal becomes more informative, making deviation more costly because an opting-out firm loses access to a richer pooled signal.



\Cref{fig:sharing_regions_baseline_platform}\textup{(a)} illustrates the monotone increase of $r_0(n)$ in $n$. Thus, in the reciprocal-access baseline, larger markets can sustain information sharing for a wider range of private-signal noise levels.

The analysis above shows that full information sharing may arise endogenously in an $n$-firm market, but only within the sharing-feasible region. Outside this region, additional platform instruments are needed. We next examine how an external platform-owned signal affects equilibrium participation.

\begin{proposition}[$n$-Firm Information-Sharing Equilibrium with Platform Signal]
\label{prop:platform-signal-n}
Consider the $n$-firm Cournot model with fixed parameters $n, u, \tau, \tau_p > 0$. 
\begin{enumerate}
    \item There exists a threshold $\tau_p^* = \tau_p^*(n, u, \tau) \in (0, \infty)$ such that the Sharing profile is the symmetric Bayesian Nash equilibrium if and only if $\tau_p \le \tau_p^*$.
    \item The threshold $\tau_p^*$ is uniquely characterized by the indifference condition:
    \begin{equation}
    \label{eq:taup-star-equation}
    \frac{\text{Var}(\mu_{n})}{(n+1)^2} = \beta_0(\tau_p)^2 \text{Var}(\mu_1),
    \end{equation}
where $ \mu_n = \left(\frac{1}{u}+\frac{n}{\tau}+\frac{1}{\tau_p}\right)^{-1} \left(\sum_{i=1}^n \frac{s_i}{\tau}+\frac{s_p}{\tau_p}\right), \quad \mu_1 = \left(\frac{1}{u}+\frac{1}{\tau}\right)^{-1} \frac{s_1}{\tau}$.

\end{enumerate}
\end{proposition}

\Cref{prop:platform-signal-n}, proved in \cref{app:proof:prop:platform-signal-n}, shows that an external platform signal can restore information sharing when the baseline reciprocal-access mechanism is insufficient. The indifference condition in \eqref{eq:taup-star-equation} equates the informational value of participation with the strategic benefit of non-participation. Thus, $\tau_p^*$ separates regimes in which platform-owned information alone can sustain sharing from those in which it cannot.


\begin{figure}[t]
    \centering
    \begin{subfigure}[t]{0.31\textwidth}
        \centering
        \includegraphics[width=\linewidth]{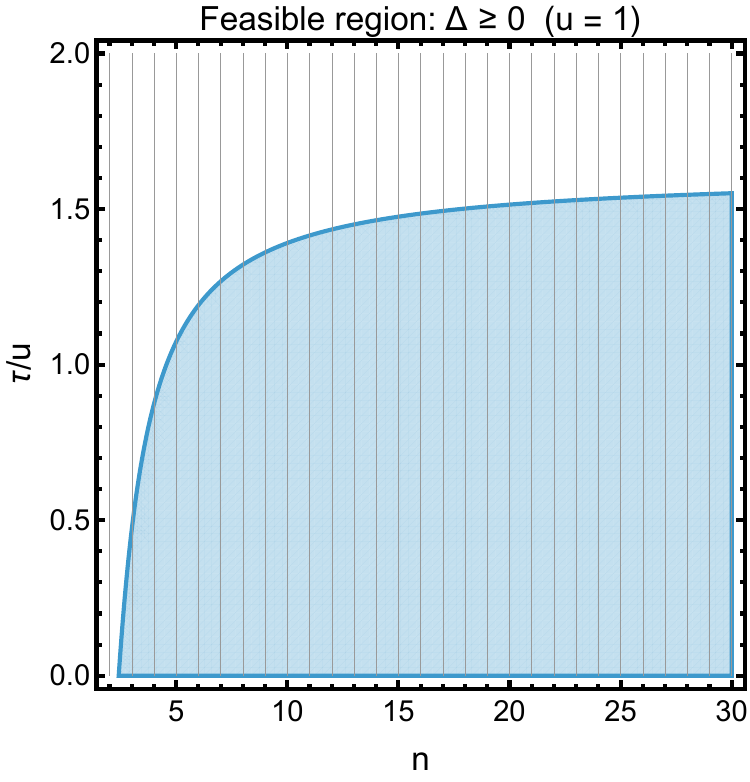}
        \subcaption{No platform signal}
        \label{fig:sharing_region_no_privacy}
    \end{subfigure}\hfill
    \begin{subfigure}[t]{0.31\textwidth}
        \centering
        \includegraphics[width=\linewidth]{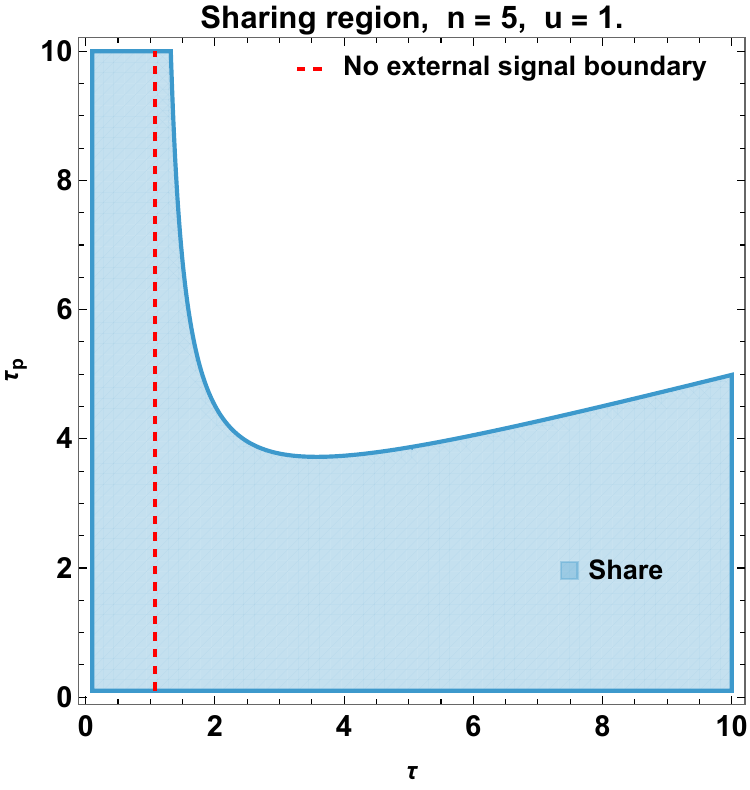}
        \subcaption{Platform signal, $n=5$}
        \label{fig:sharing_region_with_platform_signal_n5}
    \end{subfigure}\hfill
    \begin{subfigure}[t]{0.31\textwidth}
        \centering
        \includegraphics[width=\linewidth]{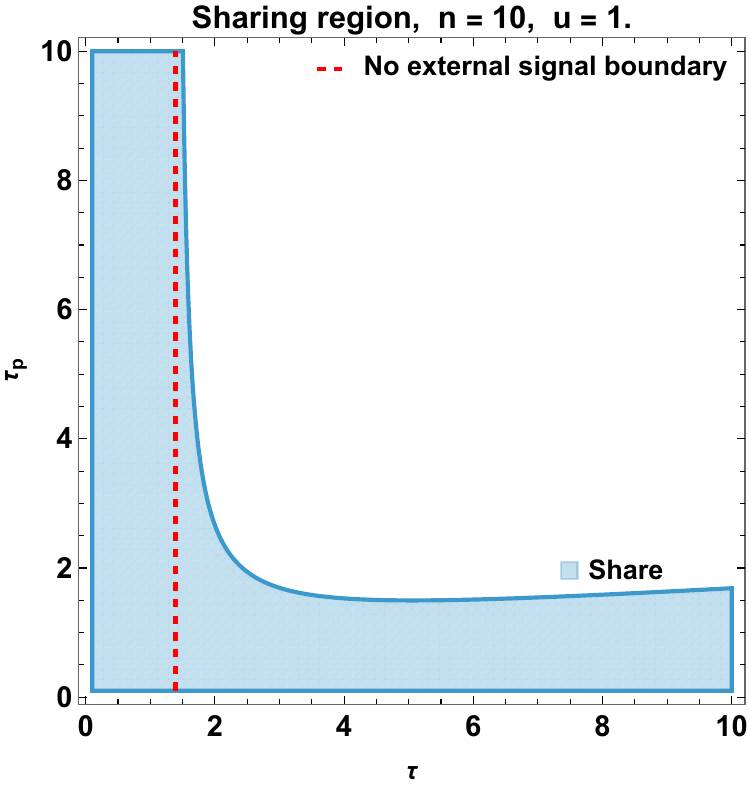}
        \subcaption{Platform signal, $n=10$}
        \label{fig:sharing_region_with_platform_signal_n10}
    \end{subfigure}
    \caption{\textbf{Sharing regions under reciprocal access and platform information.}
    Panel (a) shows the no-platform-signal benchmark, where the boundary is defined by $\Phi(n,\tau/u)=0$ in \cref{eq:PHI}, with asymptote $r_\infty=(1+\sqrt{5})/2$. Panels (b) and (c) show the sharing-feasible regions under an external platform signal for $n=5$ and $n=10$, respectively; shaded regions indicate where sharing is a symmetric Bayesian Nash equilibrium in the $(\tau,\tau_p)$-plane. The dashed vertical line is the no-external-signal benchmark.}
    \label{fig:sharing_regions_baseline_platform}
    \label{fig:sharing_region_with_platform_signal}
\end{figure}


\Cref{fig:sharing_regions_baseline_platform}\textup{(b)--(c)} illustrate how platform-owned information expands the sharing-feasible region. The comparison between $n=5$ and $n=10$ also shows that market size affects participation incentives through competing forces: reciprocal access becomes more valuable in larger markets, but common information can also intensify Cournot competition. When the platform signal is too noisy to sustain participation on its own, additional design instruments are needed. We next introduce privacy-preserving aggregation.

\begin{proposition}[Privacy Noise Induces Voluntary Sharing among $n$ Firms]
\label{prop:privacy-restores-sharing-n}
Consider the $n$-firm symmetric Cournot model described above.
Suppose the platform signal has finite variance, i.e., $\tau_p<\infty$, but is not precise enough
to sustain full participation under the non-private mechanism, i.e.,
$\tau_p^*(n,u,\tau)<\tau_p<\infty$.
Then there exists a finite privacy noise level $m^*>0$ such that, under the privacy-preserving
mechanism with noise variance $m\ge m^*$, every firm strictly prefers to participate when all
other firms participate.
Consequently, the sharing profile is sustained as a symmetric equilibrium.
\end{proposition}

\Cref{prop:privacy-restores-sharing-n}, proved in \cref{app:proof:prop:privacy-restores-sharing-n}, establishes the complementarity between external platform information and privacy design. The platform signal preserves an informational benefit from participation, while privacy noise reduces the extent to which a firm's report affects its rivals' beliefs. Thus, when the platform signal alone is insufficient, a sufficiently large noise level can restore participation incentives.

\begin{figure}[t]
    \centering
    \includegraphics[width=0.4\textwidth]{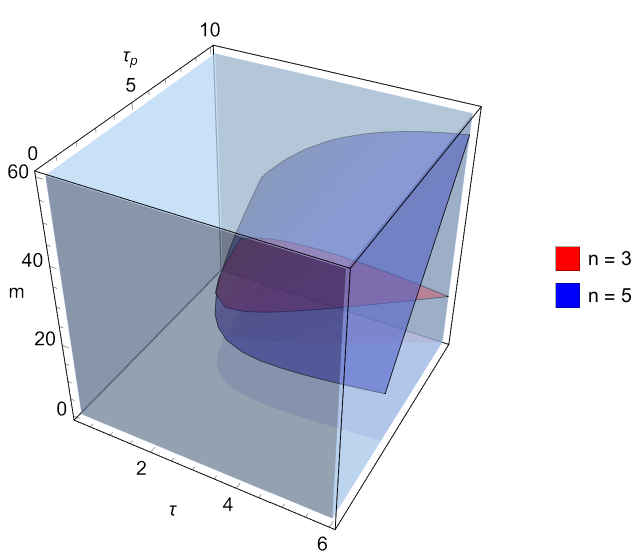}
    \caption{
    \textbf{Indifference surfaces under the privacy-preserving mechanism.}
    The surfaces are defined by $\Delta(\tau,\tau_p,m)=0$ in the $(\tau,\tau_p,m)$ space. Parameter values above the surface correspond to regions where privacy noise restores participation. The surfaces are plotted for $u=1$ and $a=1$, with red for $n=3$ and blue for $n=5$.
    }
    \label{fig:sharing_region_under_privacy_mechanism}
\end{figure}

\Cref{fig:sharing_region_under_privacy_mechanism} illustrates how the required privacy noise varies with the information environment. A noisier platform signal, reflected by a larger $\tau_p$, generally requires stronger privacy protection because participation relies more heavily on reducing strategic exposure. By contrast, the effect of firms' private-signal noise $\tau$ is not globally monotone: increasing $\tau$ can both reduce the value of private information and increase reliance on shared or platform-provided information.

The comparison across $n$ shows that market size changes the strength of competitive externalities. Larger markets may expand the non-sharing region, but the required privacy noise is not uniformly higher for larger $n$ because the effect depends on the precision of the platform signal. Overall, privacy expands the sharing-feasible region only when it is calibrated jointly with firms' signal precision, platform-signal accuracy, and market size.

\section{Welfare Effects and Platform Choice of Privacy Noise}
\label{sec:welfare-platform-choice}

We now examine the welfare implications of privacy-preserving information sharing. Privacy noise can relax firms' participation constraints by reducing strategic exposure, but it also lowers the precision of the posterior signal released by the platform. Thus, conditional on inducing participation, the platform faces a trade-off between incentive compatibility and informational efficiency.

Let $\sigma_S^2(m)$ denote the variance of the posterior mean under full platform participation:
\[
\sigma_S^2(m)
=
u-
\left(
\frac{1}{u}
+
\frac{n}{\tau+m}
+
\frac{1}{\tau_p}
\right)^{-1}.
\]
Under full participation, expected total surplus is
\[
\mathbb{E}[TS^S(m)]
=
\frac{n(n+2)}{2(n+1)^2}
\left(a^2+\sigma_S^2(m)\right).
\]
Under no information sharing, define
\[
\sigma_P^2=\frac{u^2}{u+\tau},
\qquad
\rho=\frac{u}{u+\tau},
\qquad
\beta=\frac{1}{2+(n-1)\rho}.
\]
Then
\[
\mathbb{E}[TS^{NS}]
=
\frac{n(n+2)a^2}{2(n+1)^2}
+
\frac{n\beta^2\sigma_P^2}{2}
\left[3+(n-1)\rho\right].
\]
Hence the total-surplus gain from platform participation is
\[
\Delta TS(m)
=
\frac{n}{2}
\left[
\frac{n+2}{(n+1)^2}\sigma_S^2(m)
-
\beta^2\sigma_P^2[3+(n-1)\rho]
\right].
\]
The derivations of consumer surplus, producer surplus, and additional comparative statics with respect to $u$, $\tau$, and $n$ are provided in \cref{app:welfare-derivations}.

\begin{proposition}[Effect of privacy noise]
\label{prop:privacy-noise-welfare}
The expected consumer surplus, producer surplus, and total surplus under full
platform participation are strictly decreasing in the privacy noise $m$. Moreover,
because the no-sharing regime is independent of $m$, the surplus gains
$\Delta CS$, $\Delta PS$, and $\Delta TS$ are also strictly decreasing in $m$.
\end{proposition}

\begin{proposition}[Effect of platform signal noise]
\label{prop:platform-signal-noise-welfare}
The expected consumer surplus, producer surplus, and total surplus under full
platform participation are strictly decreasing in the platform signal noise $\tau_p$.
Moreover, because the no-sharing regime is independent of $\tau_p$, the surplus
gains $\Delta CS$, $\Delta PS$, and $\Delta TS$ are also strictly decreasing in
$\tau_p$.
\end{proposition}

\Cref{prop:privacy-noise-welfare,prop:platform-signal-noise-welfare} show that the welfare value of platform participation is increasing in the informativeness of the released posterior. Privacy noise may be necessary for implementation, but once participation is incentive compatible, additional noise only lowers welfare. Similarly, a less precise platform signal reduces the welfare gain from sharing.

\subsection{Implementation and Platform Choice}
\label{subsec:inducing-full-sharing}

The participation constraints defined in \cref{sec:problem} also determine when the platform can rule out partial-participation equilibria.

\begin{proposition}[Full-sharing implementation]
\label{prop:full-sharing-implementation}
If $m\in\mathcal{M}^{F}$, then full participation is the unique pure-strategy Nash equilibrium of the participation game.
\end{proposition}

The reason is that $m\in\mathcal{M}^{F}$ implies $\Delta^k(m)>0$ for every $k=0,\ldots,n-1$. Hence, at any partial-participation profile, an outsider strictly prefers to join; at the full-participation profile, no firm wants to exit. Thus, all partial-sharing profiles are eliminated, and full sharing is uniquely implemented.

Combining this implementation condition with \Cref{prop:privacy-noise-welfare}, the platform solves $\max_{m\in\mathcal{M}^{F}}\mathbb{E}[TS^S(m)]$. Since expected surplus under full participation is decreasing in $m$, the platform chooses the least distortionary feasible privacy level, $m^{F}=\inf\{m\ge 0:m\in\mathcal{M}^{F}\}$. This distinguishes implementation from welfare improvement: privacy noise may induce full sharing, but the implemented outcome improves welfare only if $\Delta TS(m^F)>0$.

\subsection{Numerical Illustration}
\label{subsec:numerical-implementation}

We illustrate the implementation procedure numerically. The figures below visualize three objects generated by the model: the joining incentives $\Delta^k(m)$ at different participation levels, the least privacy noise $m^F$ needed to implement full participation, and the welfare classification of parameter regions after implementation. In each exercise, the remaining primitives are fixed at the baseline values used in the numerical computation.

\begin{figure}[t]
    \centering
    \includegraphics[width=0.98\textwidth]{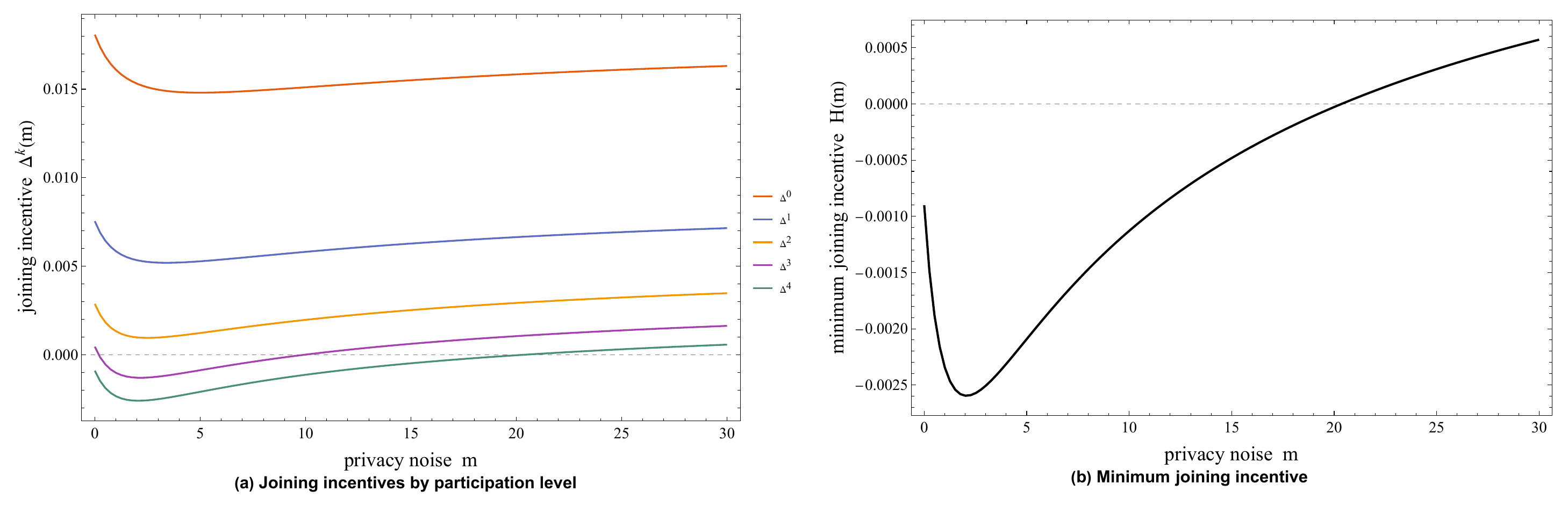}
    \caption{\textbf{Joining incentives and implementation.}
    For a fixed parameter configuration, Panel (a) plots the joining incentive $\Delta^k(m)$ as a function of privacy noise $m$, separately for each possible number $k$ of existing participants. Panel (b) plots the lower envelope $H(m)=\min_{k=0,\ldots,n-1}\Delta^k(m)$. Full participation is implementable exactly when $H(m)>0$.}
    \label{fig:delta_curves}
\end{figure}

\Cref{fig:delta_curves} shows why implementation requires checking all participation levels. Even if some joining constraints are positive, full sharing is not uniquely implemented unless the lowest joining incentive is also positive. The figure also illustrates the non-monotone role of privacy noise: moderate noise can reduce strategic exposure, whereas excessive noise lowers the value of the platform signal.

\begin{figure}[t]
    \centering
    \includegraphics[width=0.98\textwidth]{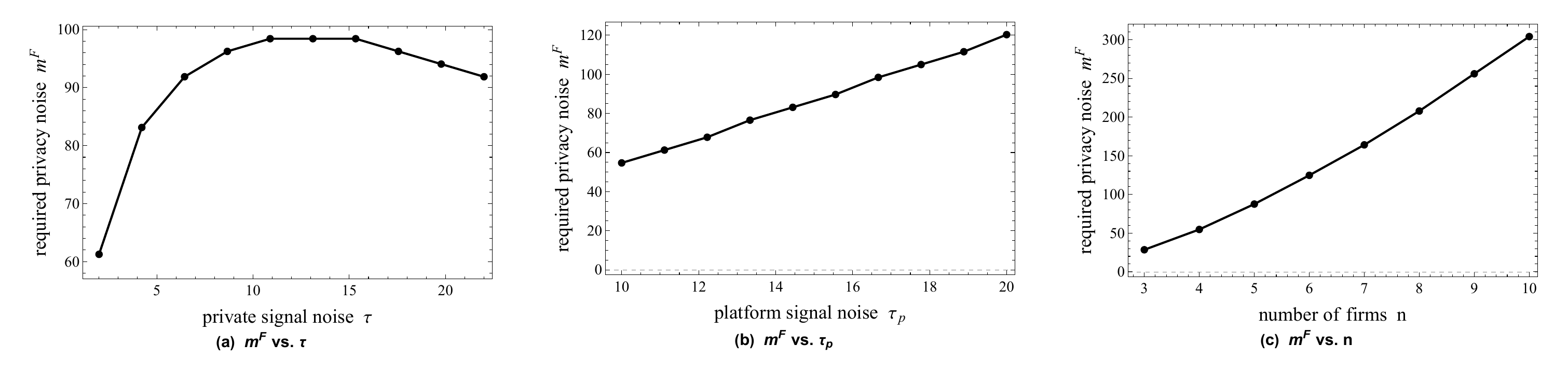}
    \caption{\textbf{Required privacy noise.}
    Each panel varies one primitive parameter and reports the least feasible privacy noise $m^F=\inf\{m\ge0:H(m)>0\}$ required to implement full participation. Panel (a) varies private signal noise $\tau$, Panel (b) varies platform signal noise $\tau_p$, and Panel (c) varies the number of firms $n$.}
    \label{fig:required_privacy}
\end{figure}

\Cref{fig:required_privacy} translates the implementation condition into the amount of privacy protection required by the platform. Panel (a) shows that the required privacy noise $m^F$ can vary non-monotonically with firms' private signal noise $\tau$. When private signals are moderately noisy, firms' reports remain strategically valuable to rivals while participation becomes harder to sustain, so stronger privacy protection is needed. When private signals become sufficiently noisy, however, the reports themselves carry less useful information, and the competitive externality from sharing weakens, reducing the required noise. Panel (b) shows that $m^F$ increases with platform signal noise $\tau_p$: as the external signal becomes less precise, the shared posterior relies more heavily on firms' reports, so the platform must add more privacy noise to limit strategic exposure. Panel (c) shows that implementation becomes more demanding as the number of firms increases, because each participating firm's report can benefit more competitors. Overall, the figure illustrates that the required privacy level is shaped jointly by private signal precision, platform-signal accuracy, and market size.

\begin{figure}[t]
    \centering
    \includegraphics[width=0.98\textwidth]{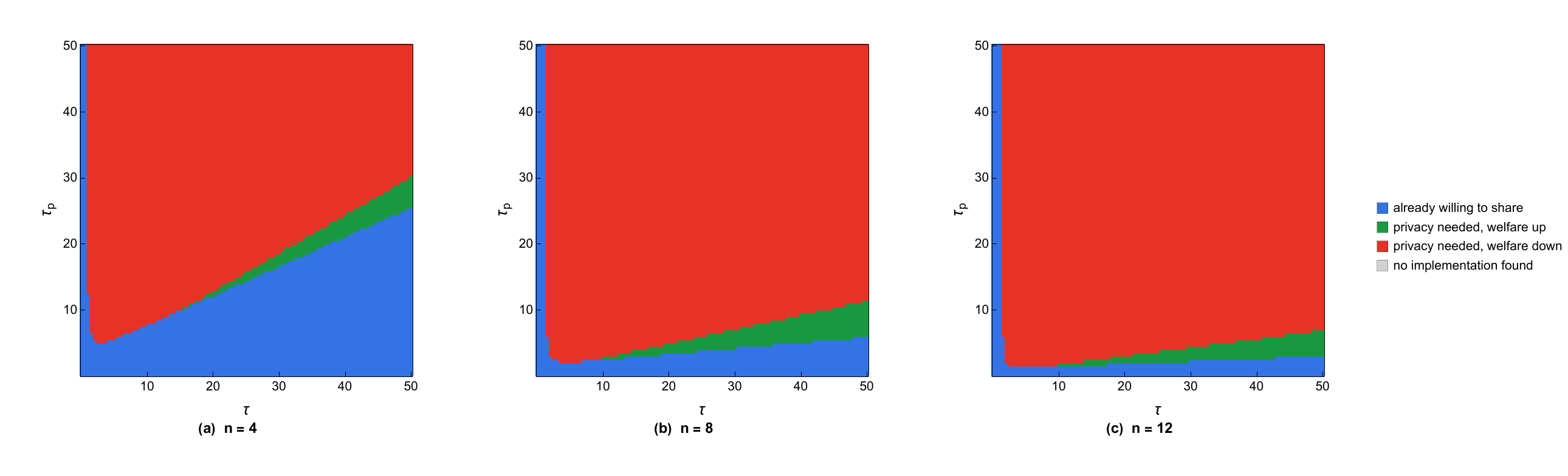}
    \caption{\textbf{Implementation and welfare regions.}
    Each panel classifies points in the $(\tau,\tau_p)$ parameter space for a fixed market size $n$. The colors indicate whether firms are already willing to share without privacy noise, whether privacy noise is needed and raises total surplus, whether privacy noise is needed but lowers total surplus, or whether no implementing privacy level is found within the numerical search range.}
    \label{fig:welfare_regions}
\end{figure}

\Cref{fig:welfare_regions} connects implementation with welfare. The blue region corresponds to parameters where full participation is already incentive compatible. The green region shows where privacy noise is needed and successfully induces welfare-improving sharing. The red region shows that implementation alone is not sufficient for welfare improvement: privacy may induce participation, but the required noise can reduce posterior precision enough that total surplus falls relative to no sharing. Thus, the platform must evaluate both incentive compatibility and welfare when choosing the privacy mechanism.

\section{Conclusion}

This paper studies voluntary information sharing among competing firms on digital platforms. We show that information sharing involves a fundamental trade-off: better inference about market uncertainty can improve decisions, but shared information also aligns firms' beliefs and intensifies strategic coupling in Cournot competition. 

Our analysis yields three main insights. First, in symmetric duopoly, privacy protection alone cannot sustain sharing without external platform information; a sufficiently informative platform signal can restore participation by providing informational value independent of rivals' reports. Second, under reciprocal access \citep{kirby1988trade}, an $n$-firm market may sustain sharing even without privacy protection or external information, because opting out means losing access to the pooled signal generated by other participants. Third, when reciprocal access or platform information alone is insufficient, calibrated privacy noise can restore participation by reducing strategic exposure while preserving informational value through the platform signal.

These results clarify the platform's design problem. To sustain full sharing as the unique participation equilibrium, the platform must make joining profitable at every level of existing participation and thereby rule out partial-participation equilibria. Privacy noise can help satisfy these constraints, but it also reduces posterior precision. Hence, the platform prefers the least distortionary feasible noise level, and privacy-induced sharing improves welfare only when participation gains outweigh informational losses. Overall, privacy is a strategic design instrument whose effectiveness depends on platform information, market size, and firms' signal quality.

Several limitations suggest directions for future research. One direction is to allow firms' demand signals to be correlated through overlapping customers, common shocks, or shared supply-chain conditions. Another is to extend the model to product differentiation, network competition, or localized substitution, where privacy and access rules may need to account for asymmetric competitive relationships. Together, these extensions would broaden the scope of privacy-preserving information sharing from symmetric Cournot markets to richer platform environments with correlated information and heterogeneous competitive relationships.

\clearpage
\bibliographystyle{informs2014}
\bibliography{refs}


\begin{APPENDICES}
\makeatletter
\renewcommand*{\theHsection}{appendix.\Alph{section}}
\renewcommand*{\theHsubsection}{appendix.\Alph{section}.\arabic{subsection}}
\renewcommand*{\theHsubsubsection}{appendix.\Alph{section}.\arabic{subsection}.\arabic{subsubsection}}
\makeatother
\renewcommand\thefigure{\thesection.\arabic{figure}}
\setcounter{figure}{0}
\phantomsection

\clearpage
\section*{Appendix Table of Contents}
\phantomsection\label{app:appendix-table-of-contents}

\begingroup
\small
\setlength{\parskip}{2pt}

\newcommand{\appendixtocsection}[3]{%
  \noindent\hyperref[#1]{\textbf{#2\quad #3}}\dotfill\pageref*{#1}\par
}

\newcommand{\appendixtocsubsection}[3]{%
  \noindent\hspace*{1.5em}\hyperref[#1]{#2\quad #3}\dotfill\pageref*{#1}\par
}

\newcommand{\appendixtocspace}{\vspace{0.25em}}

\appendixtocsection{app:a:proofs-of-main-results}{A}{Proofs of Main Results}
\appendixtocspace
\appendixtocsubsection{app:proof:prop1:negative-result}{A.1}{Proposition~\ref*{prop1:negative-result}: No sharing equilibria for two firms without external signals}
\appendixtocsubsection{app:proof:prop:platform-signal-2}{A.2}{Proposition~\ref*{prop:platform-signal-2}: Information-sharing equilibrium for two firms with platform signal}
\appendixtocsubsection{app:proof:prop:privacy-restores-sharing-2}{A.3}{Proposition~\ref*{prop:privacy-restores-sharing-2}: Privacy noise induces voluntary sharing for two firms}
\appendixtocsubsection{app:proof:prop:sharing_equilibrium_noprivacy}{A.4}{Proposition~\ref*{prop:sharing_equilibrium_noprivacy}: $n$-Firm sharing equilibrium without privacy protection}
\appendixtocsubsection{app:proof:cor:threshold_sharing_monotonicity}{A.5}{Corollary~\ref*{cor:threshold_sharing_monotonicity}: Characterizing the threshold for $n$-Firm sharing equilibrium}
\appendixtocsubsection{app:proof:prop:platform-signal-n}{A.6}{Proposition~\ref*{prop:platform-signal-n}: $n$-Firm information-sharing equilibrium with platform signal}
\appendixtocsubsection{app:proof:prop:privacy-restores-sharing-n}{A.7}{Proposition~\ref*{prop:privacy-restores-sharing-n}: Privacy noise induces voluntary sharing among $n$ firms}
\appendixtocsubsection{app:proof:prop:privacy-noise-welfare}{A.8}{Proposition~\ref*{prop:privacy-noise-welfare}: Effect of privacy noise}
\appendixtocsubsection{app:proof:prop:platform-signal-noise-welfare}{A.9}{Proposition~\ref*{prop:platform-signal-noise-welfare}: Effect of platform signal noise}
\appendixtocsubsection{app:prop:full-sharing-implementation}{A.10}{Proposition~\ref*{prop:full-sharing-implementation}: Full-sharing implementation}

\appendixtocsection{sec:heterogeneous-sequential}{B}{Information Sharing with Heterogeneous Signals and Under Sequential Entry}
\appendixtocspace
\appendixtocsubsection{app:b-1:information-sharing-with-heterogeneous-private-signals}{B.1}{Information Sharing with Heterogeneous Private Signals}
\appendixtocsubsection{app:b-2:heterogeneous-signals-under-sequential-participation}{B.2}{Heterogeneous Signals under Sequential Participation}

\appendixtocsection{app:welfare-derivations}{C}{Additional Welfare Analysis}
\appendixtocspace
\appendixtocsubsection{app:surplus-derivations}{C.1}{Surplus Derivations}
\appendixtocsubsection{app:welfare-comparative-statics}{C.2}{Additional Comparative Statics}

\appendixtocsection{app:d:proofs-of-additional-results}{D}{Proofs of Additional Results}
\appendixtocspace
\appendixtocsubsection{app:proof:prop:demand-uncertainty-welfare}{D.1}{Proposition~\ref*{prop:demand-uncertainty-welfare}: Effect of demand uncertainty}
\appendixtocsubsection{app:proof:prop:private-signal-noise-welfare}{D.2}{Proposition~\ref*{prop:private-signal-noise-welfare}: Effect of private signal noise}
\appendixtocsubsection{app:proof:prop:number-of-firms-welfare}{D.3}{Proposition~\ref*{prop:number-of-firms-welfare}: Effect of the number of firms}
\appendixtocsubsection{app:proof:prop:privacy-restores-sharing-hetero}{D.4}{Proposition~\ref*{prop:privacy-restores-sharing-hetero}: Privacy restores voluntary sharing under heterogeneous signals}
\appendixtocsubsection{app:proof:prop:sequential-heterogeneous}{D.5}{Proposition~\ref*{prop:sequential-heterogeneous}: Privacy-induced participation under sequential entry}

\endgroup
\medskip

\crefalias{section}{appendix}
\crefalias{subsection}{appendix}
\section{Proofs of Main Results}\label{app:a:proofs-of-main-results}

\subsection{\texorpdfstring{\Cref{prop1:negative-result}}{prop1negativeresult}: No sharing equilibria for two firms without external signals} \label{app:proof:prop1:negative-result}

\begin{proof}
We analyze the unilateral deviation incentive of firm~1.

When firm~1 participates (and there is no platform-provided external signal, i.e., $\tau_p\to\infty$), the platform releases the posterior mean
\[
\mu
=
\mathbb{E}\!\left[\theta \,\middle|\, \tilde{s}_1,\tilde{s}_2\right],
\qquad
\tilde{s}_i = s_i + v_i,
\quad
v_i \sim \mathcal{N}(0,m_i).
\]

Under symmetry, we restrict attention to $m_1=m_2=m$.
In this case, the posterior mean $\mu$ is a linear function of the aggregate statistic
\[
\tilde{s}
=
\frac{1}{2}\sum_{i=1}^{2}\tilde{s}_i.
\]
Since the mapping between $\mu$ and $\tilde{s}$ is linear and one-to-one (with coefficients depending only on parameters), conditioning on $\mu$ is equivalent to conditioning on $\tilde{s}$.
Therefore, without loss of generality, we may write firm~1's information set under participation as
\[
\mathcal{I}_1^{S}
=
\{s_1,\tilde{s}\}.
\]

\paragraph{Sharing Equilibrium.}

Because $(\theta,s_1,s_2,\tilde{s})$ are jointly Gaussian and the profit function is quadratic, the Bayesian Nash equilibrium in the quantity stage is linear in the sufficient statistics. Hence, for $i\in\{1,2\}$,
\[
q_i^{S}=c_1+d_1 s_i+e_1 \tilde{s}.
\]

Firm~1’s first-order condition under inverse demand $p=a+\theta-q_1-q_2$ is
\[
q_1^{S}
=\frac{1}{2}\Big(a+\mathbb{E}[\theta\mid s_1,\tilde{s}]
-\mathbb{E}[q_2^{S}\mid s_1,\tilde{s}]\Big).
\]
Substituting $q_2^{S}=c_1+d_1 s_2+e_1 \tilde{s}$ yields
\begin{equation}\label{eq:FOC-participate}
c_1+d_1 s_1+e_1 \tilde{s}
=\frac{1}{2}\Big(
a-\mathbb{E}[c_1+d_1 s_2+e_1 \tilde{s}\mid s_1,\tilde{s}]
+\mathbb{E}[\theta\mid s_1,\tilde{s}]
\Big).
\end{equation}

Let
\[
Y=\begin{bmatrix}s_1\\ \tilde{s}\end{bmatrix}.
\]
Then
\[
\mathbb{E}[s_2\mid s_1,\tilde{s}]
=\Sigma_{s_2,Y}\Sigma_{YY}^{-1}Y,
\qquad
\mathbb{E}[\theta\mid s_1,\tilde{s}]
=\Sigma_{\theta,Y}\Sigma_{YY}^{-1}Y,
\]
with
\[
\Sigma_{YY}=
\begin{bmatrix}
u+\tau & u+\dfrac{\tau}{2}\\[0.4em]
u+\dfrac{\tau}{2} & u+\dfrac{\tau}{2}+\dfrac{m}{2}
\end{bmatrix},\quad
\Sigma_{s_2,Y}=
\begin{bmatrix}
u & u+\dfrac{\tau}{2}
\end{bmatrix},\quad
\Sigma_{\theta,Y}=
\begin{bmatrix}
u & u
\end{bmatrix}.
\]

Substituting these conditional expectations into \eqref{eq:FOC-participate} and matching coefficients gives
\[
c_1=\frac{a}{3},\qquad
d_1=\frac{2 m u}{6 m u+\tau^2+2\tau u+4 m \tau},\qquad
e_1=\frac{2 \tau u}{3\left(6 m u+\tau^2+2\tau u+4 m \tau\right)}.
\]

Substituting the equilibrium strategies into the profit function yields firm~1’s ex-ante expected profit under participation:
\begin{equation}\label{eq:profit-participate}
\mathbb{E}\!\left[\pi_1^{S}\right]
=
\frac{1}{9}\left(
a^2+
\frac{
2 u^2\left(18 m^2(\tau+u)+\tau^2(\tau+2u)+m\tau(7\tau+12u)\right)
}{
\left(4 m\tau+\tau^2+6 m u+2\tau u\right)^2
}
\right).
\end{equation}

\paragraph{Uniliteral Deviation.}

If firm~1 deviates and does not participate, its information set is
\[
\mathcal{I}_1^{NS}=\{s_1\}.
\]
Since there are only two firms and no external signal, information aggregation does not occur. 
Accordingly, firm~2 must also base its decision solely on its own private signal $s_2$.

We look for a linear strategy
\[
q_1^{NS}=c_0+d_0 s_1.
\]
The first-order condition is
\[
q_1^{NS}
=\frac{1}{2}\Big(a+\mathbb{E}[\theta\mid s_1]-\mathbb{E}[q_2\mid s_1]\Big).
\]
Gaussian updating implies
\[
\mathbb{E}[s_2\mid s_1]=\frac{u}{u+\tau}s_1,
\qquad
\mathbb{E}[\theta\mid s_1]=\frac{u}{u+\tau}s_1.
\]
Matching coefficients yields
\[
c_0=\frac{a}{3},\qquad
d_0=\frac{u}{3u+2\tau},
\]
and hence
\[
q_1^{NS}=\frac{a}{3}+\frac{u}{3u+2\tau}s_1.
\]

Substituting into the profit function gives firm~1’s ex-ante expected profit under non-participation:
\begin{equation}\label{eq:profit-deviate}
\mathbb{E}\!\left[\pi_1^{NS}\right]
=
\frac{a^2}{9}+\frac{u^2(\tau+u)}{(2\tau+3u)^2}.
\end{equation}

Define the unilateral participation gain:
\[
\Delta
=
\mathbb{E}\!\left[\pi_1^{S}\right]
-
\mathbb{E}\!\left[\pi_1^{NS}\right].
\]
Using \eqref{eq:profit-participate} and \eqref{eq:profit-deviate}, direct algebra yields
\[
\Delta
=
\frac{\tau^2 u^2\left(-2m(2\tau+3u)(4\tau+9u)-\tau(\tau+2u)(\tau+3u)\right)}
{9(2\tau+3u)^2\left(2u(3m+\tau)+\tau(4m+\tau)\right)^2}
<0,
\qquad \forall\, m\ge 0.
\]

Therefore, participation is strictly dominated by non-participation when the other firm participates, for any privacy noise level $m\ge 0$. Hence, full participation cannot be a Nash equilibrium under any privacy-preserving mechanism of the stated form. This proves the proposition.

\end{proof}

\subsection{\texorpdfstring{\Cref{prop:platform-signal-2}}{propplatformsignal2}: Information-sharing equilibrium for two firms with platform signal}\label{app:proof:prop:platform-signal-2}

\begin{proof}

We analyze the unilateral incentive of firm~$1$.

For any participating set $J\subseteq\{1,2\}$, the platform releases the posterior mean
\[
\mu_J := \mathbb{E}\!\left[\theta \,\middle|\, \{s_i\}_{i\in J},\, s_p\right].
\]
Under joint Gaussianity, $m_J$ admits the precision-weighted form
\begin{equation*}
\label{eq:posterior-mean-J-2}
\mu_J
=
V_J\Big(\frac{1}{\tau}\sum_{i\in J} s_i + \frac{1}{\tau_p}s_p\Big),
\qquad
V_J
:=
\Big(\frac{1}{u}+\frac{|J|}{\tau}+\frac{1}{\tau_p}\Big)^{-1},
\end{equation*}
where $V_J=\Var(\theta\mid \{s_i\}_{i\in J},s_p)$ is the (deterministic) posterior variance.

Moreover, since $\mu_J=\mathbb{E}[\theta\mid \{s_i\}_{i\in J},s_p]$, the law of total variance yields
\begin{equation*}
\label{eq:var-mJ-2}
\Var(\mu_J)
=
u - V_J.
\end{equation*}

In particular, under full participation $J=\{1,2\}$,
\begin{equation*}
\label{eq:var-m12}
\Var(\mu_{2}) = u - V_{2},
\qquad
V_{2}=\Big(\frac{1}{u}+\frac{2}{\tau}+\frac{1}{\tau_p}\Big)^{-1}.
\end{equation*}
Firm~$1$'s private posterior mean is
\begin{equation}
\label{eq:m1-def-2}
\mu_1:=\mathbb{E}[\theta\mid s_1]=\frac{u}{u+\tau}\,s_1,
\qquad
\Var(\mu_1)=\frac{u^2}{u+\tau}.
\end{equation}

\paragraph{Sharing Equilibrium.}
When both firms participate, the platform releases $\mu_{2}=\mathbb E[\theta\mid s_1,s_2,s_p]$.
Given $\mu_{2}$, the posterior mean of $\theta$ is exactly $\mu_{2}$; hence each firm behaves as if demand intercept were $a+\mu_{2}$.
A symmetric best response implies
\begin{equation*}
\label{eq:q-joinall-2}
q_i^{S} \;=\; \frac{a+\mu_{2}}{3},\qquad i\in\{1,2\}.
\end{equation*}
As in the $n$-firm case, conditional expected profit at the best response equals the square of the chosen quantity:
\[
\mathbb E[\pi_i\mid \mu_{2}] = (q_i^{J})^2.
\]
Therefore the ex-ante expected profit is
\begin{equation}
\label{eq:profit-joinall-2}
\mathbb E[\pi_1^{S}]
=\mathbb E[(q_1^{S})^2]
=\frac{a^2}{9}+\frac{\Var(\mu_{2})}{9}.
\end{equation}

\paragraph{Uniliteral Deviation. }
Suppose firm~$1$ does not participate, so its information set is $\mathcal I_1^{NS}=\{s_1\}$ and its sufficient statistic is $m_1$ in \eqref{eq:m1-def-2}.
Firm~$2$ participates; the platform releases
\begin{equation}
\label{eq:m2}
\mu_2:=\mathbb E[\theta\mid s_2,s_p]
=V_2\Big(\frac{1}{\tau}s_2+\frac{1}{\tau_p}s_p\Big),
\qquad
V_2:=\Big(\frac{1}{u}+\frac{1}{\tau}+\frac{1}{\tau_p}\Big)^{-1}.
\end{equation}
We look for a linear Bayesian Nash equilibrium of the form
\begin{equation}
\label{eq:linear-NS-2}
q_1^{NS}=\alpha_0+\beta_0 \mu_1,
\qquad
q_2^{S}=\alpha_1+\beta_1 \mu_2.
\end{equation}

\smallskip
\emph{Key regressions.}
Because $(\mu_1,\mu_2)$ are jointly Gaussian with mean $0$, the conditional expectations are linear:
\begin{equation}
\label{eq:regressions-2}
\mathbb E[\mu_2\mid \mu_1]=\delta\, \mu_1,
\qquad
\mathbb E[\mu_1\mid \mu_2]=\gamma\, \mu_2,
\end{equation}
where
\begin{equation*}
\label{eq:delta-gamma-2}
\delta=\frac{\Cov(\mu_2,\mu_1)}{\Var(\mu_1)},
\qquad
\gamma=\frac{\Cov(\mu_2,\mu_1)}{\Var(\mu_2)}.
\end{equation*}
A direct covariance calculation gives the closed form for $\delta$.
Let
\begin{equation*}
\label{eq:A-def-2}
A(\tau_p):=\frac{1}{\tau}+\frac{1}{\tau_p}.
\end{equation*}
Then
\begin{equation*}
\label{eq:cov-2}
\Cov(\mu_2,\mu_1)
=\Cov\Big(\mathbb E[\theta\mid s_2,s_p],\ \mathbb E[\theta\mid s_1]\Big)
=\frac{u^2}{u+\tau}\,V_2\,A(\tau_p),
\end{equation*}
and therefore, using $\Var(\mu_1)=u^2/(u+\tau)$,
\begin{equation*}
\label{eq:delta-closed-2}
\delta = V_2\,A(\tau_p).
\end{equation*}
Also, using $\Var(\mu_2)=u-V_2$ from \eqref{eq:var-mJ-2},
\begin{equation*}
\label{eq:gamma-closed-2}
\gamma
=\frac{\Cov(\mu_2,\mu_1)}{\Var(\mu_2)}
=\frac{\frac{u^2}{u+\tau}V_2A(\tau_p)}{u-V_2}.
\end{equation*}

\smallskip
\emph{Solving for $(\alpha_0,\beta_0,\alpha_1,\beta_1)$.}
Firm~$1$'s best response to $q_2$ is
\[
q_1=\frac{1}{2}\Big(a+\mathbb E[\theta\mid \mu_1]-\mathbb E[q_2\mid \mu_1]\Big)
=\frac{1}{2}\Big(a+\mu_1-\mathbb E[q_2^P\mid \mu_1]\Big),
\]
and by \eqref{eq:linear-NS-2} and \eqref{eq:regressions-2},
\[
\mathbb E[q_2^P\mid \mu_1]=\alpha_1+\beta_1\mathbb E[\mu_2\mid \mu_1]=\alpha_1+\beta_1\delta\,\mu_1.
\]
Equating coefficients in $q_1^{NS}=\alpha_0+\beta_0 m_1$ yields
\begin{equation}
\label{eq:eqs-firm1-2}
\alpha_0=\frac{a-\alpha_1}{2},
\qquad
\beta_0=\frac{1-\beta_1\delta}{2}.
\end{equation}
Firm~$2$'s best response is
\[
q_2=\frac{1}{2}\Big(a+\mathbb E[\theta\mid m_2]-\mathbb E[q_1\mid m_2]\Big)
=\frac{1}{2}\Big(a+\mu_2-\mathbb E[q_1^{NS}\mid \mu_2]\Big),
\]
where
\[
\mathbb E[q_1^{NS}\mid \mu_2]=\alpha_0+\beta_0\mathbb E[\mu_1\mid \mu_2]
=\alpha_0+\beta_0\gamma\,\mu_2.
\]
Matching coefficients in $q_2^{S}=\alpha_1+\beta_1 m_2$ gives
\begin{equation}
\label{eq:eqs-firm2-2}
\alpha_1=\frac{a-\alpha_0}{2},
\qquad
\beta_1=\frac{1-\beta_0\gamma}{2}.
\end{equation}
Solving \eqref{eq:eqs-firm1-2}--\eqref{eq:eqs-firm2-2} yields $\alpha_0=\alpha_1=a/3$ and the closed form
\begin{equation}
\label{eq:beta0-closed-2}
\beta_0(\tau_p)
=
\frac{1-\frac{1}{2}\delta}{2-\frac{1}{2}\delta\,\gamma}
=
\frac{1-\frac12 V_2A(\tau_p)}
{2-\frac12\Big(V_2A(\tau_p)\Big)\cdot 
\frac{\frac{u^2}{u+\tau}V_2A(\tau_p)}{u-V_2} }.
\end{equation}

\smallskip

Using the same “best-response value” identity, conditional expected profit equals $(q_1^{NS})^2$; hence
\begin{equation}
\label{eq:profit-NS-2}
\mathbb E[\pi_1^{NS}]
=\mathbb E[(q_1^{NS})^2]
=\frac{a^2}{9}+\beta_0(\tau_p)^2\,\Var(\mu_1)
=
\frac{a^2}{9}+\beta_0(\tau_p)^2\frac{u^2}{u+\tau}.
\end{equation}

Define the unilateral gain from joining when the rival joins:
\[
\Delta_2(\tau_p):=\mathbb E[\pi_1^{J}]-\mathbb E[\pi_1^{NS}].
\]
Combining \eqref{eq:profit-joinall-2} and \eqref{eq:profit-NS-2} yields
\begin{equation}
\label{eq:Delta-2}
\Delta(\tau_p)
=
\frac{\Var(\mu_{2})}{9}
-
\beta_0(\tau_p)^2\,\frac{u^2}{u+\tau},
\qquad
\Var(\mu_{2})=u-\Big(\frac{1}{u}+\frac{2}{\tau}+\frac{1}{\tau_p}\Big)^{-1},
\end{equation}
with $\beta_0(\tau_p)$ given by \eqref{eq:beta0-closed-2} and $V_2$ defined in \eqref{eq:m2}.

The sharing profile is a symmetric equilibrium if and only if $\Delta_2(\tau_p)\ge 0$.
Define $\tau_p^*$ as the (unique) solution to $\Delta_2(\tau_p)=0$, which is exactly the indifference condition \eqref{eq:taup-star-equation-2}.
This equation provides an explicit computable characterization of $\tau_p^*$ (e.g., via a one-dimensional root search), completing the proof.

To see why the Sharing profile is \emph{unique}, note that when neither firm shares, each firm can still obtain the platform-provided signal $s_p$, and hence by unilaterally deviating to share it can additionally induce the platform to incorporate its private signal into the public posterior. This strictly increases the firm’s informational advantage relative to the no-sharing outcome (where only $s_p$ is available), and therefore yields a strictly higher expected payoff. Consequently, the No-Sharing profile cannot be a symmetric Bayesian Nash equilibrium whenever $\tau_p$ is sufficiently small, implying that Sharing is the only symmetric equilibrium in that regime.

\end{proof}

\subsection{\texorpdfstring{\Cref{prop:privacy-restores-sharing-2}}{propprivacyrestoressharing2}: Privacy noise induces voluntary sharing  for two firms}\label{app:proof:prop:privacy-restores-sharing-2}

\begin{proof}
We analyze the unilateral incentive of firm~$1$.

When a firm participates, the platform observes a privatized report
\[
y_i:=s_i+\xi_i=\theta+(\varepsilon_i+\xi_i),
\qquad
\varepsilon_i\sim \mathcal N(0,\tau),\ \xi_i\sim \mathcal N(0,m),
\]
so that $\varepsilon_i+\xi_i\sim \mathcal N(0,\tau+m)$.
If both firms participate, the platform releases the posterior mean
\[
t_{12}(m):=\mathbb{E}\!\left[\theta \,\middle|\, y_1,y_2,s_p\right].
\]
Under joint Gaussianity,
\begin{equation}
\label{eq:t12-form-proof}
t_{12}(m)
=
V_2(m)\left(
\frac{1}{\tau+m}(y_1+y_2)+\frac{1}{\tau_p}s_p
\right),
\qquad
V_2(m):=
\left(\frac{1}{u}+\frac{2}{\tau+m}+\frac{1}{\tau_p}\right)^{-1}.
\end{equation}
Moreover, by the law of total variance,
\begin{equation*}
\Var\!\big(t_{12}(m)\big)
=
u - V_2(m).
\end{equation*}

If firm~$1$ does not participate while firm~$2$ does, the platform releases
\begin{equation*}
t_{2}(m):=\mathbb{E}\!\left[\theta \,\middle|\, y_2,s_p\right]
=
V_1(m)\left(
\frac{1}{\tau+m}y_2+\frac{1}{\tau_p}s_p
\right),
\qquad
V_1(m):=
\left(\frac{1}{u}+\frac{1}{\tau+m}+\frac{1}{\tau_p}\right)^{-1}.
\end{equation*}

\paragraph{Sharing equilibrium.}
Suppose both firms participate.
Firm~$i$ observes $(s_i,t_{12}(m))$ and we look for a symmetric linear equilibrium
\[
q_i^{S}(m)=c+d\,s_i+e\,t_{12}(m),\qquad i=1,2.
\]
Let $t:=t_{12}(m)$.
Firm~$i$’s first-order condition under inverse demand $p=a+\theta-q_1-q_2$ is
\[
q_i^{S}(m)
=
\frac12\Big(a-\mathbb{E}[q_j^{S}(m)\mid s_i,t]+\mathbb{E}[\theta\mid s_i,t]\Big),
\qquad j\neq i.
\]
Using Gaussian projection, both conditional expectations are linear in $(s_i,t)$.
Matching coefficients yields a unique solution, with the constant term unchanged:
\begin{equation*}
c=\frac{a}{3}.
\end{equation*}
The remaining coefficients $(d,e)$ are uniquely determined as continuous functions of $(u,\tau,\tau_p,m)$.

Conditional expected profit equals the square of the chosen quantity, hence
\[
\mathbb{E}[\pi_1^{S}(m)]
=
\mathbb{E}\!\big[(q_1^{S}(m))^2\big]
=
c^2+d^2\Var(s_1)+e^2\Var(t)+2de\,\Cov(s_1,t),
\]
where $\Var(s_1)=u+\tau$, $\Var(t)=u-V_2(m)$, and $\Cov(s_1,t)$ follows from \eqref{eq:t12-form-proof}.

\paragraph{Unilateral deviation.}
Suppose firm~$1$ does not participate and uses
\[
q_1^{NS}(m)=c_0+d_0\,s_1.
\]
Firm~$2$ participates and observes $(s_2,t_2(m))$, using
\[
q_2^{S}(m)=c_1+d_1\,s_2+e_1\,t_2(m).
\]

Firm~$1$’s best response is
\[
q_1^{NS}
=
\frac12\Big(a+\mathbb{E}[\theta\mid s_1]-\mathbb{E}[q_2^{S}(m)\mid s_1]\Big).
\]
A direct calculation yields
\[
\mathbb{E}[t_2(m)\mid s_1]
=
\frac{\Cov(t_2(m),s_1)}{\Var(s_1)}\,s_1,
\]
and matching coefficients implies
\begin{equation*}
c_0=c_1=\frac{a}{3},
\end{equation*}
with $(d_0,d_1,e_1)$ uniquely determined as continuous functions of $(u,\tau,\tau_p,m)$.

As before,
\[
\mathbb{E}[\pi_1^{NS}(m)]
=
\mathbb{E}\!\big[(q_1^{NS}(m))^2\big]
=
\frac{a^2}{9}+d_0^2\,(u+\tau).
\]

Define
\[
\Delta(\tau_p,m)
:=
\mathbb{E}[\pi_1^{S}(m)]-\mathbb{E}[\pi_1^{NS}(m)].
\]
When $m=0$, the privacy-preserving platform coincides with the non-private platform.
Since $\tau_p>\tau_p^*(u,\tau)$, full participation is not an equilibrium without privacy, implying
\[
\Delta(\tau_p,0)<0.
\]
As $m\to\infty$, the precision of privatized reports vanishes and
\[
t_{12}(m)\ \xrightarrow[]{m\nearrow\infty}\ \mathbb{E}[\theta\mid s_p],
\qquad
t_{2}(m)\ \xrightarrow[]{m\nearrow\infty}\ \mathbb{E}[\theta\mid s_p],
\]
in $L^2$.
Thus, participation grants access to the informative public posterior $\mathbb{E}[\theta\mid s_p]$, while non-participation leaves firm~$1$ conditioning only on $s_1$.
Since $\tau_p<\infty$ implies $\Var(\mathbb{E}[\theta\mid s_p])>0$, this strictly raises firm~$1$’s optimized ex-ante payoff, yielding
\[
\lim_{m\to\infty}\Delta(\tau_p,m)>0.
\]
Finally, $\Delta(\tau_p,m)$ is continuous in $m$.
By the intermediate value theorem, there exists $m^*>0$ such that $\Delta(\tau_p,m^*)=0$, and any $m>m^*$ yields $\Delta(\tau_p,m)>0$.

For any $m>m^*$ we have $\Delta(\tau_p,m)>0$, so sharing is a strict best response when the opponent shares. 
If both firms do not share, a unilateral deviation to Sharing adds non-degenerate information (for any finite $m$) to the public posterior based on $s_p$, which strictly increases the deviator’s optimized ex-ante payoff. 
Hence No-Sharing cannot be an equilibrium. 
Therefore, for all $m>m^*$, Sharing is the unique symmetric Bayesian Nash equilibrium.

\end{proof}

\subsection{\texorpdfstring{\Cref{prop:sharing_equilibrium_noprivacy}}{propsharingequilibriumnoprivacy}: \texorpdfstring{$n$}{n}-Firm sharing equilibrium without privacy protection}\label{app:proof:prop:sharing_equilibrium_noprivacy}

\begin{proof}

Suppose first that all firms participate and privacy noise is symmetric, i.e., $m_j = m$ for all $j$. 
Each shared signal is perturbed as
\[
\tilde{s}_j = s_j + v_j,
\qquad 
v_j \sim \mathcal{N}(0,m).
\]

The platform releases the posterior mean
\[
\mu
=
\mathbb{E}\!\left[\theta \,\middle|\, \tilde{s}_1,\dots,\tilde{s}_n \right].
\]

Under symmetry, the posterior mean is a linear function of the aggregate statistic
\[
\tilde{s}
=
\frac{1}{n}\sum_{j=1}^{n}\tilde{s}_j.
\]

Because the mapping between $\mu$ and $\tilde{s}$ is linear and one-to-one 
(with coefficients depending only on model parameters), conditioning on $\mu$ 
is equivalent to conditioning on $\tilde{s}$. 
Hence, although the platform releases $\mu$, the statistic $\tilde{s}$ is sufficient, 
and the informational content is identical.

Therefore, without loss of generality, firm $i$’s information set under sharing can be written as
\[
\mathcal{I}_i^{S}
=
\{ s_i,\tilde{s} \}.
\]

If firm $i$ deviates while all other firms share, the platform aggregates only the remaining $n-1$ perturbed signals:
\[
\tilde{s}_{-i}
=
\frac{1}{n-1}\sum_{j\ne i}\tilde{s}_j.
\]

The platform releases
\[
\mu_{-i}
=
\mathbb{E}\!\left[\theta \,\middle|\, \{\tilde{s}_j\}_{j\ne i} \right],
\]
which is again a linear one-to-one transformation of $\tilde{s}_{-i}$.

Thus, the deviating firm observes only its private signal,
\[
\mathcal{I}_i^{D} = \{ s_i \},
\]
while each participating firm $j \ne i$ observes
\[
\{ s_j,\tilde{s}_{-i} \}.
\]

\paragraph{Sharing Equilibrium.}
Because $(\theta,s_1,\ldots,s_n,\bar s)$ are jointly Gaussian and profits are quadratic, a Bayesian Nash equilibrium is linear in the sufficient statistics.
We look for a symmetric linear strategy
\[
q_i^{S}=c_S+d_S s_i+e_S \bar s.
\]
Conditional on $\bar s$, the private signal $s_i$ only contains idiosyncratic noise about $\varepsilon_i$ and provides no additional information about the common state $\theta$.
Hence in equilibrium the coefficient on $s_i$ vanishes: $d_S=0$ and
\[
q_i^{S}=c_S+e_S \bar s.
\]
Imposing the Cournot first-order condition
\[
q_i^{S}=\frac12\Big(a+\mathbb E[\theta\mid s_i,\bar s]-\mathbb E[\sum_{j\ne i}q_j^{S}\mid s_i,\bar s]\Big),
\]
and matching coefficients yields
\[
c_S=\frac{a}{n+1},
\qquad
e_S=\frac{n u}{(n+1)(n u+\tau)}.
\]
Substituting these strategies into $\pi_i=(a+\theta-\sum_{j}q_j)q_i$ gives the ex-ante expected profit under full sharing:
\begin{equation*}
\mathbb E[\pi_i^{S}]
=
\frac{a^2 n u+a^2\tau+n u^2}{n^{3}u+n^{2}\tau+2n^{2}u+2n\tau+n u+\tau}.
\end{equation*}

\paragraph{Uniliteral Deviation.}
Now suppose firm $i$ deviates and does not share while all $n-1$ rivals share among themselves.
We consider linear strategies of the form
\[
q_i^{NS}=c_0+d_0 s_i,
\]
and for each participating rival $j\ne i$,
\[
q_j^{NS}=c_1+d_1 s_j+e_1 \bar s_{-i}.
\]
Conditional on $\bar s_{-i}$ the private signal $s_j$ only reveals idiosyncratic noise and does not improve inference about $\theta$; thus $d_1=0$ and insiders use only $\bar s_{-i}$:
\[
q_j^{NS}=c_1+e_1 \bar s_{-i}.
\]
Imposing the best-response (first-order) conditions for the deviator and for an insider and matching coefficients yields
\[
c_0=c_1=\frac{a}{n+1},\qquad
d_0=\frac{u\,(n\tau+n u-u)}{2n^{2}\tau u+n^{2}u^{2}+2n\tau^{2}-u^{2}},\qquad
e_1=\frac{u\,(n-1)(2\tau+u)}{2n^{2}\tau u+n^{2}u^{2}+2n\tau^{2}-u^{2}}.
\]
Substituting into the profit function gives the deviator's ex-ante expected profit $\mathbb E[\pi_i^{\text{deviate}}]$.

Define
\[
\Delta
=
\mathbb E[\pi_i^{S}]-\mathbb E[\pi_i^{NS}].
\]
Direct algebra yields a closed-form expression for the unilateral participation gain:
\[
\Delta
=
u^2\left(
\frac{n}{(n+1)^2\,(n u+\tau)}
-
\frac{(u+\tau)\big((n-1)u+n\tau\big)^2}
{\big((n^2-1)u^2+2n^2u\tau+2n\tau^2\big)^2}
\right).
\]

Since $u>0$, the sign of $\Delta$ is fully determined by the bracketed expression. Therefore,
\[
\Delta \ge 0
\quad \Longleftrightarrow \quad
\frac{n}{(n+1)^2\,(n u+\tau)}
\;\ge\;
\frac{(u+\tau)\big((n-1)u+n\tau\big)^2}
{\big((n^2-1)u^2+2n^2u\tau+2n\tau^2\big)^2}.
\]

Hence, full information sharing constitutes a Nash equilibrium if and only if the above inequality holds. 

\end{proof}

\subsection{\texorpdfstring{\Cref{cor:threshold_sharing_monotonicity}}{corthresholdsharingmonotonicity}: Characterizing the threshold for \texorpdfstring{$n$}{n}-Firm sharing equilibrium}\label{app:proof:cor:threshold_sharing_monotonicity}

We begin with the following helper lemma.

\begin{lemma}[Concavity and single-threshold structure]\label{lem:concavity_threshold}
Let $r=\tau/u>0$, and let $\Phi(n,r)$ denote the polynomial obtained by clearing denominators in the condition $\Delta\ge 0$, so that
\[
\Delta\ge 0 \quad \Longleftrightarrow \quad \Phi(n,r)\ge 0.
\]
Then, for all $n\ge 3$ and all $r>0$, the function $\Phi(n,r)$ is strictly concave in $r$, i.e.,
\[
\frac{\partial^2 \Phi(n,r)}{\partial r^2}<0.
\]
Moreover, for each fixed $n\ge 3$, $\Phi(n,r)$ satisfies
\[
\Phi(n,0)>0
\qquad\text{and}\qquad
\lim_{r\to\infty}\Phi(n,r)=-\infty.
\]
Consequently, for every $n\ge 3$, the equation $\Phi(n,r)=0$ admits a unique positive solution $r>0$, and
\[
\Phi(n,r)\ge 0 
\quad \Longleftrightarrow \quad 
0<r\le r_0(n),
\]
where $r_0(n)$ denotes this unique positive root.
Equivalently,
\[
\Delta\ge 0 
\quad \Longleftrightarrow \quad 
\frac{\tau}{u}\le r_0(n).
\]
\end{lemma}

\proof
By construction, the inequality $\Delta\ge 0$ is equivalent to $\Phi(n,r)\ge 0$, where $\Phi(n,r)$ is a polynomial in $r$ whose coefficients depend on $n$.

Direct differentiation yields
\[
\frac{\partial^2 \Phi(n,r)}{\partial r^2}
= -2n^4 + (4-6r)n^3 + (6r-6)n^2 - 4n .
\]
We show that this expression is strictly negative for all $n\ge 3$ and all $r>0$.
Rewrite it as
\[
\frac{\partial^2 \Phi(n,r)}{\partial r^2}
=
\big(-2n^4-4n\big)
+
\big((4-6r)n^3+(6r-6)n^2\big).
\]
The first term is strictly negative. For the second term, factor out $n^2$:
\[
(4-6r)n^3+(6r-6)n^2
= n^2\big[(4-6r)n+(6r-6)\big].
\]
For $n\ge 3$ and $r>0$, we have
\[
(4-6r)n+(6r-6)
= 4n-6 + r(6-6n) \le 4n-6,
\]
since $6-6n\le 0$ and $r>0$. Therefore,
\[
(4-6r)n^3+(6r-6)n^2 \le 4n^3-6n^2.
\]
It follows that
\[
\frac{\partial^2 \Phi(n,r)}{\partial r^2}
\le
-2n^4 + 4n^3 - 6n^2 - 4n
=
-2n\big(n^3-2n^2+3n+2\big).
\]
Let $g(n)=n^3-2n^2+3n+2$. Then $g'(n)=3n^2-4n+3$, whose discriminant is $(-4)^2-4\cdot 3\cdot 3=-20<0$, implying $g'(n)>0$ for all $n$. Hence $g(n)$ is strictly increasing and
\[
g(3)=27-18+9+2=20>0,
\]
so $g(n)>0$ for all $n\ge 3$. Since $-2n<0$ for $n\ge 3$, we conclude
\[
\frac{\partial^2 \Phi(n,r)}{\partial r^2}<0
\qquad \forall\, n\ge 3,\ r>0.
\]
Thus $\Phi(n,r)$ is strictly concave in $r$ on $(0,\infty)$.

Evaluating at $r=0$ yields $\Phi(n,0)>0$ for all $n\ge 3$. Moreover, since the leading term in $r$ is cubic with negative coefficient, it follows that
\[
\lim_{r\to\infty}\Phi(n,r)=-\infty.
\]

Strict concavity together with $\Phi(n,0)>0$ and $\lim_{r\to\infty}\Phi(n,r)=-\infty$ implies that for each fixed $n\ge 3$, the equation $\Phi(n,r)=0$ admits a unique positive solution $r>0$. Moreover,
\[
\Phi(n,r)\ge 0 \iff 0<r\le r_0(n),
\]
where $r_0(n)$ denotes this unique positive root.  
Rewriting in terms of the original parameters yields
\[
\Delta\ge 0 \iff \frac{\tau}{u}\le r_0(n).
\]
\endproof

We can now give the proof of \Cref{cor:threshold_sharing_monotonicity}.

\begin{proof} The existence and uniqueness of a positive solution $r>0$ to $\Phi(n,r)=0$ for each fixed $n\ge 3$, as well as the threshold characterization
\[
\Phi(n,r)\ge 0 \iff 0<r\le r_0(n)
\qquad\Longleftrightarrow\qquad
\Delta\ge 0 \iff \frac{\tau}{u}\le r_0(n),
\]
follow directly from Lemma~\ref{lem:concavity_threshold}. It remains to prove the monotonicity of $r_0(n)$ in $n$.

Since $\Phi$ is homogeneous of degree $3$ in $(\tau,u)$, substituting $r=\tau/u$ and dividing by $u^3>0$ yields an equivalent polynomial (with the same sign) of the form
\[
\Phi(n,r)=a_0(n)+a_1(n)r+a_2(n)r^2+a_3(n)r^3,
\]
where
\[
\begin{aligned}
a_0(n)&=(n-1)(n+1)\big((n-2)n-1\big)=n^4-2n^3-2n^2+2n+1,\\
a_1(n)&=n\big((n-3)n(n+1)-1\big)+1=n^4-2n^3-3n^2-n+1,\\
a_2(n)&=-n\big(n\big((n-2)n+3\big)+2\big)= -n^4+2n^3-3n^2-2n,\\
a_3(n)&=-(n-1)n^2<0 \qquad (n\ge 2).
\end{aligned}
\]
By Lemma~\ref{lem:concavity_threshold}, for each $n\ge 3$ there is a unique $r_0(n)>0$ such that
$\Phi\big(n,r_0(n)\big)=0$ and $\Phi(n,r)>0$ for $r\in(0,r_0(n))$, while $\Phi(n,r)<0$ for $r>r_0(n)$.

Let $\varphi:=\frac{1+\sqrt5}{2}$ be the golden ratio, so that $\varphi^2=\varphi+1$ and $\varphi^3=2\varphi+1$.
Evaluating $\Phi(n,r)$ at $r=\varphi$ and using these identities gives
\[
\Phi(n,\varphi)
=a_0+a_1\varphi+a_2\varphi^2+a_3\varphi^3
=\big(a_0+a_2+a_3\big)+\big(a_1+a_2+2a_3\big)\varphi .
\]
A direct simplification yields
\[
a_0(n)+a_2(n)+a_3(n)=1-n^3-4n^2<0,
\]
and
\[
a_1(n)+a_2(n)+2a_3(n)=1-2n^3-4n^2-3n<0,
\]
for every $n\ge 3$. Since $\varphi>0$, it follows that $\Phi(n,\varphi)<0$ for all $n\ge 3$.
Because $\Phi(n,0)=a_0(n)>0$ (Lemma~\ref{lem:concavity_threshold}) and $\Phi(n,\cdot)$ has a unique positive root,
we conclude that
\[
0<r_0(n)<\varphi \qquad\text{for all } n\ge 3.
\]

Fix $n\ge 3$ and consider the point $(n,r_0(n))$ where $\Phi(n,r_0(n))=0$.
Since $\Phi(n,\cdot)$ is strictly concave and changes sign from positive to negative at $r_0(n)$, the crossing is strict and hence
\[
\partial_r\Phi\big(n,r_0(n)\big)<0.
\]
(Indeed, if $\partial_r\Phi(n,r_0(n))\ge 0$, strict concavity would prevent $\Phi$ from becoming strictly negative for all $r>r_0(n)$
without creating an additional zero, contradicting uniqueness.)

We now show that
\[
\partial_n\Phi\big(n,r_0(n)\big)>0 \qquad\text{for all } n\ge 3.
\]
A direct differentiation gives
\[
\partial_n\Phi(n,r)=a_0'(n)+a_1'(n)r+a_2'(n)r^2+a_3'(n)r^3.
\]
On the zero set $\{\Phi(n,r)=0\}$ we may eliminate the cubic term by using
\[
r^3=-\frac{a_0+a_1r+a_2r^2}{a_3},
\]
which yields the identity (valid at any $(n,r)$ with $\Phi(n,r)=0$)
\[
\partial_n\Phi(n,r)
=\Big(a_0'-\frac{a_3'}{a_3}a_0\Big)
+\Big(a_1'-\frac{a_3'}{a_3}a_1\Big)r
+\Big(a_2'-\frac{a_3'}{a_3}a_2\Big)r^2.
\]
Define
\[
b_k(n):=a_k'(n)-\frac{a_3'(n)}{a_3(n)}a_k(n),\qquad k\in\{0,1,2\}.
\]
Then, on $\Phi(n,r)=0$,
\[
\partial_n\Phi(n,r)=b_0(n)+b_1(n)r+b_2(n)r^2.
\]
Since $a_3(n)=-(n-1)n^2$ and $a_3'(n)=-n(3n-2)$, we have
\[
\frac{a_3'(n)}{a_3(n)}=\frac{3n-2}{n(n-1)} \;>\;0 \qquad (n>1).
\]
A straightforward simplification using the explicit $a_k(n)$ yields
\[
b_0(n)=n^3-n^2+3n-1-\frac{2}{n}>0 \qquad (n\ge 3),
\]
\[
b_1(n)=\frac{n^5-2n^4+5n^3+2n^2-4n+2}{n(n-1)}>0 \qquad (n\ge 3),
\]
and
\[
b_2(n)=\frac{-n^4+2n^3+n^2+4n-2}{\,n-1\,}<0 \qquad (n\ge 3).
\]
Thus, for each fixed $n\ge 3$, the function
\[
Q_n(r):=b_0(n)+b_1(n)r+b_2(n)r^2
\]
is a strictly concave quadratic in $r$ (since $b_2(n)<0$).

At the root we have $0<r_0(n)<\varphi$. Therefore it suffices to show $Q_n(r)>0$ on the interval $[0,\varphi]$.
Because $Q_n$ is concave, its minimum over the compact interval $[0,\varphi]$ is attained at an endpoint, so
\[
\min_{r\in[0,\varphi]}Q_n(r)=\min\{Q_n(0),Q_n(\varphi)\}.
\]
We already have $Q_n(0)=b_0(n)>0$. It remains to verify $Q_n(\varphi)>0$.
Using $\varphi^2=\varphi+1$,
\[
Q_n(\varphi)=b_0(n)+b_1(n)\varphi+b_2(n)\varphi^2
=\big(b_0(n)+b_2(n)\big)+\big(b_1(n)+b_2(n)\big)\varphi.
\]
Substituting the explicit expressions for $b_0,b_1,b_2$ above and simplifying yields
\[
Q_n(\varphi)=\frac{(8+3\sqrt5)n^3+(3+3\sqrt5)n^2-(6+3\sqrt5)n+(3+\sqrt5)}{n(n-1)}.
\]
The denominator is positive for $n\ge 3$, and the numerator is a cubic polynomial with strictly positive leading coefficient
$8+3\sqrt5>0$ and positive $n^2$ coefficient as well; in particular it is strictly positive for all $n\ge 3$.
Hence $Q_n(\varphi)>0$, and consequently $Q_n(r)>0$ for all $r\in[0,\varphi]$.
Since $r_0(n)\in(0,\varphi)$, we obtain
\[
\partial_n\Phi\big(n,r_0(n)\big)=Q_n\big(r_0(n)\big)>0 \qquad (n\ge 3).
\]

By the implicit function theorem, in a neighborhood of each $n\ge 3$ the unique root $r_0(n)$ is differentiable and satisfies
\[
\frac{dr_0(n)}{dn}=-\frac{\partial_n\Phi\big(n,r_0(n)\big)}{\partial_r\Phi\big(n,r_0(n)\big)}.
\]
Since $\partial_r\Phi\big(n,r_0(n)\big)<0$ and $\partial_n\Phi\big(n,r_0(n)\big)>0$.
Therefore
\[
\frac{dr_0(n)}{dn}>0 \qquad\text{for all } n\ge 3,
\]
so $r_0(n)$ is strictly increasing in $n$.

Finally, since the equilibrium condition is $\tau/u\le r_0(n)$, the feasible region for full information sharing expands
monotonically with $n$ and is bounded by the implicit curve $\Phi\!\left(n,\tau/u\right)=0$.

\end{proof}

\subsection{\texorpdfstring{\Cref{prop:platform-signal-n}}{propplatformsignaln}: \texorpdfstring{$n$}{n}-Firm information-sharing equilibrium with platform signal}\label{app:proof:prop:platform-signal-n}

\begin{proof}

We analyze the unilateral incentive of firm~$1$.

Fix any participating set $J$. By the platform rule in Section~\ref{sec:problem}, the released statistic is the posterior mean
\[
\mu_J := \mathbb{E}\!\left[\theta \,\middle|\, \{s_i\}_{i\in J},\, s_p\right].
\]
Under joint Gaussianity, $\mu_J$ is linear in the signals and can be written in the precision-weighted form
\begin{equation*}
\mu_J
=
V_J\Big(\frac{1}{\tau}\sum_{i\in J} s_i + \frac{1}{\tau_p}s_p\Big),
\qquad
V_J
:=
\Big(\frac{1}{u}+\frac{|J|}{\tau}+\frac{1}{\tau_p}\Big)^{-1},
\end{equation*}
where $V_J$ is the posterior variance $\Var(\theta\mid \{s_i\}_{i\in J},s_p)$.

Moreover, since $\mu_J=\mathbb{E}[\theta\mid \{s_i\}_{i\in J},s_p]$, the law of total variance yields
\begin{equation*}
\Var(\mu_J)
=
\Var\!\Big(\mathbb{E}[\theta\mid \{s_i\}_{i\in J},s_p]\Big)
=
\Var(\theta)-\mathbb{E}\!\big[\Var(\theta\mid \{s_i\}_{i\in J},s_p)\big]
=
u - V_J,
\end{equation*}
where the last equality uses $\Var(\theta)=u$ and the fact that, in the Gaussian signal model, the posterior variance
$\Var(\theta\mid \{s_i\}_{i\in J},s_p)=V_J$ is deterministic (i.e., does not depend on the realized signals).

In particular, under full participation $J=[n]$,
\begin{equation*}
\Var(\mu_{n}) = u - V_{n},
\qquad
V_{n}=\Big(\frac{1}{u}+\frac{n}{\tau}+\frac{1}{\tau_p}\Big)^{-1}.
\end{equation*}
Firm~$1$'s private posterior mean is
\begin{equation}
\label{eq:m1-def}
\mu_1:=\mathbb{E}[\theta\mid s_1]=\frac{u}{u+\tau}\,s_1,
\qquad
\Var(\mu_1)=\frac{u^2}{u+\tau}.
\end{equation}

\paragraph{Sharing Equilibrium.}
When all firms participate, the platform releases $m_{[n]}=\mathbb E[\theta\mid s_1,\dots,s_n,s_p]$.
Given $m_{[n]}$, the posterior mean of $\theta$ is exactly $m_{[n]}$; hence each firm behaves as if demand intercept were $a+m_{[n]}$.
A symmetric best response implies
\begin{equation*}
q_i^{S} \;=\; \frac{a+\mu_{n}}{n+1},\qquad i=1,\dots,n.
\end{equation*}
Under linear demand and risk-neutrality, conditional expected profit at the best response equals the square of the chosen quantity:
\[
\mathbb E[\pi_i\mid m_{[n]}] 
= \max_{q_i}\ \mathbb E[(a+\theta-\sum_j q_j)q_i\mid m_{[n]}]
= (q_i^S)^2.
\]
Therefore the ex-ante expected profit is
\begin{equation}
\label{eq:profit-joinall}
\mathbb E[\pi_1^{S}]
=\mathbb E[(q_1^S)^2]
=\frac{a^2}{(n+1)^2}+\frac{\Var(\mu_{n})}{(n+1)^2}.
\end{equation}

\paragraph{Unilateral deviation }
Suppose firm $1$ does not participate, so its information set is $\mathcal I_1^{NS}=\{s_1\}$ and its sufficient statistic is $\mu_1$ in \eqref{eq:m1-def}.
The remaining $n-1$ firms participate; the platform releases
\begin{equation}
\label{eq:m-minus1}
\mu_{-1}:=\mathbb E[\theta\mid s_2,\dots,s_n,s_p]
=V_{-1}\Big(\frac{1}{\tau}\sum_{i=2}^n s_i+\frac{1}{\tau_p}s_p\Big),
\qquad
V_{-1}:=\Big(\frac{1}{u}+\frac{n-1}{\tau}+\frac{1}{\tau_p}\Big)^{-1}.
\end{equation}
We look for a linear Bayesian Nash equilibrium of the form
\begin{equation}
\label{eq:linear-NS}
q_1^{NS}=\alpha_0+\beta_0 \mu_1,
\qquad
q_j^{S}=\alpha_1+\beta_1 \mu_{-1},\ \ j=2,\dots,n.
\end{equation}
(Again $\mu_{-1}$ is sufficient for $\theta$ given participating information, so $q_j^P$ depends only on $m_{-1}$.)

\smallskip
\emph{Key regressions.}
Because $(m_1,m_{-1})$ are jointly Gaussian with mean $0$, the conditional expectations are linear:
\begin{equation}
\label{eq:regressions}
\mathbb E[\mu_{-1}\mid \mu_1]=\delta\, \mu_1,
\qquad
\mathbb E[\mu_1\mid \mu_{-1}]=\gamma\, \mu_{-1},
\end{equation}
where
\begin{equation*}
\label{eq:delta-gamma}
\delta=\frac{\Cov(\mu_{-1},\mu_1)}{\Var(\mu_1)},
\qquad
\gamma=\frac{\Cov(\mu_{-1},\mu_1)}{\Var(\mu_{-1})}.
\end{equation*}
A direct covariance calculation using \eqref{eq:m1-def}--\eqref{eq:m-minus1} gives a simple closed form for $\delta$.
Let
\begin{equation*}
\label{eq:A-def}
A(\tau_p):=\frac{n-1}{\tau}+\frac{1}{\tau_p}.
\end{equation*}
Then
\begin{equation*}
\label{eq:cov}
\Cov(\mu_{-1},\mu_1)=\Cov\Big(\mathbb E[\theta\mid s_2,\dots,s_n,s_p],\ \mathbb E[\theta\mid s_1]\Big)
=\frac{u^2}{u+\tau}\,V_{-1}\,A(\tau_p),
\end{equation*}
and therefore, using $\Var(\mu_1)=u^2/(u+\tau)$,
\begin{equation}
\label{eq:delta-closed}
\delta = V_{-1}\,A(\tau_p).
\end{equation}
Also, using $\Var(\mu_{-1})=u-V_{-1}$ from \eqref{eq:m-minus1},
\begin{equation}
\label{eq:gamma-closed}
\gamma
=\frac{\Cov(\mu_{-1},\mu_1)}{\Var(\mu_{-1})}
=\frac{\frac{u^2}{u+\tau}V_{-1}A(\tau_p)}{u-V_{-1}}.
\end{equation}

\smallskip
\emph{Solving for $(\alpha_0,\beta_0,\alpha_1,\beta_1)$.}
Firm $1$'s best response to opponents' total quantity $Q_{-1}$ is
\[
q_1=\frac{1}{2}\Big(a+\mathbb E[\theta\mid \mu_1]-\mathbb E[Q_{-1}\mid \mu_1]\Big)
=\frac{1}{2}\Big(a+\mu_1-(n-1)\mathbb E[q_2^S\mid \mu_1]\Big),
\]
and by \eqref{eq:linear-NS} and \eqref{eq:regressions},
\[
\mathbb E[q_2^S\mid \mu_1]=\alpha_1+\beta_1 \mathbb E[\mu_{-1}\mid \mu_1]=\alpha_1+\beta_1\delta \mu_1.
\]
Equating coefficients in $q_1^{NS}=\alpha_0+\beta_0 \mu_1$ yields
\begin{equation}
\label{eq:eqs-firm1}
\alpha_0=\frac{a-(n-1)\alpha_1}{2},
\qquad
\beta_0=\frac{1-(n-1)\beta_1\delta}{2}.
\end{equation}
For any participating firm $j\ge 2$, the best response condition is
\[
q_j=\frac{1}{2}\Big(a+\mathbb E[\theta\mid \mu_{-1}] - \mathbb E[q_1^{NS}+(n-2)q_2^P\mid \mu_{-1}]\Big)
=\frac{1}{2}\Big(a+\mu_{-1}-\mathbb E[q_1^{NS}\mid \mu_{-1}]-(n-2)q_2^P\Big),
\]
where
\[
\mathbb E[q_1^{NS}\mid \mu_{-1}]=\alpha_0+\beta_0\mathbb E[\mu_1\mid \mu_{-1}]
=\alpha_0+\beta_0\gamma \mu_{-1}.
\]
Matching coefficients in $q_j^P=\alpha_1+\beta_1 \mu_{-1}$ gives
\begin{equation}
\label{eq:eqs-firm2}
\alpha_1=\frac{a-\alpha_0}{n},
\qquad
\beta_1=\frac{1-\beta_0\gamma}{n}.
\end{equation}
Solving \eqref{eq:eqs-firm1}--\eqref{eq:eqs-firm2} yields $\alpha_0=\alpha_1=a/(n+1)$ and
\begin{equation}
\label{eq:beta0-closed}
\beta_0(\tau_p)
=
\frac{1-\frac{n-1}{n}\delta}{2-\frac{n-1}{n}\delta\,\gamma}
=
\frac{1-\frac{n-1}{n}V_{-1}A(\tau_p)}
{2-\frac{n-1}{n}\Big(V_{-1}A(\tau_p)\Big)\cdot 
\frac{\frac{u^2}{u+\tau}V_{-1}A(\tau_p)}{u-V_{-1}} }.
\end{equation}
(Here $\delta,\gamma$ are given explicitly by \eqref{eq:delta-closed}--\eqref{eq:gamma-closed}.)

Using the same best-response identity, conditional expected profit equals $(q_1^{NS})^2$; hence
\begin{equation}
\label{eq:profit-NS}
\mathbb E[\pi_1^{NS}]
=\mathbb E[(q_1^{NS})^2]
=\frac{a^2}{(n+1)^2}+\beta_0(\tau_p)^2\,\Var(\mu_1)
=
\frac{a^2}{(n+1)^2}+\beta_0(\tau_p)^2\frac{u^2}{u+\tau}.
\end{equation}

Define the unilateral gain from joining when others join:
\[
\Delta(\tau_p):=\mathbb E[\pi_1^{S}]-\mathbb E[\pi_1^{NS}].
\]
Combining \eqref{eq:profit-joinall} and \eqref{eq:profit-NS} yields the fully explicit expression
\begin{equation*}
\label{eq:Delta-n}
\Delta(\tau_p)
=
\frac{\Var(\mu_{n})}{(n+1)^2}
-
\beta_0(\tau_p)^2\,\frac{u^2}{u+\tau},
\qquad
\Var(\mu_{n})=u-\Big(\frac{1}{u}+\frac{n}{\tau}+\frac{1}{\tau_p}\Big)^{-1},
\end{equation*}
with $\beta_0(\tau_p)$ given by \eqref{eq:beta0-closed} and $V_{-1}$ defined in \eqref{eq:m-minus1}.

The sharing profile is a symmetric equilibrium if and only if $\Delta_n(\tau_p)\ge 0$.
Define $\tau_p^*$ as the (unique) solution to $\Delta_n(\tau_p)=0$, which is exactly the indifference condition
\[
\frac{\Var(\mu_{n})}{(n+1)^2}
=
\beta_0(\tau_p)^2\frac{u^2}{u+\tau},
\]
equivalently \eqref{eq:taup-star-equation}.  This equation provides an explicit computable characterization of $\tau_p^*$ (e.g., via a one-dimensional root search), completing the proof.

\end{proof}

\subsection{\texorpdfstring{\Cref{prop:privacy-restores-sharing-n}}{propprivacyrestoressharingn}: Privacy noise induces voluntary sharing among $n$ firms}\label{app:proof:prop:privacy-restores-sharing-n}
\begin{proof}
We analyze the unilateral incentive of firm~$1$.

When a set $J\subseteq\{1,\dots,n\}$ participates, the platform observes privatized reports
\[
y_i:=s_i+\xi_i=\theta+(\varepsilon_i+\xi_i),\qquad \varepsilon_i\sim\mathcal N(0,\tau),\ \xi_i\sim\mathcal N(0,m),
\]
so that $\varepsilon_i+\xi_i\sim \mathcal N(0,\tau+m)$.
The platform releases the posterior mean
\[
t_J(m):=\mathbb{E}\!\left[\theta \,\middle|\, \{y_i\}_{i\in J},\, s_p\right].
\]
Under joint Gaussianity, $t_J(m)$ is linear and admits the precision-weighted form
\begin{equation}
\label{eq:tJ-form-proof}
t_J(m)
=
V_J(m)\left(
\frac{1}{\tau+m}\sum_{i\in J} y_i+\frac{1}{\tau_p}s_p
\right),
\qquad
V_J(m):=
\left(\frac{1}{u}+\frac{|J|}{\tau+m}+\frac{1}{\tau_p}\right)^{-1}.
\end{equation}
Moreover, since $t_J(m)=\mathbb{E}[\theta\mid \{y_i\}_{i\in J},s_p]$, the law of total variance yields
\begin{equation*}
\label{eq:var-tJ-proof}
\Var\!\big(t_J(m)\big)
=
\Var(\theta)-\mathbb{E}\!\big[\Var(\theta\mid \{y_i\}_{i\in J},s_p)\big]
=
u - V_J(m),
\end{equation*}
where the last equality uses $\Var(\theta)=u$ and the fact that in the Gaussian signal model the posterior variance
$\Var(\theta\mid \{y_i\}_{i\in J},s_p)=V_J(m)$ is deterministic.

We will use two special cases. Under full participation $J=[n]$,
\begin{equation*}
\label{eq:t-all-proof}
t_{[n]}(m)=V_n(m)\left(\frac{1}{\tau+m}\sum_{i=1}^n y_i+\frac{1}{\tau_p}s_p\right),
\qquad
V_n(m):=\left(\frac{1}{u}+\frac{n}{\tau+m}+\frac{1}{\tau_p}\right)^{-1}.
\end{equation*}
Under unilateral non-participation by firm~$1$, the platform aggregates only firms $2,\dots,n$:
\begin{equation*}
\label{eq:t-minus1-proof}
t_{-1}(m):=t_{[n]\setminus\{1\}}(m)
=
V_{n-1}(m)\left(\frac{1}{\tau+m}\sum_{i=2}^n y_i+\frac{1}{\tau_p}s_p\right),
\qquad
V_{n-1}(m):=\left(\frac{1}{u}+\frac{n-1}{\tau+m}+\frac{1}{\tau_p}\right)^{-1}.
\end{equation*}

\paragraph{Sharing equilibrium.}
Suppose all firms participate. Firm~$i$ observes $(s_i,t_{[n]}(m))$ and we look for a symmetric linear equilibrium of the form
\begin{equation}
\label{eq:q-joinall-privacy}
q_i^{S}(m)=c+d\,s_i+e\,t_{[n]}(m),\qquad i=1,\dots,n.
\end{equation}
Let $t:=t_{[n]}(m)$ and define $X:=(s_i,t)^{\mathsf T}$.
By Gaussian projection, for any $j\neq i$ and for $\theta$,
\begin{equation}
\label{eq:gauss-proj-joinall}
\mathbb{E}[s_j\mid s_i,t]=\Sigma_{s_jX}\Sigma_{XX}^{-1}X,
\qquad
\mathbb{E}[\theta\mid s_i,t]=\Sigma_{\theta X}\Sigma_{XX}^{-1}X,
\end{equation}
where the covariance objects are explicit. Writing $t=\alpha\sum_{k=1}^n y_k+\beta s_p$ with
\[
\alpha=\frac{V_n(m)}{\tau+m},\qquad \beta=\frac{V_n(m)}{\tau_p},
\]
we have
\[
\Var(s_i)=u+\tau,\qquad \Var(t)=u-V_n(m),
\qquad \Cov(\theta,t)=\Var(t)=u-V_n(m),
\]
and
\[
\Cov(s_i,t)=\Cov(s_j,t)=:\,C_{st},
\qquad
C_{st}=\alpha(nu+\tau)+\beta u\quad (j\neq i).
\]
\[
\Cov(s_i,s_j)=u\quad (j\neq i),
\qquad
\Cov(\theta,s_i)=u.
\]

Hence
\begin{align}
\label{eq:SigmaXX-joinall}
\Sigma_{XX}&=
\begin{bmatrix}
u+\tau & \alpha(nu+\tau)+\beta u\\
\alpha(nu+\tau)+\beta u & u-V_n(m)
\end{bmatrix}, \\
\Sigma_{s_jX}&=\begin{bmatrix} u & \alpha(nu+\tau)+\beta u\end{bmatrix}. \nonumber \\
\Sigma_{\theta X}&=\begin{bmatrix} u & u-V_n(m)\end{bmatrix}, \nonumber
\end{align}

Firm~$i$'s first-order condition under inverse demand $p=a+\theta-\sum_{k=1}^n q_k$ is
\begin{equation}
\label{eq:foc-joinall-privacy}
q_i^{S}(m)=\frac12\Big(a-\mathbb{E}[Q_{-i}\mid s_i,t]+\mathbb{E}[\theta\mid s_i,t]\Big),
\qquad Q_{-i}:=\sum_{k\neq i}q_k^{J}(m).
\end{equation}
Using \eqref{eq:q-joinall-privacy} and \eqref{eq:gauss-proj-joinall},
\[
\mathbb{E}[Q_{-i}\mid s_i,t]
=(n-1)c+(n-1)d\,\mathbb{E}[s_j\mid s_i,t]+(n-1)e\,t.
\]
Substituting into \eqref{eq:foc-joinall-privacy}, and matching coefficients on $(1,s_i,t)$ yields a linear system in $(c,d,e)$ with unique solution.\footnote{Uniqueness follows because $\Sigma_{XX}$ is positive definite for $u>0$, $\tau>0$, $\tau_p\in(0,\infty)$, and $m\ge 0$.}
In particular, the constant term is unchanged:
\begin{equation}
\label{eq:c-privacy}
c=\frac{a}{n+1}.
\end{equation}
The remaining coefficients can be written compactly in terms of the regression (projection) coefficients obtained from \eqref{eq:SigmaXX-joinall}. Define
\[
\mathbb{E}[s_j\mid s_i,t]=\lambda_s\,s_i+\lambda_t\,t,\qquad
\mathbb{E}[\theta\mid s_i,t]=\kappa_s\,s_i+\kappa_t\,t,
\]
where $(\lambda_s,\lambda_t)$ and $(\kappa_s,\kappa_t)$ are computed explicitly from \eqref{eq:gauss-proj-joinall}--\eqref{eq:SigmaXX-joinall}.

Matching coefficients on $(s_i,t)$ yields the linear system
\begin{equation}
\label{eq:de-system-privacy}
\begin{aligned}
d &= \frac12\big(\kappa_s-(n-1)d\lambda_s\big),\\
e &= \frac12\big(\kappa_t-(n-1)d\lambda_t\big)-\frac{n-1}{2}e.
\end{aligned}
\end{equation}

i.e., equivalently,
\begin{equation}
\label{eq:de-matched-privacy}
d=\frac12\Big(\kappa_s-(n-1)d\lambda_s\Big),
\qquad
e=\frac12\Big(\kappa_t-(n-1)d\lambda_t\Big)-\frac{n-1}{2}e.
\end{equation}
(These yield closed-form rational expressions in $(n,u,\tau,\tau_p,m)$ upon substituting $\lambda_s,\lambda_t,\kappa_s,\kappa_t$ from \eqref{eq:SigmaXX-joinall}.)

Under linear demand and risk-neutrality, conditional expected profit at the best response equals the square of the chosen quantity, hence
\[
\mathbb{E}[\pi_i^{S}(m)\mid s_i,t]=(q_i^{S}(m))^2,
\qquad\Rightarrow\qquad
\mathbb{E}[\pi_1^{S}(m)]=\mathbb{E}\!\big[(q_1^{S}(m))^2\big].
\]
Using \eqref{eq:q-joinall-privacy} and $\mathbb{E}[s_1]=\mathbb{E}[t]=0$,
\begin{equation}
\label{eq:profit-joinall-privacy}
\mathbb{E}[\pi_1^{S}(m)]
=
c^2+d^2\Var(s_1)+e^2\Var(t)+2de\,\Cov(s_1,t),
\end{equation}
where $\Var(s_1)=u+\tau$, $\Var(t)=u-V_n(m)$, and $\Cov(s_1,t)=\alpha(nu+\tau)+\beta u$.

\paragraph{Unilateral deviation.}
Suppose firm~$1$ does not participate, so it observes only $s_1$ and uses a linear strategy
\begin{equation}
\label{eq:q1-NS-privacy}
q_1^{NS}(m)=c_0+d_0\,s_1.
\end{equation}
The remaining $n-1$ firms participate. The platform releases $t_{-1}(m)$ in \eqref{eq:t-minus1-proof}.
Each participating firm $j\in\{2,\dots,n\}$ observes $(s_j,t_{-1}(m))$ and uses a linear strategy
\begin{equation}
\label{eq:qj-P-privacy}
q_j^{S}(m)=c_1+d_1\,s_j+e_1\,t_{-1}(m).
\end{equation}

\smallskip
\emph{Firm 1 (non-participation).}
Firm~$1$'s best response is
\[
q_1=\frac12\Big(a+\mathbb{E}[\theta\mid s_1]-\mathbb{E}[Q_{-1}\mid s_1]\Big),
\qquad Q_{-1}:=\sum_{j=2}^n q_j^{S}(m).
\]
We have $\mathbb{E}[\theta\mid s_1]=\frac{u}{u+\tau}s_1$ and $\mathbb{E}[s_2\mid s_1]=\frac{u}{u+\tau}s_1$.
Moreover, since $t_{-1}(m)$ is Gaussian and linear in $(\theta,\varepsilon_2,\dots,\varepsilon_n,\xi_2,\dots,\xi_n,\eta)$,
\[
\mathbb{E}[t_{-1}(m)\mid s_1]=\frac{\Cov(t_{-1}(m),s_1)}{\Var(s_1)}\,s_1.
\]
Writing $t_{-1}(m)=\alpha_{-1}\sum_{k=2}^n y_k+\beta_{-1}s_p$ with
\[
\alpha_{-1}=\frac{V_{n-1}(m)}{\tau+m},\qquad \beta_{-1}=\frac{V_{n-1}(m)}{\tau_p},
\]
a direct covariance calculation gives
\[
\Cov(s_1,t_{-1}(m))=\alpha_{-1}(n-1)u+\beta_{-1}u,\qquad \Var(s_1)=u+\tau.
\]
Therefore,
\begin{equation}
\label{eq:Etminus1-given-s1}
\mathbb{E}[t_{-1}(m)\mid s_1]=\frac{\alpha_{-1}(n-1)u+\beta_{-1}u}{u+\tau}\,s_1.
\end{equation}
Using \eqref{eq:qj-P-privacy} and linearity,
\[
\mathbb{E}[Q_{-1}\mid s_1]
=(n-1)c_1+(n-1)d_1\frac{u}{u+\tau}s_1+(n-1)e_1\,\mathbb{E}[t_{-1}(m)\mid s_1].
\]
Substituting into the best response and matching $(1,s_1)$ with \eqref{eq:q1-NS-privacy} yields
\begin{equation}
\label{eq:c0-d0-rel}
c_0=\frac{a-(n-1)c_1}{2},
\qquad
d_0=\frac12\left(
\frac{u}{u+\tau}-(n-1)d_1\frac{u}{u+\tau}-(n-1)e_1\frac{\alpha_{-1}(n-1)u+\beta_{-1}u}{u+\tau}
\right).
\end{equation}

\smallskip
\emph{Participating firms $j\ge 2$.}
Fix $j\in\{2,\dots,n\}$ and let $X:=(s_j,t_{-1}(m))^{\mathsf T}$.
Then Gaussian projection implies
\[
\mathbb{E}[s_1\mid s_j,t_{-1}(m)]=\Sigma_{s_1X}\Sigma_{XX}^{-1}X,
\qquad
\mathbb{E}[\theta\mid s_j,t_{-1}(m)]=\Sigma_{\theta X}\Sigma_{XX}^{-1}X.
\]
Moreover, for any $k\in\{2,\dots,n\}\setminus\{j\}$,
\[
\mathbb{E}[s_k\mid s_j,t_{-1}(m)]
=\Sigma_{s_kX}\Sigma_{XX}^{-1}X.
\]

The covariance elements are explicit:
\[
\Var(s_j)=u+\tau,\qquad 
\Var\!\big(t_{-1}(m)\big)=u-V_{n-1}(m),\qquad 
\Cov(\theta,t_{-1}(m))=\alpha_{-1}(n-1)u+\beta_{-1}u,
\]
\[
\Cov(s_j,t_{-1}(m))=\alpha_{-1}((n-1)u+\tau)+\beta_{-1}u,
\qquad
\Cov(s_1,t_{-1}(m))=\alpha_{-1}(n-1)u+\beta_{-1}u,
\]
\[
\Cov(s_1,s_j)=u,\qquad 
\Cov(\theta,s_j)=u,
\qquad
\Cov(s_k,s_j)=u,
\qquad
\Cov(s_k,t_{-1}(m))=\Cov(s_j,t_{-1}(m)).
\]

Thus the $2\times2$ covariance matrix $\Sigma_{XX}$ of $(s_j,t_{-1}(m))$ and the cross-covariance vectors
$\Sigma_{s_1X}$ ,$\Sigma_{\theta X}$, $\Sigma_{s_kX}$ are known explicitly and can be inverted in closed form.

Since $s_1=\theta+\varepsilon_1$, where $\varepsilon_1$ is independent of 
$(s_j,t_{-1}(m))$ and has zero mean, it follows that
\[
\mathbb{E}[s_1\mid s_j,t_{-1}(m)]
=\mathbb{E}[\theta\mid s_j,t_{-1}(m)].
\]
Therefore, Gaussian projection yields
\[
\mathbb{E}[s_1\mid s_j,t_{-1}(m)]
=\mathbb{E}[\theta\mid s_j,t_{-1}(m)]
=\phi_s\,s_j+\phi_t\,t_{-1}(m),
\]
and for any $k\neq j$, $k\ge 2$,
\[
\mathbb{E}[s_k\mid s_j,t_{-1}(m)]
=\chi_s\,s_j+\chi_t\,t_{-1}(m),
\]
where
\[
(\phi_s,\phi_t):=\Sigma_{s_1X}\Sigma_{XX}^{-1}
=\Sigma_{\theta X}\Sigma_{XX}^{-1},
\qquad
(\chi_s,\chi_t):=\Sigma_{s_kX}\Sigma_{XX}^{-1}.
\]

The first-order condition for firm $j$ is
\[
q_j=\frac12\Big(a+\mathbb{E}[\theta\mid s_j,t_{-1}(m)]
-\mathbb{E}[q_1^{NS}(m)+\sum_{k\in\{2,\dots,n\}\setminus\{j\}}q_k^{S}(m)\mid s_j,t_{-1}(m)]\Big).
\]

Using the linear strategies
\[
q_1^{NS}(m)=c_0+d_0 s_1,
\qquad
q_k^{S}(m)=c_1+d_1 s_k+e_1 t_{-1}(m),
\]
we obtain
\[
\mathbb{E}[q_1^{NS}(m)\mid s_j,t]
= c_0+d_0(\phi_s s_j+\phi_t t),
\]
\[
\mathbb{E}\!\left[\sum_{k\neq j,k\ge 2} q_k^{S}(m)\mid s_j,t\right]
=(n-2)c_1+(n-2)d_1(\chi_s s_j+\chi_t t)+(n-2)e_1 t.
\]

Matching coefficients on $(1,s_j,t_{-1}(m))$ with
$q_j^{S}(m)=c_1+d_1 s_j+e_1 t_{-1}(m)$ yields the linear system
\begin{equation}
\label{eq:c1-d1-e1-rel-correct}
\begin{aligned}
c_1 &= \frac12\Big(a-c_0-(n-2)c_1\Big),\\
d_1 &= \frac12\Big(\phi_s - d_0\,\phi_s - (n-2)d_1\,\chi_s\Big),\\
e_1 &= \frac12\Big(\phi_t - d_0\,\phi_t - (n-2)d_1\,\chi_t - (n-2)e_1\Big).
\end{aligned}
\end{equation}

Combining \eqref{eq:c0-d0-rel} and the first equation in
\eqref{eq:c1-d1-e1-rel-correct} gives
\begin{equation*}
\label{eq:c0c1-final}
c_0=c_1=\frac{a}{n+1},
\end{equation*}
and the remaining coefficients $(d_0,d_1,e_1)$ are uniquely determined by the resulting
$3\times3$ linear system in \eqref{eq:c0-d0-rel} and
\eqref{eq:c1-d1-e1-rel-correct}.

Conditional expected profit at the best response equals the square of the chosen quantity, hence
\[
\mathbb{E}[\pi_1^{NS}(m)\mid s_1]=(q_1^{NS}(m))^2,
\qquad\Rightarrow\qquad
\mathbb{E}[\pi_1^{NS}(m)]=\mathbb{E}\!\big[(q_1^{NS}(m))^2\big]
=c_0^2+d_0^2\Var(s_1),
\]
so
\begin{equation}
\label{eq:profit-NS-privacy}
\mathbb{E}[\pi_1^{NS}(m)]
=
\frac{a^2}{(n+1)^2}+d_0^2\,(u+\tau).
\end{equation}

Define the unilateral participation gain under privacy noise,
\[
\Delta(\tau_p,m):=\mathbb{E}[\pi_1^{S}(m)]-\mathbb{E}[\pi_1^{NS}(m)].
\]
By construction, the privacy mechanism with $m=0$ coincides with the non-private platform. Since $\tau_p>\tau_p^*(n,u,\tau)$ implies that full participation is not an equilibrium under the non-private platform, we have
\begin{equation}
\label{eq:Delta-neg-at-0}
\Delta_n(\tau_p,0)<0.
\end{equation}

Next, consider the limit $m\to\infty$. In \eqref{eq:tJ-form-proof}, the precision of each privatized report is $1/(\tau+m)\to 0$, so the platform statistic satisfies
\[
t_{[n]}(m)\ \xrightarrow[]{m\nearrow\infty}\ \mathbb{E}[\theta\mid s_p],
\qquad
t_{-1}(m)\ \xrightarrow[]{m\nearrow\infty}\ \mathbb{E}[\theta\mid s_p],
\]
in $L^2$ (and hence in distribution). Thus, as $m\to\infty$, the sharing environment grants firm~$1$ access (through participation) to the informative public posterior $\mathbb{E}[\theta\mid s_p]$, whereas under unilateral non-participation firm~$1$ remains restricted to $s_1$ while its rivals condition on $(s_j,\mathbb{E}[\theta\mid s_p])$. Since $\tau_p<\infty$ implies that $\mathbb{E}[\theta\mid s_p]$ is non-degenerate and payoff-relevant, this strictly increases firm~$1$'s optimized ex-ante payoff, implying
\begin{equation}
\label{eq:Delta-pos-infty}
\lim_{m\to\infty}\Delta(\tau_p,m)>0.
\end{equation}

Finally, $\Delta_n(\tau_p,m)$ is continuous in $m$ on $[0,\infty)$: all equilibrium coefficients are obtained by solving linear systems whose entries are polynomials in the covariance elements of $(\theta,s_i,s_p,\{y_i\})$, and these covariance elements depend continuously on $m$ (indeed, rationally through $V_n(m)$ and $V_{n-1}(m)$).
Therefore, by \eqref{eq:Delta-neg-at-0}--\eqref{eq:Delta-pos-infty} and the intermediate value theorem, there exists $m^*>0$ such that $\Delta_n(\tau_p,m^*)=0$, and any $m>m^*$ yields $\Delta_n(\tau_p,m)>0$.
This proves that under the privacy-preserving platform, full participation is sustained as a symmetric equilibrium for a suitable privacy noise level.
\end{proof}

\subsection{\texorpdfstring{\Cref{prop:privacy-noise-welfare}}{propprivacynoisewelfare}: Effect of privacy noise}
\label{app:proof:prop:privacy-noise-welfare}

\begin{proof}
The privacy noise \(m\) affects welfare only through the informativeness of the
platform posterior. Differentiating \(\sigma_S^2\) with respect to \(m\), we obtain
\[
\frac{\partial \sigma_S^2}{\partial m}
=
-\frac{n}{(\tau+m)^2\mathcal{P}_n^2}
<0.
\]
Since
\[
\frac{\partial \mathbb{E}[CS^S]}{\partial \sigma_S^2}
=
\frac{n^2}{2(n+1)^2}>0,
\]
\[
\frac{\partial \mathbb{E}[PS^S]}{\partial \sigma_S^2}
=
\frac{n}{(n+1)^2}>0,
\]
and
\[
\frac{\partial \mathbb{E}[TS^S]}{\partial \sigma_S^2}
=
\frac{n(n+2)}{2(n+1)^2}>0,
\]
it follows that
\[
\frac{\partial \mathbb{E}[CS^S]}{\partial m}<0,
\qquad
\frac{\partial \mathbb{E}[PS^S]}{\partial m}<0,
\qquad
\frac{\partial \mathbb{E}[TS^S]}{\partial m}<0.
\]
The no-sharing surplus levels do not depend on \(m\). Therefore,
\[
\frac{\partial \Delta CS}{\partial m}<0,
\qquad
\frac{\partial \Delta PS}{\partial m}<0,
\qquad
\frac{\partial \Delta TS}{\partial m}<0.
\]
\end{proof}

\subsection{\texorpdfstring{\Cref{prop:platform-signal-noise-welfare}}{propplatformsignalnoisewelfare}: Effect of platform signal noise}
\label{app:proof:prop:platform-signal-noise-welfare}

\begin{proof}
The platform signal noise \(\tau_p\) affects welfare only through the informativeness
of the shared posterior. Differentiating \(\sigma_S^2\) with respect to \(\tau_p\), we obtain
\[
\frac{\partial \sigma_S^2}{\partial \tau_p}
=
-\frac{1}{\tau_p^2\mathcal{P}_n^2}
<0.
\]
Because each sharing-regime welfare term is increasing in \(\sigma_S^2\), we have
\[
\frac{\partial \mathbb{E}[CS^S]}{\partial \tau_p}<0,
\qquad
\frac{\partial \mathbb{E}[PS^S]}{\partial \tau_p}<0,
\qquad
\frac{\partial \mathbb{E}[TS^S]}{\partial \tau_p}<0.
\]
The no-sharing surplus levels are independent of \(\tau_p\). Hence,
\[
\frac{\partial \Delta CS}{\partial \tau_p}<0,
\qquad
\frac{\partial \Delta PS}{\partial \tau_p}<0,
\qquad
\frac{\partial \Delta TS}{\partial \tau_p}<0.
\]
\end{proof}

\subsection{\texorpdfstring{\Cref{prop:full-sharing-implementation}}{propfullsharing}: Full-sharing implementation}
\label{app:prop:full-sharing-implementation}
\begin{proof}
Consider any participation profile that is not full participation. Then there exists at least
one firm that does not participate. Let \(k\in\{0,\ldots,n-1\}\) be the number of other firms
that participate. By assumption,
\[
\Delta^k(m)>0.
\]
Therefore, this nonparticipating firm can strictly increase its expected profit by joining the
platform. Hence, no non-full participation profile can be a pure-strategy equilibrium.

At the full-participation profile, each firm faces \(n-1\) other participants. Since
\[
\Delta^{n-1}(m)>0,
\]
no firm has a profitable deviation from participation to nonparticipation. Therefore, full
participation is a pure-strategy equilibrium. Since all non-full participation profiles have
been ruled out, full participation is the unique pure-strategy equilibrium.
\end{proof}

\section{Information Sharing with Heterogeneous Signals and Under Sequential Entry} \label{sec:heterogeneous-sequential}

\subsection{Information Sharing with Heterogeneous Private Signals}\label{app:b-1:information-sharing-with-heterogeneous-private-signals}

We extend the baseline model to allow for heterogeneity in firms’ private information.
Firm $i$ observes a private signal
\[
s_i = \theta + \varepsilon_i, \qquad \varepsilon_i \sim \mathcal N(0,\tau_i),
\]
where signal noise variances $(\tau_i)_{i=1}^n$ may differ across firms.
As before, the platform observes an external signal
\[
s_p = \theta + \varepsilon_p, \qquad \varepsilon_p \sim \mathcal N(0,\tau_p).
\]

Under the privacy-preserving mechanism, if firm $i$ participates, its reported signal is perturbed as
\[
\tilde s_i = s_i + v_i, \qquad v_i \sim \mathcal N(0,m_i),
\]
where the platform may choose firm-specific privacy noise levels $(m_i)_{i=1}^n$.
Given the privatized reports of participating firms and the platform signal,
the platform releases the posterior mean of the common state to all participants.

\begin{proposition}[Privacy Restores Voluntary Sharing under Heterogeneous Signals]
\label{prop:privacy-restores-sharing-hetero}
Consider the $n$-firm model with heterogeneous signal precisions $(\tau_i)_{i=1}^n$ described above.
Suppose the platform signal has finite variance, i.e., $\tau_p<\infty$, but is not precise enough
to sustain full participation under the non-private mechanism.
Then there exists a finite vector of privacy noise levels
$\mathbf m^\ast=(m_1^\ast,\dots,m_n^\ast)\in(0,\infty)^n$ such that,
under the privacy-preserving mechanism with $\mathbf m \ge \mathbf m^\ast$,
every firm strictly prefers to participate when all other firms participate. Consequently, the sharing profile is sustained as an equilibrium of the participation game.
\end{proposition}

\Cref{prop:privacy-restores-sharing-hetero}, proved in \cref{app:proof:prop:privacy-restores-sharing-hetero}, shows that the participation-restoring
role of privacy noise is robust to heterogeneity in firms’ information quality.
Even when firms possess private signals of different precisions, the platform can design a
privacy-preserving mechanism---potentially with asymmetric noise levels across firms---that
induces voluntary participation by all firms.

The key economic insight is unchanged from the homogeneous benchmark.
When privacy noise is sufficiently large, the information revealed through firms’ reports becomes
negligible, and the platform’s released statistic converges to the public posterior based solely
on the platform signal.
Because the platform signal has finite precision, access to this public posterior strictly
improves each firm’s information set relative to non-participation, under which a deviating firm
relies only on its private signal while competitors condition on a richer information set.
As a result, participation becomes a strictly dominant response to full participation by others,
despite heterogeneity in private signal accuracy.

To illustrate how heterogeneity in signal accuracy affects the required level of privacy noise,
we present a numerical example in a two-firm setting.
\Cref{fig:sharing_region_hete_tau} depicts the information-sharing regions in the $(m_1,m_2)$ plane
under the privacy-preserving mechanism.
The shaded regions indicate parameter values for which firm~1 and firm~2, respectively, are willing
to participate in the platform, while their intersection corresponds to values of $(m_1,m_2)$ for which
joint participation constitutes a Nash equilibrium.

Across panels, we vary firm~1’s signal precision by changing $\tau_1\in\{1,3,5\}$, while holding
firm~2’s signal noise fixed at $\tau_2=5$.
Other parameters are set to $u=1$, $\tau_p=20$, and $a=1$.
The figure reveals a clear comparative-static pattern.
As firm~1’s private signal becomes more precise (i.e., as $\tau_1$ decreases), the minimum privacy
noise required to induce participation by firm~2 decreases, whereas the privacy noise required for
firm~1’s own participation increases.
Intuitively, a more informative signal from firm~1 raises the informational benefit of participation
for firm~2, reducing its need for privacy protection, while simultaneously increasing firm~1’s
exposure to information leakage, thereby necessitating stronger privacy noise.

In the extreme case $\tau_1=1$, firm~2 is willing to participate even when $m_2=0$,
provided that the privacy noise applied to firm~1’s report is not too large.
In this regime, firm~2 benefits sufficiently from accessing firm~1’s highly informative signal,
making participation attractive despite the absence of direct privacy protection.

\begin{figure}[t]
    \centering
    \begin{subfigure}[t]{0.4\textwidth}
        \centering
        \includegraphics[width=\textwidth]{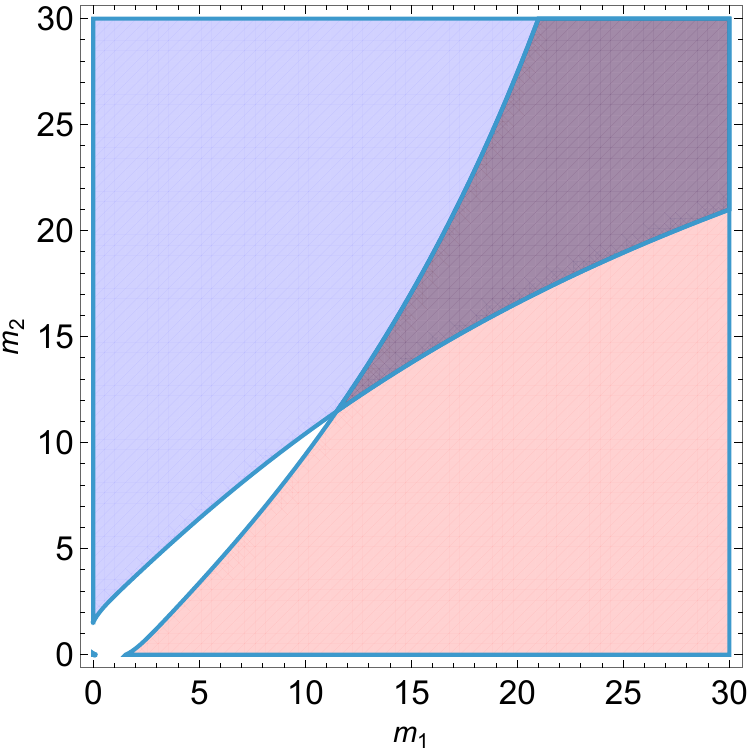}
        \caption{$\tau_1=5$}
        \label{fig:hete_tau5}
    \end{subfigure}
    \hfill
    \begin{subfigure}[t]{0.4\textwidth}
        \centering
        \includegraphics[width=\textwidth]{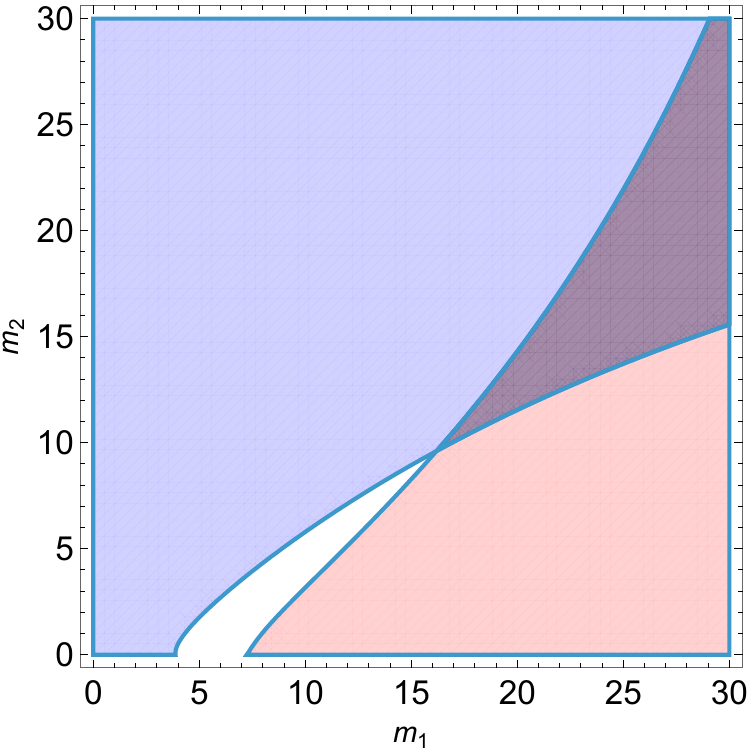}
        \caption{$\tau_1=3$}
        \label{fig:hete_tau3}
    \end{subfigure}

    \vspace{0.6em}

    \hspace*{0.15\textwidth}
    \begin{subfigure}[t]{0.55\textwidth}
        \centering
        \includegraphics[width=\textwidth]{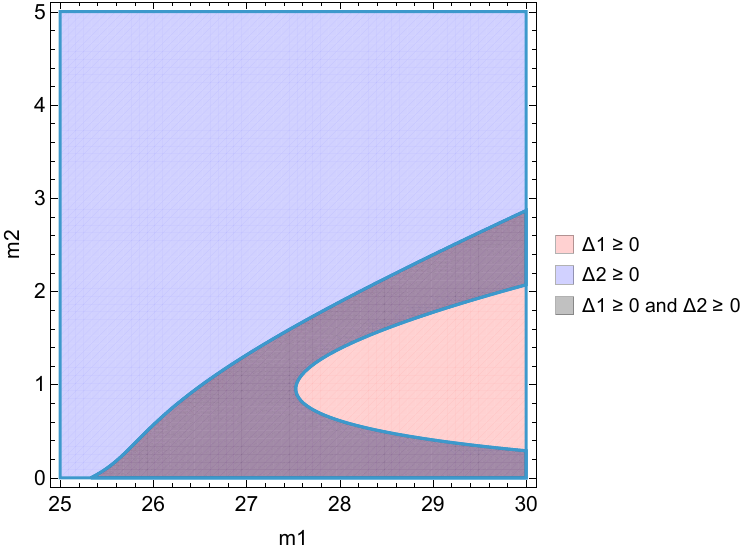}
        \caption{$\tau_1=1$}
        \label{fig:hete_tau1}
    \end{subfigure}

    \caption{
    \textbf{Information-sharing regions under heterogeneous private signals.}
    Each panel depicts the regions in the $(m_1,m_2)$ plane where firm~1 and firm~2 are willing to participate in the privacy-preserving information-sharing platform.
    The intersection corresponds to parameter values for which joint participation constitutes a Nash equilibrium.
    Across panels, firm~1's signal precision varies ($\tau_1\in\{1,3,5\}$), while firm~2's signal noise is fixed at $\tau_2=5$.
    Other parameters are $u=1$, $\tau_p=20$, and $a=1$.
    }
    \label{fig:sharing_region_hete_tau}
\end{figure}

To complement the two-dimensional sharing regions in
\Cref{fig:sharing_region_hete_tau,fig:equilibrium_m_vs_tau1}
provide a one-dimensional comparative-static perspective.
As firm~1’s signal noise $\tau_1$ increases, its equilibrium privacy noise requirement $m_1^*$
declines, reflecting the reduced privacy cost associated with a less informative signal.
In contrast, the equilibrium noise required to induce firm~2’s participation, $m_2^*$,
increases as the informational benefit from firm~1’s signal diminishes.

\begin{figure}[t]
    \centering
    \begin{subfigure}[t]{0.48\textwidth}
        \centering
        \includegraphics[width=\textwidth]{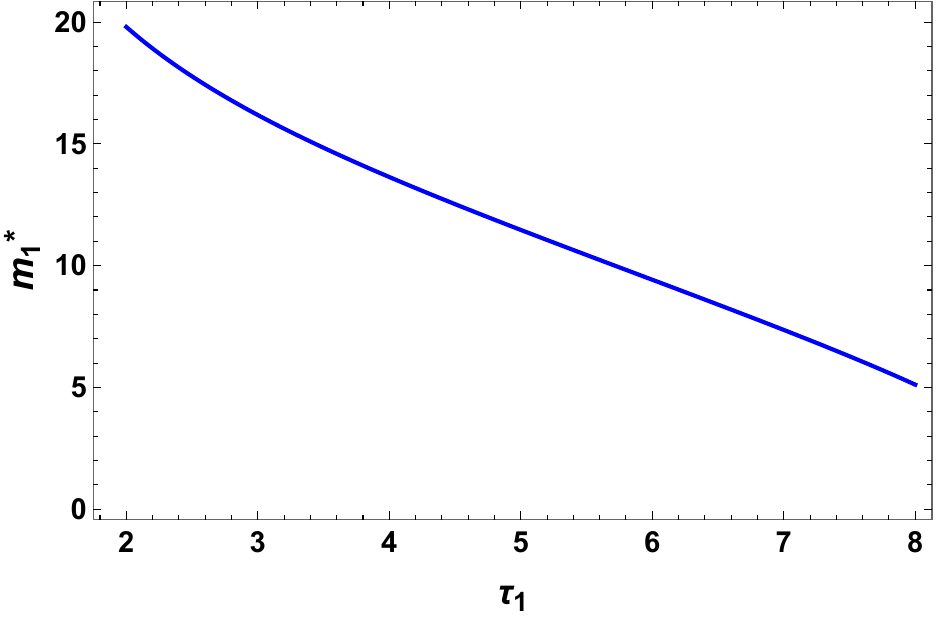}
        \caption{Equilibrium $m_1^*$ vs $\tau_1$}
        \label{fig:m1_tau1}
    \end{subfigure}
    \hfill
    \begin{subfigure}[t]{0.48\textwidth}
        \centering
        \includegraphics[width=\textwidth]{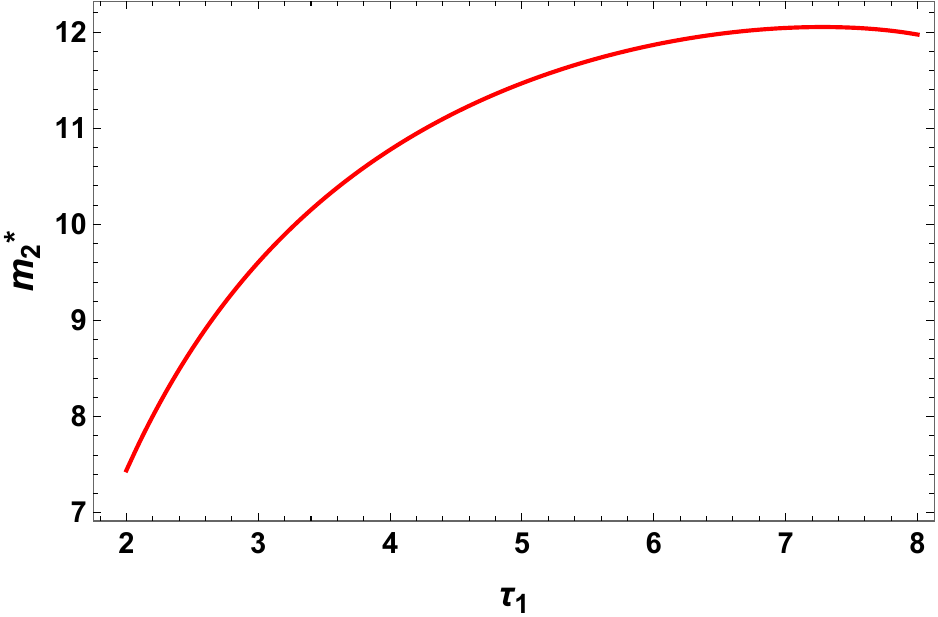}
        \caption{Equilibrium $m_2^*$ vs $\tau_1$}
        \label{fig:m2_tau1}
    \end{subfigure}

    \caption{
    \textbf{Impact of firm 1's signal precision $\tau_1$ on equilibrium privacy noise levels.} 
The left panel shows how the platform’s equilibrium privacy noise level for firm~1, $m_1^*$, varies with firm~1’s signal precision $\tau_1$, while the right panel shows the corresponding cross-effect on the equilibrium noise level $m_2^*$ assigned to firm~2. 
Fixed parameters are $\tau_2=5$, $u=1$, $\tau_p=20$, and $a=1$.
    }
    \label{fig:equilibrium_m_vs_tau1}
\end{figure}

\subsection{Heterogeneous Signals under Sequential Participation}\label{app:b-2:heterogeneous-signals-under-sequential-participation}

In this subsection, we study participation incentives in a heterogeneous information environment
under sequential decision-making.
We consider a setting in which all firms other than firm $i$ have already joined the platform
and share their signals using a common privacy noise level $m$.
Firm $i$ observes a private signal with noise variance $\tau_i$ that may differ from that of its competitors.

Rather than solving for a joint equilibrium in participation and privacy choices,
we analyze the participation decision of firm $i$ conditional on the existing participation of its competitors.
The platform is allowed to tailor the privacy noise level $m_i$ applied to firm $i$'s shared signal
in order to induce its participation.
This formulation captures a natural dynamic interpretation in which firms join the platform sequentially,
and the platform designs firm-specific privacy protection to encourage additional participation. In the following proposition, by focusing on this best-response problem, we isolate how firm-level signal heterogeneity shapes participation incentives, abstracting from the equilibrium feedback effects of simultaneous participation decisions (proof in \cref{app:proof:prop:sequential-heterogeneous}).

\begin{proposition}[Privacy-Induced Participation under Sequential Entry]
\label{prop:sequential-heterogeneous}
Consider the $n$-firm Cournot model under sequential participation.
Suppose all firms other than firm $i$ have already joined the platform
and share their signals using a common privacy noise level $m$.
Firm $i$ observes a private signal with noise variance $\tau_i$
and, if participating, is subject to a firm-specific privacy noise $m_i$.

Assume the platform signal has finite variance $\tau_p<\infty$ but is not sufficiently informative
to induce firm $i$ to participate without privacy protection, i.e.,
$\tau_p^*(n,u,\tau_i,\tau)<\tau_p<\infty$.
Then there exists a finite threshold $\bar m_i>0$ such that
assigning $m_i\ge \bar m_i$ induces firm $i$ to prefer participation,
taking other firms’ participation as given.
\end{proposition}

We complement the analytical results with a numerical illustration that characterizes
how the firm-specific privacy noise threshold varies with signal heterogeneity.
Throughout the experiment, we fix the prior variance at $u=1$ and the demand intercept at $a=1$.
Firm $i$’s signal noise variance $\tau_i$ varies over the interval $[2,6]$,
while all other firms share signals with a common noise variance $\tau=5$.
The platform provides an external signal with noise variance $\tau_p=10$,
which by itself is insufficient to induce participation in the absence of privacy protection.
All other firms apply a common privacy noise level $m=10$.

For each value of $\tau_i$ and each market size $n\in\{3,5,10\}$,
we numerically compute the minimal firm-specific privacy noise $m_i^*$
that makes firm $i$ indifferent between participating and not participating,
taking the participation of all other firms as given.
The threshold $m_i^*$ is obtained by solving $\Delta_i(\tau_i,m_i)=0$,
where $\Delta_i$ denotes firm $i$’s unilateral participation gain.

\begin{figure}[t]
    \centering
    \includegraphics[width=0.7\linewidth]{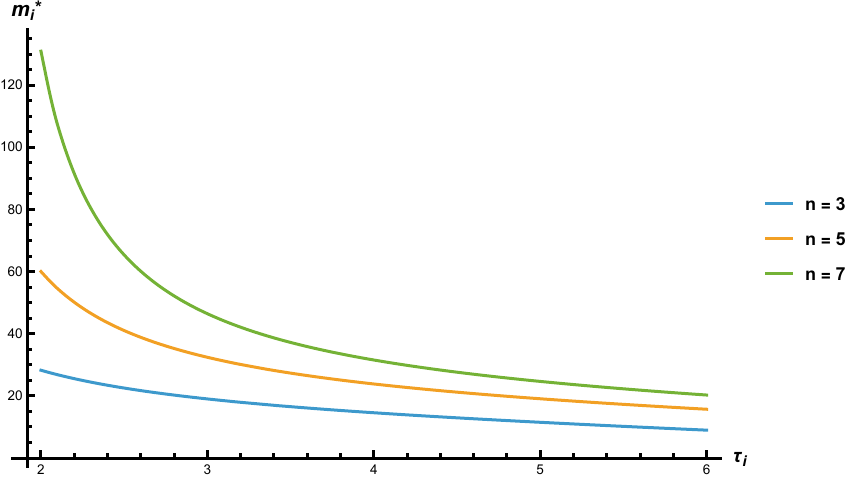}
    \caption{Firm-specific privacy noise threshold $m_i^*$ as a function of signal noise $\tau_i$,
    for different numbers of competing firms $n$.}
    \label{fig:mi_threshold_tau_i}
\end{figure}

\Cref{fig:mi_threshold_tau_i} plots the resulting relationship between $\tau_i$ and $m_i^*$.
Two patterns emerge clearly.
First, for any fixed market size $n$, the required privacy noise $m_i^*$ is decreasing in $\tau_i$:
firms with noisier private signals require less privacy protection to be willing to share.
Second, for any given $\tau_i$, the threshold $m_i^*$ is increasing in $n$,
reflecting stronger competitive externalities in larger markets,
which in turn necessitate stronger privacy protection to sustain individual participation.

\section{Additional Welfare Analysis}
\label{app:welfare-derivations}

This appendix provides the surplus derivations and additional comparative statics omitted from the main text.

\subsection{Surplus Derivations}
\label{app:surplus-derivations}

For the inverse demand function $P(Q,\theta)=a+\theta-Q$, consumer surplus is
\[
CS(Q,\theta)
=
\int_0^Q (a+\theta-x)\,dx
-
P(Q,\theta)Q
=
\frac{Q^2}{2}.
\]
Because production costs are normalized to zero, producer surplus equals total firm profit:
\[
PS(Q,\theta)
=
\sum_{i=1}^n \pi_i
=
QP(Q,\theta)
=
Q(a+\theta-Q).
\]
Therefore, total surplus is
\[
TS(Q,\theta)
=
CS(Q,\theta)+PS(Q,\theta)
=
(a+\theta)Q-\frac{Q^2}{2}.
\]
We use these expressions to derive expected surplus under full platform participation and under the no-sharing benchmark.

\paragraph{Full platform participation.}
Suppose all firms participate in the platform and the platform uses a symmetric privacy noise level $m_i=m$. The platform releases the posterior mean
\[
\mu_n
=
\mathbb{E}[\theta \mid \widetilde{s}_1,\ldots,\widetilde{s}_n,s_p],
\]
where
\[
\widetilde{s}_i=s_i+v_i,
\qquad
v_i\sim \mathcal{N}(0,m),
\]
and
\[
s_p=\theta+\varepsilon_p,
\qquad
\varepsilon_p\sim \mathcal{N}(0,\tau_p).
\]
The posterior variance of $\theta$ under full participation is
\[
V_n^S
=
\left(
\frac{1}{u}
+
\frac{n}{\tau+m}
+
\frac{1}{\tau_p}
\right)^{-1}.
\]
Thus, the variance of the posterior mean is
\[
\sigma_S^2(m)
:=
\operatorname{Var}(\mu_n)
=
u-V_n^S
=
u-
\left(
\frac{1}{u}
+
\frac{n}{\tau+m}
+
\frac{1}{\tau_p}
\right)^{-1}.
\]

Under full information sharing, firms observe the common posterior mean $\mu_n$. The symmetric Cournot quantity is
\[
q_i^S
=
\frac{a+\mu_n}{n+1},
\]
and total output is
\[
Q^S
=
\sum_{i=1}^n q_i^S
=
\frac{n(a+\mu_n)}{n+1}.
\]
Since $\mathbb{E}[\mu_n]=0$, we have
\[
\mathbb{E}[(a+\mu_n)^2]=a^2+\sigma_S^2(m).
\]
Therefore,
\[
\mathbb{E}[CS^S]
=
\frac{1}{2}\mathbb{E}[(Q^S)^2]
=
\frac{n^2}{2(n+1)^2}
\left(a^2+\sigma_S^2(m)\right),
\]
\[
\mathbb{E}[PS^S]
=
\sum_{i=1}^n \mathbb{E}[\pi_i^S]
=
\frac{n}{(n+1)^2}
\left(a^2+\sigma_S^2(m)\right),
\]
and
\[
\mathbb{E}[TS^S]
=
\mathbb{E}[CS^S]+\mathbb{E}[PS^S]
=
\frac{n(n+2)}{2(n+1)^2}
\left(a^2+\sigma_S^2(m)\right).
\]

\paragraph{No information sharing.}
Now suppose firms do not participate in the platform. Firm $i$ observes only its private signal
\[
s_i=\theta+\varepsilon_i,
\qquad
\varepsilon_i\sim \mathcal{N}(0,\tau).
\]
Let
\[
\mu_i
=
\mathbb{E}[\theta\mid s_i]
=
\frac{u}{u+\tau}s_i
\]
denote firm $i$'s private posterior mean. Define
\[
\sigma_P^2
:=
\operatorname{Var}(\mu_i)
=
\frac{u^2}{u+\tau},
\]
and
\[
\rho
:=
\frac{\operatorname{Cov}(\mu_i,\mu_j)}{\operatorname{Var}(\mu_i)}
=
\frac{u}{u+\tau},
\qquad i\neq j.
\]
Thus,
\[
\operatorname{Cov}(\mu_i,\mu_j)=\rho\sigma_P^2.
\]

In the symmetric Bayesian Cournot equilibrium without information sharing, firm $i$'s quantity takes the linear form
\[
q_i^{NS}
=
\frac{a}{n+1}
+
\beta \mu_i,
\]
where
\[
\beta
=
\frac{1}{2+(n-1)\rho}.
\]
Total output is
\[
Q^{NS}
=
\sum_{i=1}^n q_i^{NS}
=
\frac{na}{n+1}
+
\beta\sum_{i=1}^n \mu_i.
\]
Since
\[
\operatorname{Var}\left(\sum_{i=1}^n \mu_i\right)
=
n\sigma_P^2[1+(n-1)\rho],
\]
expected consumer surplus under no sharing is
\[
\mathbb{E}[CS^{NS}]
=
\frac{n^2a^2}{2(n+1)^2}
+
\frac{1}{2}\beta^2 n\sigma_P^2[1+(n-1)\rho].
\]
Similarly,
\[
\mathbb{E}[PS^{NS}]
=
n\mathbb{E}[(q_i^{NS})^2]
=
\frac{na^2}{(n+1)^2}
+
n\beta^2\sigma_P^2.
\]
Expected total surplus is
\[
\mathbb{E}[TS^{NS}]
=
\frac{n(n+2)a^2}{2(n+1)^2}
+
\frac{n\beta^2\sigma_P^2}{2}
[3+(n-1)\rho].
\]

\paragraph{Surplus differences.}
Define
\[
\Delta CS
=
\mathbb{E}[CS^S]-\mathbb{E}[CS^{NS}],
\]
\[
\Delta PS
=
\mathbb{E}[PS^S]-\mathbb{E}[PS^{NS}],
\]
and
\[
\Delta TS
=
\mathbb{E}[TS^S]-\mathbb{E}[TS^{NS}].
\]
The deterministic demand-intercept terms cancel out. Therefore,
\[
\Delta CS
=
\frac{1}{2}
\left[
\frac{n^2}{(n+1)^2}\sigma_S^2(m)
-
n\beta^2\sigma_P^2[1+(n-1)\rho]
\right],
\]
\[
\Delta PS
=
n
\left[
\frac{\sigma_S^2(m)}{(n+1)^2}
-
\beta^2\sigma_P^2
\right],
\]
and
\[
\Delta TS
=
\frac{n}{2}
\left[
\frac{n+2}{(n+1)^2}\sigma_S^2(m)
-
\beta^2\sigma_P^2[3+(n-1)\rho]
\right].
\]

\subsection{Additional Comparative Statics}
\label{app:welfare-comparative-statics}

This subsection records comparative statics with respect to prior uncertainty $u$, private signal noise $\tau$, and the number of firms $n$. The main text focuses on privacy noise $m$ and platform signal noise $\tau_p$ because those parameters are directly controlled or affected by the platform mechanism.

\begin{proposition}[Effect of demand uncertainty]
\label{prop:demand-uncertainty-welfare}
Under full platform participation, expected consumer surplus, producer surplus, and total surplus are strictly increasing in the prior variance $u$. Under no information sharing, expected consumer surplus, producer surplus, and total surplus are also strictly increasing in $u$.
\end{proposition}

\Cref{prop:demand-uncertainty-welfare} shows that greater demand uncertainty raises the absolute value of information in both regimes. Under platform participation, a more volatile state makes the shared posterior more valuable. Under no sharing, it also makes each firm's private posterior more valuable. Hence, the effect of $u$ on the surplus gains from platform participation is generally ambiguous: $\Delta CS$, $\Delta PS$, and $\Delta TS$ may rise or fall depending on whether the shared-posterior channel dominates the corresponding increase in the no-sharing benchmark.

\begin{proposition}[Effect of private signal noise]
\label{prop:private-signal-noise-welfare}
Under full platform participation, expected consumer surplus, producer surplus, and total surplus are strictly decreasing in the private signal noise $\tau$. Under no information sharing, expected consumer surplus and expected total surplus are strictly decreasing in $\tau$. The marginal effect of $\tau$ on expected producer surplus under no information sharing satisfies
\[
\frac{\partial \mathbb{E}[PS^{NS}]}{\partial \tau}\ge 0
\quad
\Longleftrightarrow
\quad
\tau\le \frac{(n-3)u}{2},
\]
and
\[
\frac{\partial \mathbb{E}[PS^{NS}]}{\partial \tau}\le 0
\quad
\Longleftrightarrow
\quad
\tau\ge \frac{(n-3)u}{2}.
\]
In particular, if $n\le 3$, expected producer surplus under no information sharing is strictly decreasing in $\tau$ for all $\tau>0$.
\end{proposition}

\Cref{prop:private-signal-noise-welfare} highlights that private signal noise affects welfare through both information and competition. Under full platform participation, a larger $\tau$ reduces the informativeness of firms' reports and lowers the value of the platform posterior. Under no sharing, a larger $\tau$ weakens private information, which reduces consumer surplus and total surplus. Producer surplus is more subtle because weaker private information can also soften state-contingent competitive responses. When $\tau$ is small and $n>3$, this strategic-softening effect can dominate; when $\tau$ is large, the information-loss effect dominates.

\begin{proposition}[Effect of the number of firms]
\label{prop:number-of-firms-welfare}
Treating $n$ as continuous for comparative-static intuition, expected consumer surplus and expected total surplus are strictly increasing in $n$ under both full platform participation and no information sharing.

The effect of $n$ on expected producer surplus is ambiguous. Under full platform participation,
\[
\mathbb{E}[PS^S]
=
\frac{n}{(n+1)^2}(a^2+\sigma_S^2(m)).
\]
The Cournot producer-surplus coefficient $n/(n+1)^2$ is decreasing in $n$ for $n>1$, while the informativeness term $\sigma_S^2(m)$ is increasing in $n$. Hence, expected producer surplus under full platform participation may increase or decrease with $n$.

Under no information sharing,
\[
\mathbb{E}[PS^{NS}]
=
\frac{na^2}{(n+1)^2}
+
n\beta^2\sigma_P^2.
\]
The deterministic Cournot component $na^2/(n+1)^2$ is decreasing in $n$ for $n>1$. The information-dependent component has marginal effect
\[
\frac{\partial}{\partial n}
\left(
n\beta^2\sigma_P^2
\right)
=
\sigma_P^2
\frac{2-(n+1)\rho}
{\left[2+(n-1)\rho\right]^3}.
\]
Thus, this component increases with $n$ when
\[
\rho<\frac{2}{n+1},
\]
and decreases with $n$ when
\[
\rho>\frac{2}{n+1}.
\]
Therefore, expected producer surplus under no information sharing may also increase or decrease with $n$.
\end{proposition}

\Cref{prop:number-of-firms-welfare} shows that market size affects welfare through both competition and information. Consumer surplus and total surplus increase with $n$ because adding firms raises aggregate output. Producer surplus may move in either direction because stronger Cournot competition compresses rents, while additional firms may improve information aggregation under platform participation or change the information-dependent response term under no sharing.

\section{Proofs of Additional Results}\label{app:d:proofs-of-additional-results}

\subsection{\texorpdfstring{\Cref{prop:demand-uncertainty-welfare}}{propdemanduncertaintywelfare}: Effect of demand uncertainty}
\label{app:proof:prop:demand-uncertainty-welfare}

\begin{proof}
We first consider the full-participation regime. Let
\[
K_n
=
\frac{n}{\tau+m}
+
\frac{1}{\tau_p}.
\]
Then
\[
\sigma_S^2
=
u-
\left(
\frac{1}{u}+K_n
\right)^{-1}
=
\frac{u^2K_n}{1+uK_n}.
\]
Differentiating with respect to \(u\), we obtain
\[
\frac{\partial \sigma_S^2}{\partial u}
=
\frac{uK_n(2+uK_n)}{(1+uK_n)^2}
>0.
\]
Since expected consumer surplus, producer surplus, and total surplus under full
platform participation are all increasing in \(\sigma_S^2\), it follows that
\[
\frac{\partial \mathbb{E}[CS^S]}{\partial u}>0,
\qquad
\frac{\partial \mathbb{E}[PS^S]}{\partial u}>0,
\qquad
\frac{\partial \mathbb{E}[TS^S]}{\partial u}>0.
\]

We now consider the no-sharing regime. Recall that
\[
\sigma_P^2=\frac{u^2}{u+\tau},
\qquad
\rho=\frac{u}{u+\tau},
\qquad
\beta=\frac{1}{2+(n-1)\rho}.
\]
Equivalently,
\[
\beta
=
\frac{u+\tau}{(n+1)u+2\tau}.
\]

For consumer surplus under no information sharing, the information-dependent component is
\[
\beta^2\sigma_P^2[1+(n-1)\rho]
=
\frac{u^2(nu+\tau)}
{\left((n+1)u+2\tau\right)^2}.
\]
Differentiating this expression with respect to \(u\), we obtain
\[
\frac{\partial}{\partial u}
\left[
\frac{u^2(nu+\tau)}
{\left((n+1)u+2\tau\right)^2}
\right]
=
\frac{
u\left[n(n+1)u^2+6n u\tau+4\tau^2\right]
}
{\left((n+1)u+2\tau\right)^3}
>0.
\]
Therefore,
\[
\frac{\partial \mathbb{E}[CS^{NS}]}{\partial u}>0.
\]

For producer surplus under no information sharing, the information-dependent component is
\[
\beta^2\sigma_P^2
=
\frac{u^2(u+\tau)}
{\left((n+1)u+2\tau\right)^2}.
\]
Differentiating gives
\[
\frac{\partial}{\partial u}
\left[
\frac{u^2(u+\tau)}
{\left((n+1)u+2\tau\right)^2}
\right]
=
\frac{
u\left[(n+1)u^2+6u\tau+4\tau^2\right]
}
{\left((n+1)u+2\tau\right)^3}
>0.
\]
Hence,
\[
\frac{\partial \mathbb{E}[PS^{NS}]}{\partial u}>0.
\]

Finally, for total surplus under no information sharing, the information-dependent
component is
\[
\beta^2\sigma_P^2[3+(n-1)\rho]
=
\frac{u^2((n+2)u+3\tau)}
{\left((n+1)u+2\tau\right)^2}.
\]
Differentiating yields
\[
\frac{\partial}{\partial u}
\left[
\frac{u^2((n+2)u+3\tau)}
{\left((n+1)u+2\tau\right)^2}
\right]
=
\frac{
u\left[(n+1)(n+2)u^2+6(n+2)u\tau+12\tau^2\right]
}
{\left((n+1)u+2\tau\right)^3}
>0.
\]
Therefore,
\[
\frac{\partial \mathbb{E}[TS^{NS}]}{\partial u}>0.
\]
This completes the proof.
\end{proof}

\subsection{\texorpdfstring{\Cref{prop:private-signal-noise-welfare}}{propprivatesignalnoisewelfare}: Effect of private signal noise}
\label{app:proof:prop:private-signal-noise-welfare}

\begin{proof}
We first consider the full-participation regime. Recall that
\[
\sigma_S^2
=
u-
\left(
\frac{1}{u}
+
\frac{n}{\tau+m}
+
\frac{1}{\tau_p}
\right)^{-1}.
\]
Let
\[
\mathcal{P}_n
=
\frac{1}{u}
+
\frac{n}{\tau+m}
+
\frac{1}{\tau_p}.
\]
Then
\[
\sigma_S^2=u-\mathcal{P}_n^{-1}.
\]
Differentiating with respect to \(\tau\), we obtain
\[
\frac{\partial \sigma_S^2}{\partial \tau}
=
-\frac{n}{(\tau+m)^2\mathcal{P}_n^2}
<0.
\]
Since expected consumer surplus, producer surplus, and total surplus under full
platform participation are all increasing in \(\sigma_S^2\), it follows that
\[
\frac{\partial \mathbb{E}[CS^S]}{\partial \tau}<0,
\qquad
\frac{\partial \mathbb{E}[PS^S]}{\partial \tau}<0,
\qquad
\frac{\partial \mathbb{E}[TS^S]}{\partial \tau}<0.
\]

We now consider the no-sharing regime. Recall that
\[
\sigma_P^2=\frac{u^2}{u+\tau},
\qquad
\rho=\frac{u}{u+\tau},
\qquad
\beta=\frac{1}{2+(n-1)\rho}.
\]
Equivalently,
\[
\beta
=
\frac{u+\tau}{(n+1)u+2\tau}.
\]

For consumer surplus under no information sharing, the information-dependent component is
\[
\beta^2\sigma_P^2[1+(n-1)\rho].
\]
Using the expressions above, we can write
\[
\beta^2\sigma_P^2[1+(n-1)\rho]
=
\frac{u^2(nu+\tau)}
{\left((n+1)u+2\tau\right)^2}.
\]
Differentiating this expression with respect to \(\tau\), we obtain
\[
\frac{\partial}{\partial \tau}
\left[
\frac{u^2(nu+\tau)}
{\left((n+1)u+2\tau\right)^2}
\right]
=
-\frac{
u^2\left[(3n-1)u+2\tau\right]
}
{\left((n+1)u+2\tau\right)^3}
<0.
\]
Therefore,
\[
\frac{\partial \mathbb{E}[CS^{NS}]}{\partial \tau}<0.
\]

For producer surplus under no information sharing,
\[
\mathbb{E}[PS^{NS}]
=
\frac{na^2}{(n+1)^2}
+
n\beta^2\sigma_P^2.
\]
Using
\[
\beta=\frac{u+\tau}{(n+1)u+2\tau}
\qquad\text{and}\qquad
\sigma_P^2=\frac{u^2}{u+\tau},
\]
we can write the information-dependent component as
\[
n\beta^2\sigma_P^2
=
n
\frac{u^2(u+\tau)}
{\left((n+1)u+2\tau\right)^2}.
\]
Differentiating with respect to \(\tau\), we obtain
\[
\frac{\partial \mathbb{E}[PS^{NS}]}{\partial \tau}
=
n u^2
\frac{(n-3)u-2\tau}
{\left((n+1)u+2\tau\right)^3}.
\]
Hence,
\[
\frac{\partial \mathbb{E}[PS^{NS}]}{\partial \tau}\ge 0
\quad
\Longleftrightarrow
\quad
\tau\le \frac{(n-3)u}{2},
\]
and
\[
\frac{\partial \mathbb{E}[PS^{NS}]}{\partial \tau}\le 0
\quad
\Longleftrightarrow
\quad
\tau\ge \frac{(n-3)u}{2}.
\]
If \(n\le 3\), then \((n-3)u/2\le 0\). Therefore, for all economically relevant
\(\tau>0\),
\[
\frac{\partial \mathbb{E}[PS^{NS}]}{\partial \tau}<0.
\]

Finally, for total surplus under no information sharing, the information-dependent
component is
\[
\beta^2\sigma_P^2[3+(n-1)\rho].
\]
Using the expressions above, we can write
\[
\beta^2\sigma_P^2[3+(n-1)\rho]
=
\frac{u^2((n+2)u+3\tau)}
{\left((n+1)u+2\tau\right)^2}.
\]
Differentiating with respect to \(\tau\), we obtain
\[
\frac{\partial}{\partial \tau}
\left[
\frac{u^2((n+2)u+3\tau)}
{\left((n+1)u+2\tau\right)^2}
\right]
=
-\frac{
u^2\left[(n^2+4n-1)u+6\tau\right]
}
{\left((n+1)u+2\tau\right)^3}
<0.
\]
Therefore,
\[
\frac{\partial \mathbb{E}[TS^{NS}]}{\partial \tau}<0.
\]
This completes the proof.
\end{proof}

\subsection{\texorpdfstring{\Cref{prop:number-of-firms-welfare}}{propnumberoffirmswelfare}: Effect of the number of firms}
\label{app:proof:prop:number-of-firms-welfare}

\begin{proof}
Treating \(n\) as continuous for comparative-static intuition, recall that under full
platform participation,
\[
\sigma_S^2
=
u-
\left(
\frac{1}{u}
+
\frac{n}{\tau+m}
+
\frac{1}{\tau_p}
\right)^{-1}.
\]
Let
\[
\mathcal{P}_n
=
\frac{1}{u}
+
\frac{n}{\tau+m}
+
\frac{1}{\tau_p}.
\]
Then
\[
\sigma_S^2=u-\mathcal{P}_n^{-1}.
\]
Differentiating with respect to \(n\), we obtain
\[
\frac{\partial \sigma_S^2}{\partial n}
=
\frac{1}{(\tau+m)\mathcal{P}_n^2}
>0.
\]
Thus, increasing \(n\) improves the informativeness of the platform posterior.

We first consider the full-participation regime. Expected consumer surplus is
\[
\mathbb{E}[CS^S]
=
\frac{n^2}{2(n+1)^2}(a^2+\sigma_S^2).
\]
Differentiating with respect to \(n\), we obtain
\[
\frac{\partial \mathbb{E}[CS^S]}{\partial n}
=
\frac{n}{(n+1)^3}(a^2+\sigma_S^2)
+
\frac{n^2}{2(n+1)^2}
\frac{\partial \sigma_S^2}{\partial n}
>0.
\]
Hence, expected consumer surplus under full platform participation is strictly increasing
in \(n\).

Expected total surplus under full platform participation is
\[
\mathbb{E}[TS^S]
=
\frac{n(n+2)}{2(n+1)^2}(a^2+\sigma_S^2).
\]
Since
\[
\frac{\partial}{\partial n}
\left[
\frac{n(n+2)}{2(n+1)^2}
\right]
=
\frac{1}{(n+1)^3}
>0,
\]
we have
\[
\frac{\partial \mathbb{E}[TS^S]}{\partial n}
=
\frac{1}{(n+1)^3}(a^2+\sigma_S^2)
+
\frac{n(n+2)}{2(n+1)^2}
\frac{\partial \sigma_S^2}{\partial n}
>0.
\]
Thus, expected total surplus under full platform participation is strictly increasing in \(n\).

Expected producer surplus under full platform participation is
\[
\mathbb{E}[PS^S]
=
\frac{n}{(n+1)^2}(a^2+\sigma_S^2).
\]
The coefficient
\[
\frac{n}{(n+1)^2}
\]
satisfies
\[
\frac{\partial}{\partial n}
\left[
\frac{n}{(n+1)^2}
\right]
=
\frac{1-n}{(n+1)^3},
\]
which is negative for \(n>1\). At the same time,
\[
\frac{\partial \sigma_S^2}{\partial n}>0.
\]
Therefore, the effect of \(n\) on \(\mathbb{E}[PS^S]\) depends on the relative strength of
these two forces. More explicitly,
\[
\frac{\partial \mathbb{E}[PS^S]}{\partial n}
=
\frac{1-n}{(n+1)^3}(a^2+\sigma_S^2)
+
\frac{n}{(n+1)^2}
\frac{\partial \sigma_S^2}{\partial n}.
\]
Thus, expected producer surplus under full platform participation increases in \(n\) when
the second term dominates the first, and decreases when the first term dominates.

We now consider the no-sharing regime. Expected consumer surplus is
\[
\mathbb{E}[CS^{NS}]
=
\frac{n^2a^2}{2(n+1)^2}
+
\frac{1}{2}n\beta^2\sigma_P^2[1+(n-1)\rho],
\]
where
\[
\beta=\frac{1}{2+(n-1)\rho}.
\]
Since \(\sigma_P^2\) and \(\rho\) do not depend on \(n\), differentiating gives
\[
\frac{\partial \mathbb{E}[CS^{NS}]}{\partial n}
=
\frac{na^2}{(n+1)^3}
+
\frac{\sigma_P^2}{2}
\frac{
n\rho(3-\rho)+(1-\rho)(2-\rho)
}
{\left[2+(n-1)\rho\right]^3}.
\]
Because \(\rho\in(0,1)\), both terms are positive. Hence,
\[
\frac{\partial \mathbb{E}[CS^{NS}]}{\partial n}>0.
\]

Expected total surplus under no information sharing is
\[
\mathbb{E}[TS^{NS}]
=
\frac{n(n+2)a^2}{2(n+1)^2}
+
\frac{n\beta^2\sigma_P^2}{2}[3+(n-1)\rho].
\]
Differentiating gives
\[
\frac{\partial \mathbb{E}[TS^{NS}]}{\partial n}
=
\frac{a^2}{(n+1)^3}
+
\frac{\sigma_P^2}{2}
\frac{
n\rho(1-\rho)+(2-\rho)(3-\rho)
}
{\left[2+(n-1)\rho\right]^3}.
\]
Again, since \(\rho\in(0,1)\), both terms are positive. Therefore,
\[
\frac{\partial \mathbb{E}[TS^{NS}]}{\partial n}>0.
\]

Finally, expected producer surplus under no information sharing is
\[
\mathbb{E}[PS^{NS}]
=
\frac{na^2}{(n+1)^2}
+
n\beta^2\sigma_P^2.
\]
The deterministic Cournot component satisfies
\[
\frac{\partial}{\partial n}
\left[
\frac{na^2}{(n+1)^2}
\right]
=
\frac{(1-n)a^2}{(n+1)^3},
\]
which is negative for \(n>1\).

The information-dependent component is
\[
n\beta^2\sigma_P^2
=
\frac{n\sigma_P^2}
{\left[2+(n-1)\rho\right]^2}.
\]
Differentiating with respect to \(n\), we obtain
\[
\frac{\partial}{\partial n}
\left[
n\beta^2\sigma_P^2
\right]
=
\sigma_P^2
\frac{2-(n+1)\rho}
{\left[2+(n-1)\rho\right]^3}.
\]
Therefore, the information-dependent component is increasing in \(n\) if
\[
\rho<\frac{2}{n+1},
\]
decreasing in \(n\) if
\[
\rho>\frac{2}{n+1},
\]
and has zero marginal effect if
\[
\rho=\frac{2}{n+1}.
\]

Combining the deterministic Cournot component and the information-dependent component,
the overall effect of \(n\) on expected producer surplus under no sharing depends on
whether the information-dependent effect is large enough to offset the decline in the
deterministic Cournot component. This completes the proof.
\end{proof}

\subsection{\texorpdfstring{\Cref{prop:privacy-restores-sharing-hetero}}{propprivacyrestoressharinghetero}: Privacy restores voluntary sharing under heterogeneous signals}\label{app:proof:prop:privacy-restores-sharing-hetero}

\begin{proof}
We prove the claim by analyzing each firm’s unilateral incentive and then taking a common
high-noise limit. The argument follows the same structure as in the homogeneous case.

Firm $i$ observes a private signal $s_i=\theta+\varepsilon_i$ with $\varepsilon_i\sim\mathcal N(0,\tau_i)$,
and the platform observes an external signal $s_p=\theta+\varepsilon_p$ with $\varepsilon_p\sim\mathcal N(0,\tau_p)$.
If a set $J\subseteq\{1,\dots,n\}$ participates, the platform receives privatized reports
\[
y_i := s_i+\xi_i=\theta+(\varepsilon_i+\xi_i),\qquad \xi_i\sim\mathcal N(0,m_i),
\]
so that $\varepsilon_i+\xi_i\sim \mathcal N(0,\tau_i+m_i)$.
The platform releases the posterior mean
\[
t_J(\mathbf m):=\mathbb E\!\left[\theta \,\middle|\, \{y_i\}_{i\in J},\, s_p\right],
\qquad \mathbf m:=(m_1,\dots,m_n).
\]
Under joint Gaussianity, $t_J(\mathbf m)$ is linear and admits the precision-weighted form
\begin{equation}
\label{eq:tJ-hetero-form}
t_J(\mathbf m)
=
V_J(\mathbf m)\left(
\sum_{i\in J}\frac{1}{\tau_i+m_i}y_i+\frac{1}{\tau_p}s_p
\right),
\qquad
V_J(\mathbf m):=
\left(\frac{1}{u}+\sum_{i\in J}\frac{1}{\tau_i+m_i}+\frac{1}{\tau_p}\right)^{-1}.
\end{equation}
Moreover, since $t_J(\mathbf m)=\mathbb E[\theta\mid \{y_i\}_{i\in J},s_p]$, the law of total variance yields
\begin{equation}
\label{eq:var-tJ-hetero}
\Var\!\big(t_J(\mathbf m)\big)=u-V_J(\mathbf m),
\end{equation}
where the posterior variance $V_J(\mathbf m)$ is deterministic in the Gaussian model.

Under full participation $J=[n]$,
\[
t_{[n]}(\mathbf m)=V_{n}(\mathbf m)\left(\sum_{i=1}^n\frac{1}{\tau_i+m_i}y_i+\frac{1}{\tau_p}s_p\right),
\]
and under unilateral non-participation by firm $i$, the platform aggregates only $[n]\setminus\{i\}$:
\[
t_{-i}(\mathbf m):=t_{[n]\setminus\{i\}}(\mathbf m)
=
V_{-i}(\mathbf m)\left(\sum_{j\neq i}\frac{1}{\tau_j+m_j}y_j+\frac{1}{\tau_p}s_p\right).
\]

\paragraph{Sharing equilibrium.}
Fix $\mathbf m\ge 0$ and suppose all firms participate. Because primitives are jointly Gaussian
and payoffs are quadratic under linear demand, there exists a (unique) Bayesian Nash equilibrium
in which each firm $i$ uses an affine strategy in its information $(s_i,t_{[n]}(\mathbf m))$:
\[
q_i^{S}(\mathbf m)=c_i(\mathbf m)+d_i(\mathbf m)\,s_i+e_i(\mathbf m)\,t_{[n]}(\mathbf m).
\]
The equilibrium coefficients solve a linear system obtained from the first-order conditions and
Gaussian projection identities of the form
$\mathbb E[\theta\mid s_i,t_{[n]}(\mathbf m)]$ and $\mathbb E[s_j\mid s_i,t_{[n]}(\mathbf m)]$.
All entries of this system are polynomials in the relevant covariance elements of
$(\theta,s_1,\dots,s_n,s_p,y_1,\dots,y_n)$, and hence depend continuously on $\mathbf m$
through the terms $(\tau_i+m_i)^{-1}$ and $V_J(\mathbf m)$ in \eqref{eq:tJ-hetero-form}.
Therefore, the equilibrium coefficients and the ex-ante sharing payoff
$\mathbb E[\pi_i^{S}(\mathbf m)]$ vary continuously with $\mathbf m$.

\paragraph{Unilateral deviation.}
Fix a firm $i$. Suppose firm $i$ does not participate while all other firms do.
Then firm $i$ observes only $s_i$ and uses an affine best response
\[
q_i^{NS}(\mathbf m)=\bar c_i(\mathbf m)+\bar d_i(\mathbf m)\,s_i,
\]
whereas each participating firm $j\neq i$ observes $(s_j,t_{-i}(\mathbf m))$ and uses an affine strategy
\[
q_j^{S}(\mathbf m)=\bar c_j(\mathbf m)+\bar d_j(\mathbf m)\,s_j+\bar e_j(\mathbf m)\,t_{-i}(\mathbf m).
\]
Then, these coefficients solve a linear system whose entries are covariance elements depending
continuously on $\mathbf m$. Hence the deviation payoff $\mathbb E[\pi_i^{NS}(\mathbf m)]$ is continuous in $\mathbf m$.

Define firm $i$'s unilateral participation gain under privacy noise,
\[
\Delta_i(\mathbf m):=\mathbb E[\pi_i^{S}(\mathbf m)]-\mathbb E[\pi_i^{NS}(\mathbf m)].
\]
By the preceding discussion, each $\Delta_i(\mathbf m)$ is continuous in $\mathbf m$ on $[0,\infty)^n$.

By assumption, under the non-private mechanism ($\mathbf m=\mathbf 0$) the platform signal is not precise enough
to sustain full participation; equivalently, the sharing profile fails to be incentive compatible, so
there exists at least one firm $i$ with $\Delta_i(\mathbf 0)<0$.

Next consider a common high-noise scaling path $\mathbf m(t):=t\,\mathbf 1$ with $t\ge 0$.
As $t\to\infty$, each privatized report precision satisfies $(\tau_i+t)^{-1}\to 0$.
From \eqref{eq:tJ-hetero-form},
\[
t_{[n]}(\mathbf m(t)) \xrightarrow[]{t\nearrow\infty} \mathbb E[\theta\mid s_p],
\qquad
t_{-i}(\mathbf m(t)) \xrightarrow[]{t\nearrow\infty} \mathbb [\theta\mid s_p],
\]
in $L^2$ (hence in distribution), because the contribution of the privatized reports vanishes and the posterior is driven only by $(\theta,s_p)$.
Since $\tau_p<\infty$, the public posterior $\mathbb E[\theta\mid s_p]$ is non-degenerate and payoff-relevant.

In the limit environment $t\to\infty$, if firm $i$ \emph{participates} when all others participate, it observes
$(s_i,\mathbb E[\theta\mid s_p])$; if it \emph{does not} participate while others do, it observes only $s_i$ while its rivals condition on
$(s_j,\mathbb E[\theta\mid s_p])$.
Thus, participation strictly enlarges firm $i$'s information set by adding the informative signal $\mathbb E[\theta\mid s_p]$.
In linear-quadratic Cournot, the value of (Blackwell) more informative signals is strictly positive whenever the additional signal is non-degenerate,
so firm $i$'s optimized ex-ante profit is strictly higher under participation in this limit. Hence,
\begin{equation}
\label{eq:Delta-hetero-infty-pos}
\lim_{t\to\infty}\Delta_i(\mathbf m(t))>0
\qquad \text{for each } i\in\{1,\dots,n\}.
\end{equation}

By continuity of $\Delta_i(\mathbf m(t))$ in $t$, for each firm $i$ there exists a finite threshold $t_i^\ast>0$
such that $\Delta_i(\mathbf m(t))>0$ for all $t\ge t_i^\ast$.
Let
\[
t^\ast:=\max_{i\in\{1,\dots,n\}} t_i^\ast,
\qquad
\mathbf m^\ast := t^\ast \mathbf 1.
\]
Then for all $\mathbf m\ge \mathbf m^\ast$ (componentwise), in particular for $\mathbf m=\mathbf m^\ast$ and any larger vector,
we have $\Delta_i(\mathbf m)\ge \Delta_i(\mathbf m^\ast)>0$ for every firm $i$.
Therefore, under the privacy-preserving mechanism with $\mathbf m\ge \mathbf m^\ast$, each firm strictly prefers to participate
when all other firms participate. This establishes that sharing is sustained as an equilibrium under sufficiently large privacy noise,
completing the proof.
\end{proof}

\subsection{\texorpdfstring{\Cref{prop:sequential-heterogeneous}}{propsequentialheterogeneous}: Privacy-induced participation under sequential entry}\label{app:proof:prop:sequential-heterogeneous}
\begin{proof}
Fix firm $i$ and suppose all firms $j\neq i$ have already joined the platform and share using a common privacy noise level $m$.
Firm $i$ has signal noise variance $\tau_i$ and, if it participates, the platform applies a firm-specific privacy noise level $m_i$ to its report.
Throughout, inverse demand is $p=a+\theta-\sum_{k=1}^n q_k$, the prior is $\theta\sim\mathcal N(0,u)$, and the platform signal is $s_p=\theta+\eta$ with $\eta\sim\mathcal N(0,\tau_p)$.

For each participating firm $k$, the platform observes the privatized report
\[
y_k:=s_k+\xi_k=\theta+(\varepsilon_k+\xi_k),
\qquad
\varepsilon_k\sim\mathcal N(0,\tau_k),\ \xi_k\sim\mathcal N(0,m_k),
\]
with all shocks independent across firms and from $\theta$.
For firms $j\neq i$, $(\tau_j,m_j)=(\tau,m)$, while for firm $i$, $(\tau_i,m_i)=(\tau_i,m_i)$.

If firm $i$ participates, the platform releases the posterior mean
\[
t_{[n]}(m_i):=\mathbb E\!\left[\theta \,\middle|\, \{y_j\}_{j\neq i},\, y_i,\, s_p\right].
\]
Under joint Gaussianity, $t_{[n]}(m_i)$ is linear and admits the precision-weighted form
\begin{align}
\label{eq:t-all-hetero}
t_{[n]}(m_i)
& =
V_n(m_i)\left(
\frac{1}{\tau+m}\sum_{j\neq i} y_j
+\frac{1}{\tau_i+m_i}y_i
+\frac{1}{\tau_p}s_p
\right), \\
V_n(m_i) & :=
\left(\frac{1}{u}+\frac{n-1}{\tau+m}+\frac{1}{\tau_i+m_i}+\frac{1}{\tau_p}\right)^{-1}. \nonumber
\end{align}
If firm $i$ does not participate, the platform aggregates only firms $j\neq i$ and releases

\begin{align*}
t_{-i}(m)&:=\mathbb E\!\left[\theta \,\middle|\, \{y_j\}_{j\neq i},\, s_p\right]
=
V_{n-1}(m)\left(
\frac{1}{\tau+m}\sum_{j\neq i} y_j
+\frac{1}{\tau_p}s_p
\right), \\
\qquad
V_{n-1}(m)&:=
\left(\frac{1}{u}+\frac{n-1}{\tau+m}+\frac{1}{\tau_p}\right)^{-1}.
\end{align*}
Moreover, since each released statistic is a conditional expectation, the law of total variance yields
\[
\Var\!\big(t_{[n]}(m_i)\big)=u-V_n(m_i),
\qquad
\Var\!\big(t_{-i}(m)\big)=u-V_{n-1}(m),
\]
and these are deterministic in the Gaussian model.

For later use, write $t_{[n]}(m_i)=\alpha_O\sum_{j\neq i}y_j+\alpha_I y_i+\beta s_p$, where
\[
\alpha_O=\frac{V_n(m_i)}{\tau+m},\qquad
\alpha_I=\frac{V_n(m_i)}{\tau_i+m_i},\qquad
\beta=\frac{V_n(m_i)}{\tau_p},
\]
and write $t_{-i}(m)=\alpha_{-i}\sum_{j\neq i}y_j+\beta_{-i}s_p$, where
\[
\alpha_{-i}=\frac{V_{n-1}(m)}{\tau+m},\qquad
\beta_{-i}=\frac{V_{n-1}(m)}{\tau_p}.
\]

Assume firm $i$ participates, so all firms observe a private signal and the released statistic $t_{[n]}(m_i)$.
We look for a linear Bayesian equilibrium with two types of strategies:
\[
q_i^{J}=c_I+d_I s_i+e_I\,t_{[n]}(m_i),
\qquad
q_j^{J}=c_O+d_O s_j+e_O\,t_{[n]}(m_i)\quad (j\neq i).
\]
Let $t:=t_{[n]}(m_i)$.
For firm $i$, define $X_I:=(s_i,t)^{\mathsf T}$.
By Gaussian projection, for any $j\neq i$ and for $\theta$,
\[
\mathbb E[s_j\mid s_i,t]=\Sigma_{s_jX_I}\Sigma_{X_I X_I}^{-1}X_I,
\qquad
\mathbb E[\theta\mid s_i,t]=\Sigma_{\theta X_I}\Sigma_{X_I X_I}^{-1}X_I.
\]
The required covariance elements follow from $s_k=\theta+\varepsilon_k$ and \eqref{eq:t-all-hetero}:
\[
\Var(s_i)=u+\tau_i,\qquad \Var(t)=u-V_n(m_i),\qquad \Cov(\theta,t)=\Var(t),
\]
\begin{align*}
    \Cov(s_i,t) & =\alpha_I(u+\tau_i)+\alpha_O(n-1)u+\beta u=:C_{it}, \\
\Cov(s_j,t) & =\alpha_O\big((n-1)u+\tau\big)+\alpha_I u+\beta u=:C_{ot}\quad (j\neq i),
\end{align*}
\[
\Cov(s_i,s_j)=u,\qquad \Cov(\theta,s_i)=u.
\]
Hence
\[
\Sigma_{X_I X_I}=
\begin{bmatrix}
u+\tau_i & C_{it}\\
C_{it} & u-V_n(m_i)
\end{bmatrix},
\qquad
\Sigma_{\theta X_I}=\begin{bmatrix} u & u-V_n(m_i)\end{bmatrix},
\qquad
\Sigma_{s_jX_I}=\begin{bmatrix} u & C_{ot}\end{bmatrix}.
\]
Write the projection coefficients as
\[
\mathbb E[\theta\mid s_i,t]=\kappa_S s_i+\kappa_T t,
\qquad
\mathbb E[s_j\mid s_i,t]=\lambda_S s_i+\lambda_T t,
\]
where $(\kappa_S,\kappa_T)=\Sigma_{\theta X_I}\Sigma_{X_I X_I}^{-1}$ and
$(\lambda_S,\lambda_T)=\Sigma_{s_jX_I}\Sigma_{X_I X_I}^{-1}$.

A representative other firm $j\neq i$ conditions on $X_O:=(s_j,t)^{\mathsf T}$, with covariance
\[
\Sigma_{X_O X_O}=
\begin{bmatrix}
u+\tau & C_{ot}\\
C_{ot} & u-V_n(m_i)
\end{bmatrix},
\]
and corresponding projection coefficients
\[
\mathbb E[\theta\mid s_j,t]=\kappa_S^{O}s_j+\kappa_T^{O}t,
\qquad
\mathbb E[s_i\mid s_j,t]=\phi_S s_j+\phi_T t,
\qquad
\mathbb E[s_k\mid s_j,t]=\chi_S s_j+\chi_T t\ (k\neq i,j),
\]
where $(\kappa_S^{O},\kappa_T^{O})=\Sigma_{\theta X_O}\Sigma_{X_O X_O}^{-1}$,
$(\phi_S,\phi_T)=\Sigma_{s_iX_O}\Sigma_{X_O X_O}^{-1}$,
and $(\chi_S,\chi_T)=\Sigma_{s_kX_O}\Sigma_{X_O X_O}^{-1}$.

Firm best responses satisfy, for each $k$,
\[
q_k=\frac12\Big(a+\mathbb E[\theta\mid \mathcal I_k]-\mathbb E[Q_{-k}\mid \mathcal I_k]\Big),
\qquad
Q_{-k}:=\sum_{\ell\neq k} q_\ell,
\]
where $\mathcal I_i=\sigma(s_i,t)$ and $\mathcal I_j=\sigma(s_j,t)$ for $j\neq i$.
Substituting the linear strategies and the Gaussian projections above and matching coefficients on $(1,s_i,t)$ and $(1,s_j,t)$ yields a linear system in the six unknown coefficients $(c_I,d_I,e_I,c_O,d_O,e_O)$, which has a unique solution because the relevant covariance matrices are positive definite for $u>0$, $\tau>0$, $\tau_i>0$, $\tau_p\in(0,\infty)$, and $m,m_i\ge 0$.
In particular, matching constant terms gives
\[
c_I=c_O=\frac{a}{n+1}.
\]
The remaining coefficients are uniquely determined by coefficient matching on the signal and statistic components.

Firm $i$'s ex-ante expected profit under participation is
\[
\pi_i^{S}(m_i)=\mathbb E\!\left[(q_i^{J})^2\right]
=c_I^2+d_I^2\Var(s_i)+e_I^2\Var(t)+2d_I e_I\,\Cov(s_i,t),
\]
where $\Var(s_i)=u+\tau_i$, $\Var(t)=u-V_n(m_i)$, and $\Cov(s_i,t)=C_{it}$.

\paragraph{Non-participation outcome (firm $i$ does not participate).}
Assume firm $i$ does not participate, so it observes only $s_i$ and uses a linear strategy
\[
q_i^{N}=c_0+d_0 s_i.
\]
All other firms participate and observe $(s_j,t_{-i}(m))$, using
\[
q_j^{J}=c_1+d_1 s_j+e_1 t_{-i}(m)\quad (j\neq i).
\]
Let $t^-:=t_{-i}(m)$.
For firm $i$, $\mathbb E[\theta\mid s_i]=\frac{u}{u+\tau_i}s_i$.
Moreover, since $t^-$ is Gaussian and linear in $(\theta,\{y_j\}_{j\neq i},s_p)$,
\[
\mathbb E[t^-\mid s_i]=\frac{\Cov(t^-,s_i)}{\Var(s_i)}\,s_i.
\]
Using $t^-=\alpha_{-i}\sum_{j\neq i}y_j+\beta_{-i}s_p$ and $\Cov(s_i,y_j)=u$ for $j\neq i$,
\[
\Cov(s_i,t^-)=\alpha_{-i}(n-1)u+\beta_{-i}u,
\qquad
\Var(s_i)=u+\tau_i,
\]
so
\[
\mathbb E[t^-\mid s_i]=\frac{\alpha_{-i}(n-1)u+\beta_{-i}u}{u+\tau_i}\,s_i.
\]
Substituting these projections into firm $i$'s best response and matching coefficients on $(1,s_i)$ yields linear equations relating $(c_0,d_0)$ to $(c_1,d_1,e_1)$.

For a representative participating firm $j\neq i$, define $X:=(s_j,t^-)^{\mathsf T}$.
Gaussian projection gives
\[
\mathbb E[s_i\mid s_j,t^-]=\phi_S s_j+\phi_T t^-,
\qquad
\mathbb E[s_k\mid s_j,t^-]=\chi_S s_j+\chi_T t^-\ (k\neq i,j),
\qquad
\mathbb E[\theta\mid s_j,t^-]=\kappa_S^- s_j+\kappa_T^- t^-,
\]
with coefficients computed from the explicit covariance matrix of $(s_j,t^-)$.
Plugging the linear strategies into firm $j$'s first-order condition and matching coefficients on $(1,s_j,t^-)$ yields a $3\times 3$ linear system in $(d_0,d_1,e_1)$ (and $c_0=c_1=a/(n+1)$).

Firm $i$'s ex-ante expected profit under non-participation is
\[
\pi_i^{NS}
=
\mathbb E\!\left[(q_i^{NS})^2\right]
=
c_0^2+d_0^2\Var(s_i)
=
\frac{a^2}{(n+1)^2}+d_0^2(u+\tau_i).
\]

Define the unilateral participation gain as
\[
\Delta_i(m_i):=\pi_i^{S}(m_i)-\pi_i^{NS}(m_i).
\]
All equilibrium coefficients above are obtained by solving linear systems whose entries are polynomials in the covariance elements of $(\theta,s_p,s_1,\dots,s_n,y_1,\dots,y_n)$.
These covariance elements depend continuously on $m_i$ through $V_n(m_i)$ and $\alpha_i(m_i)=V_n(m_i)/(\tau_i+m_i)$.
Hence $\pi_i^{S}(m_i)$ and $\Delta_i(m_i)$ are continuous in $m_i$ on $[0,\infty)$.

By assumption $\tau_p^*(n,u,\tau_i,\tau)<\tau_p<\infty$, the platform signal is not sufficiently informative to induce firm $i$'s participation when no privacy noise is applied to firm $i$'s report.
Equivalently, given the other firms' participation, firm $i$ strictly prefers not to participate at $m_i=0$, so $\Delta_i(0)<0$.

Consider $m_i\to\infty$.
In \eqref{eq:t-all-hetero}, the precision on firm $i$'s privatized report satisfies $1/(\tau_i+m_i)\to 0$, hence $\alpha_I\to 0$ and $V_n(m_i)\to V_{n-1}(m)$.
Therefore $t_{[n]}(m_i)\to t_{-i}(m)$ in $L^2$.
In the limit, participation by firm $i$ no longer reveals any payoff-relevant information about its private signal to competitors, while it grants firm $i$ access to the non-degenerate statistic $t_{-i}(m)$ constructed from existing participants and $s_p$.
Under non-participation, firm $i$ does not observe any platform statistic and remains restricted to $s_i$.
Since $\tau_p<\infty$ implies that $t_{-i}(m)$ is informative about $\theta$, access to $t_{-i}(m)$ strictly expands firm $i$'s information and hence strictly increases its optimized ex-ante profit, implying $\lim_{m_i\to\infty}\Delta_i(m_i)>0$.

By continuity of $\Delta_i(\cdot)$ and the inequalities $\Delta_i(0)<0$ and $\lim_{m_i\to\infty}\Delta_i(m_i)>0$, the intermediate value theorem implies that there exists $\bar m_i>0$ such that $\Delta_i(\bar m_i)=0$.
Thus, a finite firm-specific privacy noise level can render firm $i$ (weakly) willing to participate given that all other firms participate.

\end{proof}

\end{APPENDICES}



\end{document}